\documentclass[a4paper,11pt]{article}

\usepackage[british]{babel}
\usepackage{mathtools}
\usepackage{charter}%charter font
\usepackage[T1]{fontenc}
\usepackage[utf8]{inputenc}
\usepackage{physics}
\usepackage{color}
\usepackage{amsmath, amssymb, amsfonts}
\usepackage{graphicx}
\usepackage{titlesec}
\usepackage{relsize}
\usepackage{bbold} %This is seems to act weird.
\graphicspath{{Images/}}
\numberwithin{equation}{section} %Labels equations by section.
\usepackage{amsthm}
\usepackage{dsfont}
\usepackage{mwe}
\usepackage[normalem]{ulem} % to strikeout comments, normalem needed to bring emph back to normal
\usepackage{epsfig}
\usepackage[colorlinks=true,linkcolor=blue,citecolor=citegreen]{hyperref}
\usepackage[affil-it]{authblk}
\usepackage{thmtools} %correct ref in \cref in order to obtain "Lemma"
\usepackage[nameinlink,capitalise]{cleveref}
\usepackage{pdfpages}
\usepackage[font=small,labelfont=bf]{caption}
\usepackage[onehalfspacing]{setspace}
\usepackage{caption}
\usepackage{subcaption}

\usepackage{cleveref}
\crefname{lemma}{lemma}{lemmas}
\crefname{proposition}{proposition}{propositions}
\crefname{definition}{definition}{definitions}
\crefname{theorem}{theorem}{theorems}
\crefname{conjecture}{conjecture}{conjectures}
\crefname{corollary}{corollary}{corollaries}
\crefname{example}{example}{examples}
\crefname{section}{section}{sections}
\crefname{appendix}{appendix}{appendices}
\crefname{figure}{fig.}{figs.}
\crefname{equation}{eq.}{eqs.}
\crefname{table}{table}{tables}
\crefname{item}{property}{properties}
\crefname{remark}{remark}{remarks}
\crefname{problem}{}{}
\crefname{lemma}{lemma}{lemmas}
\crefname{proposition}{proposition}{propositions}
\crefname{definition}{definition}{definitions}
\crefname{theorem}{theorem}{theorems}
\crefname{conjecture}{conjecture}{conjectures}
\crefname{corollary}{corollary}{corollaries}
\crefname{example}{example}{examples}
\crefname{section}{section}{sections}
\crefname{appendix}{appendix}{appendices}
\crefname{figure}{fig.}{figs.}
\crefname{equation}{eq.}{eqs.}
\crefname{table}{table}{tables}
\crefname{item}{property}{properties}
\crefname{remark}{remark}{remarks}
\crefname{problem}{}{}

\newtheorem{theorem}{Theorem}
\newtheorem{definition}[theorem]{Definition}
\newtheorem{corollary}[theorem]{Corollary}
\newtheorem{proposition}[theorem]{Proposition}
\newtheorem{lemma}[theorem]{Lemma}

\newtheorem{remark}[theorem]{Remark}

\usepackage{changepage} %to change margins for part of the text with \begin{adjustwidth}
\usepackage{tikz}
\usepackage{tikzscale}
\usetikzlibrary{calc}
\usetikzlibrary{decorations.pathreplacing}
\usetikzlibrary{decorations}
\usetikzlibrary{positioning}
\usetikzlibrary{plotmarks}
\usetikzlibrary{backgrounds}
\usepackage{pgfplots}
\pgfplotsset{compat=newest}
\usepackage{xxcolor}
\definecolor{structure}{rgb}{0.23,0.4,0.7}

\usepackage[backend=bibtex,style=alphabetic,doi=false,isbn=false,url=false,maxbibnames=20,sorting=nty,sortcites=false]{biblatex}
\bibliography{RGlib}

\newcommand{\mysymbol}[1]{{\mbox{\raisebox{-0.3em}{\epsfysize=1.2em\epsfbox{#1}}}}}

\newcommand{\leftend}{\mysymbol{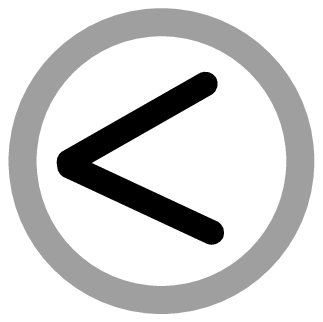}}
\newcommand{\rightend}{\mysymbol{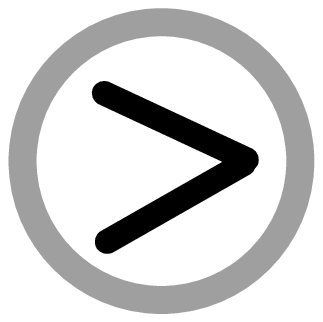}}

\def\Z{{\mathds{Z}}}
\def\R{{\mathds{R}}}
\def\N{{\mathds{N}}}
\def\C{\mathds{C}}
\def\1{{\mathds{1}}}
\def\B{{\cal{B}}}
\def\T{{\cal{T}}}
\def\Q{{\cal{Q}}}
\def\F{{\cal{F}}}
\def\D{{\cal{D}}}
\def\r{{R}}
\def\cB{{\cal{B}}}
\def\HS{{\cal{H}}}

\def\calS{{\cal{S}}}

\def\ox{{\otimes}}

\DeclareMathOperator\spann{span}

\DeclareMathOperator\spec{spec}

\DeclareMathOperator{\poly}{poly}

\newcommand{\frk}[1]{\mathfrak{#1}}
\newcommand{\Rk}{R^{(k)}}
\newcommand{\sRk}{\mathcal{R}^{(k)}}
\newcommand{\sR}{\mathcal{R}}

\newtheorem*{lemma*}{Lemma}

\crefname{section}{Section}{Sections}
\crefname{subsection}{Subsection}{Subsections}
\crefname{theorem}{Theorem}{Theorems}
\crefname{corollary}{Corollary}{Corollaries}
\crefname{lemma}{Lemma}{Lemmas}
\crefname{appendix}{Appendix}{Appendices}
\crefname{definition}{Definition}{Definitions}

\newcommand{\twocellsvert}[2]{\begin{array}{|@{}c@{}|} \hline  #1 \\ \hline  #2 \\ \hline \end{array}}
\newcommand{\fourcells}[4]{\begin{array}{|@{}c@{}|@{}c@{}|} \hline  #1 & #3 \\ \hline #2 & #4 \\ \hline   \end{array}}

\definecolor{mylightgrey}{RGB}{215,215,215}
\definecolor{mydarkgrey}{RGB}{169,169,169}
\definecolor{myorange}{RGB}{255,172,0}
\definecolor{citegreen}{RGB}{0,165,0}

\begin{document}

\title{Uncomputably Complex Renormalisation Group Flows}
\author{James D. Watson\footnote{05watson.j@gmail.com}}
\author{Emilio Onorati\footnote{e.onorati@ucl.ac.uk}}
\author{Toby S. Cubitt\footnote{t.cubitt@ucl.ac.uk}}
\affil{Department of Computer Science, University College London, UK}

\date{}

\maketitle
%\begin{adjustwidth}{-.5cm}{-.5cm}
\begin{abstract}
  Renormalisation group (RG) methods provide one of the most important techniques for analysing the physics of many-body systems, both analytically and numerically~\cite{Wilson71,Wilson74}.
  By iterating an RG map, which ``course-grains'' the description of a many-body system and generates a flow in the parameter space,
  physical properties of interest can often be extracted even for complex many-body models.
  RG analysis also provides an explanation of physical phenomena such as universality.
  Many systems exhibit simple RG flows, but
  more complicated --- even chaotic --- behaviour is also known~\cite{McKay_Berker_Kirkpatrick_1982, Svrakic_1982, Derrida_Eckmann_Erzan_1999, Damgaard_Thorleifsson_1991, Morozov_Niemi_2003}.
  Nonetheless, the general structure of such RG flows can still be analysed, elucidating the physics of the system, even if specific trajectories may be highly sensitive to the initial point.
  In contrast, recent work~\cite{Cubitt_PG_Wolf_Undecidability,Cubitt_PG_Wolf_Nature, Bausch_Cubitt_Watson} has shown that important physical properties of quantum many-body systems, such as its spectral gap and phase diagram, can be uncomputable, and thus impossible to determine even in principle.

  In this work, we show that such undecidable systems exhibit a novel type of RG flow, revealing a qualitatively different and more extreme form of unpredictability than chaotic RG flows.
  In contrast to chaotic RG flows in which initially close points can diverge exponentially according to some Lyapunov exponent, trajectories under these novel uncomputable RG flows can remain arbitrarily close together for an uncomputably large number of iterations, before abruptly diverging to different fixed points that are in separate phases.
  The structure of such uncomputable RG flows --- e.g.\ the basins of attraction of its fixed points --- is so complex that it cannot be computed or approximated, even in principle.
  To substantiate these claims, we give a mathematically rigorous construction of the block-renormalisation-group (BRG) map for the original undecidable many-body system that appeared in the literature~\cite{Cubitt_PG_Wolf_Undecidability,Cubitt_PG_Wolf_Nature}.
  We prove that each step of this RG map is efficiently computable, and that it converges to the correct fixed points, yet the resulting RG flow is uncomputable.
\end{abstract}
%\end{adjustwidth}

%\newpage
\tableofcontents

\section{Introduction}

Understanding collective properties and phases of many-body systems from an underlying model of the interactions between their constituent parts remains one of the major research areas in physics, from high-energy physics to condensed matter.
Many powerful techniques have been developed to tackle this problem.
One of the most far-reaching was the development by Wilson~\cite{Wilson71,Wilson74} of \emph{renormalisation group} (RG) techniques, building on early work by others~\cite{SBPM53,GellMannLow}.
At a conceptual level, an RG analysis involves constructing an RG map that takes as input a description of the many-body system (e.g.\ a Hamiltonian, or an action, or a partition function, etc.), and outputs a description of a new many-body system (a new Hamiltonian, or action, or partition function, etc.), that can be understood as a ``coarse-grained'' version of the original system, in such a way that physical properties of interest are preserved but irrelevant details are discarded.

For example, the RG map may ``integrate out'' the microscopic details of the interactions between the constituent particles described by the full Hamiltonian of the system. This procedure produces a coarse-grained Hamiltonian that still retains the same physics at larger length-scales~\cite{Kadanoff66}.
By repeatedly applying the RG map, the original Hamiltonian is transformed into successively simpler Hamiltonians, where the physics may be far easier to extract.
The RG map therefore describes a dynamical map on Hamiltonians, and consecutive applications of this map generates a ``flow'' in the space of Hamiltonians.
Often, the form of the Hamiltonian is preserved, and the RG flow can be characterised as a trajectory for its parameters.

The development of RG methods not only allowed sophisticated theoretical and numerical analysis of a broad range of many-body systems. It also explained phenomena such as \emph{universality}, whereby many physical systems, apparently very different, exhibit the same macroscopic behaviour, even at a quantitative level.
This is explained by the fact that these systems ``flow'' to the same fixed point under the RG dynamics.

For many condensed matter systems -- even complex strongly interacting ones -- the RG dynamics are relatively simple, exhibiting a finite number of fixed points to which the RG flow converges. Hamiltonians that converge to the same fixed point correspond to the same phase, so that the basins of attraction of the fixed points map out the phase diagram of the system.
However, more complicated RG dynamics is also possible, including chaotic RG flows with highly complex structure~\cite{McKay_Berker_Kirkpatrick_1982, Svrakic_1982, Derrida_Eckmann_Erzan_1999, Damgaard_Thorleifsson_1991, Morozov_Niemi_2003}.
Nonetheless, as with chaotic dynamics more generally, the structure and attractors of such chaotic RG flows can still be analysed, even if specific trajectories of the dynamics may be highly sensitive to the precise starting point.
This structure elucidates much of the physics of the system \cite{Grassberger_Procaccia_1983, Eckmann_Ruelle_1985, Shenker_Stanford_2014}.
RG techniques have become one of the most important technique in modern physics for understanding the properties of complex many-body systems.

On the other hand, recent work has shown that determining the macroscopic properties of many-body systems, even given a complete underlying microscopic description, can be even more intractable than previously anticipated.
In fact, \cite{Cubitt_PG_Wolf_Undecidability,Cubitt_PG_Wolf_Nature, Bausch_Cubitt_Watson} showed that this goal is unobtainable in general: they proved that the spectral gap of a quantum many-body system, as well as phase diagrams and any macroscopic property characterising a phase, can be uncomputable.

In this work, we show that the RG flow of such undecidable systems exhibits a novel type of behaviour, displaying a qualitatively new and more extreme form of unpredictability than chaotic RG flows.
Specifically, trajectories under the RG flow can remain arbitrarily close together for an uncomputable number of iterations
%.Yet after staying close for a number of iterations that is uncomputable, they
before abruptly diverging to different fixed points that correspond to separate phases (see \cref{fig:uncomputable_RG_flow}).
Thus, the structure of the RG flow --- e.g.\ the basins of attraction of the fixed points --- is so complex that it cannot be computed or approximated, even in principle.
A similar form of unpredictability has previously been seen in classical single-particle dynamics, in seminal work by Moore~\cite{Moore90,Moore_Long_1991,Bennett_1990}.
Our results show for the first time that this extreme form of unpredictability can occur in RG flows of many-body systems.
\begin{figure*}[t!]
	\centering
	\begin{subfigure}[t]{0.45\textwidth}
		\centering
		\includegraphics[width=\textwidth]{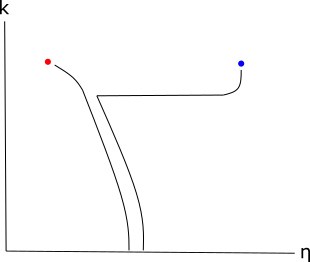}
		\caption{Uncomputable RG flow.}
		\label{fig:uncomputable_RG_flow}
	\end{subfigure}%
	~\quad
	\begin{subfigure}[t]{0.45\textwidth}
		\centering
		\includegraphics[width=\textwidth]{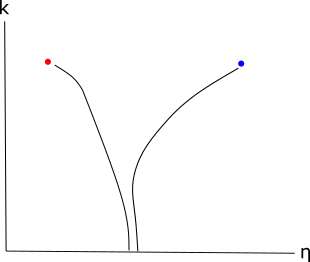}
		\caption{Chaotic RG flow.}
		\label{fig:chaotic_RG_flow}
	\end{subfigure}
	\caption{In both diagrams, $k$ represents the number of RG iterations and $\eta$ represents some parameter characterising the Hamiltonian; the blue and red dots are fixed points corresponding to different phases.
		We see that in the chaotic case, the Hamiltonians diverge exponentially in $k$, according to some Lyapunov exponent.
		In the undecidable case, the Hamiltonians remain arbitrarily close for some uncomputably large number of iterations, whereupon they suddenly diverge to different fixed points.}
\end{figure*}

The unpredictability of chaotic systems arises from the fact that even a tiny difference in the initial system parameters --- which in practice may not known exactly --- can eventually lead to exponentially diverging trajectories (see \cref{fig:chaotic_RG_flow}).
%The difficulty in predicting their behaviour arises from the fact that it is not possible in practice to determine the initial parameters exactly.
However, if the system parameters are perfectly known, it is in principle possible to determine the long-time behaviour of the RG flow.
And the more precisely the initial parameters are known, the longer it is possible to accurately predict it.

The RG flow behaviour exhibited in this work is more intractable still. Even if we know the \emph{exact} initial values of all system parameters, its RG trajectory and the fixed point it ultimately ends up at is provably impossible to predict. Moreover, no matter how close are two sets of initial parameters, it is impossible to predict how long their trajectories will remain close together.

To substantiate these claims, we give a fully rigorous mathematically proof and analysis of this qualitatively new RG behaviour, for the original undecidable many-body model in the literature~\cite{Cubitt_PG_Wolf_Undecidability, Cubitt_PG_Wolf_Nature}.
We note that our techniques can also be adapted to establish a rigorous proof of chaotic RG dynamics (see discussion in \cref{Sec:Conclusions}).
We give a rigorous construction of the block renormalisation group~\cite{Jullien_Pfeuty_Fields_1978, Jullien_Pfeuty_1979, Penson_Jullien_Pfeuty_1982, Bhattacharyya_Sil_1999} (BRG) map for this model.
We prove that the resulting RG flow converges to the correct fixed points, and preserves the order parameters and phases of the model.
Moreover, each step of the RG flow is computable (in fact, efficiently computable).
Nonetheless, the RG flow itself is uncomputable:

\begin{theorem}[Uncomputability of RG Flows -- informal statement of \cref{Theorem:Undecidability_of_RG_Flows_Formal,Theorem:tau_2_Divergence}]\label{Theorem:Informal_Undecidability_of_RG_Flows}
  We construct an RG map for the Hamiltonian of~\cite{Cubitt_PG_Wolf_Undecidability} which has the following properties:
  \begin{enumerate}
  \item The RG map is computable at each renormalisation step.
  \item The RG map preserves whether the Hamiltonian is gapped or gapless. %\eo{I suggest: ``the RG map preserves the phase of the Hamiltonian in the thermodynamic limit''}
  \item The Hamiltonian is guaranteed to converge to one of two fixed points under the RG
  flow: one gapped, with low energy properties similar to those of an Ising model with field; the other gapless, with low energy properties similar to the critical XY-model.
  \item The behaviour of the Hamiltonian under the RG mapping, and which fixed point it converges to, are uncomputable.
  \end{enumerate}
\end{theorem}

\medskip

The paper is structured as follows: in \cref{Sec:Prelims+Previous} we introduce the necessary notation and formalism, give a brief overview of real-space RG flow procedures, and review the undecidable model of~\cite{Cubitt_PG_Wolf_Nature,Cubitt_PG_Wolf_Undecidability}.
In \cref{Sec:Main_Results} we state our main results and give a high-level overview of their proofs.
The full proof of the main results is given in sections \ref{Sec:Robinson_RG}, \ref{Sec:Quantum_RG}, and \ref{sec:all_together}. \Cref{Sec:Robinson_RG,Sec:Quantum_RG} analyse the structure of real-space RG procedures applied to the undecidable model of \cite{Cubitt_PG_Wolf_Nature, Cubitt_PG_Wolf_Undecidability}; \cref{sec:all_together} proves that this RG procedure exhibits the properties and behaviour claimed in the main results.
In \cref{Sec:Fixed_Points} we discuss the properties of the fixed points of this resulting RG flow, before \cref{Sec:Conclusions} concludes.

\section{Preliminaries and Previous Work}\label{Sec:Prelims+Previous}

\subsection{Notation}

Throughout we will denote the $L\times H$ square lattice by $\Lambda(L\times H)$. 
If $L=H$ we will sometimes denote the lattice as $\Lambda(L)$.
For points $i,j\in \Lambda(L\times H)$, we will sometimes use $\langle i,j\rangle $ to denote that they are nearest neighbours.
For a Hilbert space $\mathcal{H}$, $\B(\mathcal{H})$ denotes the set of bounded linear operators on $\mathcal{H}$.
$\lambda_0(A)$ will denote the minimum eigenvalue of an operator $A\in\B(\mathcal{H})$, and more generally $\lambda_k(A)$ will denote the $(k+1)^{th}$ smallest eigenvalue.
Furthermore, we denote the spectral gap of an operator $A$ as $\Delta(A)=\lambda_1(A)-\lambda_0(A)$.
%We will typically denote renormalisation group maps by $\mathcal{R}$, and the $k$-fold iteration of this by $\mathcal{R}^{(k)}$.

Consider local interaction terms $h^{row}, h^{col}\in \mathcal{B}(\C^d\ox \C^d)$ and $h^{(1)}\in\B({\C^d})$ which define a translationally invariant Hamiltonian on an $L\times L$ lattice, $H^{\Lambda(L)}=\sum_{j=1}^L\sum_{i=1}^{L-1} h^{row}_{i,i+1} + \sum_{i=1}^L\sum_{j=1}^{L-1} h^{col}_{j,j+1} + \sum_{i,j=1}^L h^{(1)}_{i,j}$, where the sums over $i$ and $j$ are over rows and columns respectively.

We denote the renormalisation group map by $\sR$, and the $k$-fold iteration of this map by $\sRk$.
We will denote renormalised quantities and operators with $R$ or $R^{(k)}$ prefix for the renormalised and $k$-times renormalised cases respectively.
For example, denote the renormalised Hamiltonians terms as $\r(h^{row})^{i,i+1}$ and $\r(h^{col})^{j,j+1}$, and the local terms after $k$-fold iterations as $\Rk(h^{row})^{i,i+1}$ and $\Rk(h^{col})^{j,j+1}$.
We then denote the Hamiltonian defined over the lattice by the renormalised interactions as $\r(H)^{\Lambda(L)}$, and for the $k$-times iteration as $\Rk(H)^{\Lambda(L)}$.
We note that in general $\sR(h^{row}_{i,i+1})\neq \r(h^{row})^{i,i+1}$, and similarly for the other terms.

If the initial local Hilbert space is $\HS$, then the local Hilbert space after $k$ iterations of the RG map is denoted $R^{(k)}(\HS)$.
Throughout, we will denote a canonical set of local basis states by $\frk{B}$, and after the renormalisation mapping has been applied $k$ times it becomes $\frk{B}^{(k)}$, so that $\r^{(k)}(\HS)=\spann\{\ket{x}\in \frk{B}^{(k)}\}$.

It will occasionally be useful to distinguish $h^{row}$ acting on given row $j$.
When this is important, we write $h^{row}_{i,i+1}(j)$ to denote the interaction between columns $i$ and $i+1$ in the $j^{th}$ row.
Similarly $h^{col}_{j,j+1}(i)$ denotes the interaction between rows $j$ and $j+1$ in the $i^{th}$ column.

Finally, following \cite{Cubitt_PG_Wolf_Undecidability}, we adopt the following precise definitions of gapped and gapless:
\begin{definition}[Gapped, from \cite{Cubitt_PG_Wolf_Undecidability}]\label{Def:gapped}
  We say that $H^{\Lambda(L)}$ of Hamiltonians is gapped if there is a constant $\gamma>0$ and a system size $L_0\in\mathbb{ N}$ such that for all $L>L_0$, $\lambda_0(H^{\Lambda(L)})$ is non-degenerate and $\Delta(H^{\Lambda(L)})\geq\gamma$. In this case, we say that \emph{the spectral gap is at least $\gamma$}.
\end{definition}
\begin{definition}[Gapless, from \cite{Cubitt_PG_Wolf_Undecidability}]\label{Def:gapless}
  We say that $H^{\Lambda(L)}$ is gapless if there is a constant $c>0$ such that for all $\epsilon>0$ there is an $L_0\in\mathbb{N}$ so that for all $L>L_0$ any point in $[\lambda_0(H^{\Lambda(L)}),\lambda_0(H^{\Lambda(L)})+c]$ is within distance $\epsilon$ from $\spec H^{\Lambda(L)}$.
\end{definition}
We note that these definitions of gapped and gapless do not characterise all Hamiltonians; there are Hamiltonian which fit into neither definition, such as systems with closing gap or degenerate ground states.
However, \cite{Cubitt_PG_Wolf_Undecidability} showed that the particular Hamiltonians they construct always fall into one of these clear-cut cases, allowing sharp spectral gap undecidability results to be proven.

% =========================================================
\subsection{Real Space Renormalisation Group Maps}
% =========================================================

The notion of what exactly constitutes a renormalisation group scheme is somewhat imprecise, and there is no universally agreed upon definition in the literature.
We therefore start from a minimal set of conditions that we would like a mapping on Hamiltonians to satisfy, if it is to be considered a reasonable RG map.
The RG scheme we define for the Hamiltonian from \cite{Cubitt_PG_Wolf_Undecidability} will satisfy all these conditions as well as additional desirable properties.
\begin{definition}[Renormalisation Group (RG) Map] \label{Def:RG_Mapping}
  Let $\{h_i\}_i$ be an arbitrary set of $r$-local interactions $h_i\in \B((\C^d)^{\ox r})$, for $r=O(1)$ and $d\in\N$.
  A renormalisation group (RG) map
  \begin{align}
    \sR( \{h_i\} ) = \{ h_i' \}
  \end{align}
  is a mapping from one set of $r$-local interactions to a new set of $r'$-local interactions $ h_i' \in \B((\C^{d'})^{\prime \ox r'})$, with $r'\leq r$ and $d'\in\Z$, satisfying the following properties:
  \begin{enumerate}
  \item \label{RG_Condition_1} $\sR(\{h_i\})$ is a computable map.
  \item \label{RG_Condition_2} Let $H$ and $\Rk(H)$ be the Hamiltonian defined by the original local terms and the $k$-times renormalised local terms respectively.
    If $H$ is gapless, then $\r^{(k)}(H)$ is gapless, as per \cref{Def:gapless}.
    If $H$ is gapped, then $\r^{(k)}(H)$ is gapped, as per \cref{Def:gapped}.
  \item \label{RG_Condition_3}
  If the order parameter for the system has a non-analyticity between two phases of $H$, then there is a renormalised order parameter which also has a non-analyticity between the two phases for $\Rk(H)$.
  \item \label{RG_Condition_4} If the initial local Hamiltonian terms can decomposed into as
    \begin{align}
      h_i = \sum_j \alpha_j O_j,
    \end{align}
    for some operator $\{O_j\}_j$,
    then $k$-times renormalised local Hamiltonian terms are of the form
    \begin{align}
      \Rk(h)_i = \sum_j \alpha_j^{(k)} \Rk(O)_j,
    \end{align}
    where $\alpha_i^{(k)}=f(\{\alpha_i^{(k-1)}\}_i)$ for some function $f$.
  \end{enumerate}
\end{definition}

\noindent The motivation for points~\ref{RG_Condition_2} and \ref{RG_Condition_3} of \cref{Def:RG_Mapping} is that we want to preserve the quantum phase diagram of the system.
Point \ref{RG_Condition_3} of \cref{Def:RG_Mapping} requires that if we start in phase A, the system should remain in phase A under the RG flow: a key property of any RG scheme.
Furthermore, any indicators of a phase change still occur (e.g. non-analyticity of the order parameter).
Point \ref{RG_Condition_4} asks that the ``form'' of the Hamiltonian is preserved.

\medskip
Hamiltonians under RG flows have ``fixed points'' which occur where the Hamiltonian is left invariant by the action of the RG procedure. 
If $H^*$ is the fixed point a particular Hamiltonian is converging to under the RG flow, and $h^*$ is the corresponding local term, then the local terms away from the fixed point can be rewritten in terms of their deviation from the fixed point as:
\begin{align}
  h = h^* + \sum_{i} \beta_iO_i
\end{align}
and after renormalisation
\begin{align}
  \Rk(h) = h^* + \sum_{i} \beta_i^{(k)}O_i',
\end{align}
where if $\beta_i^{(k)}\rightarrow 0$ as $k\rightarrow\infty$ then $O_i$ is said to be an \emph{irrelevant operator}; if $\beta_i^{(k)}\rightarrow \infty$, then $O_i$ is a \emph{relevant operator}; and if $\beta_i^{(k)}\rightarrow c$ for a constant $c$, then $O_i$ a \emph{marginal operator}.

We note that many well-known renormalisation group schemes fit the criteria given in \cref{Def:RG_Mapping} when applied to the appropriate Hamiltonians.
In the following subsections, we review a number of these.
However, in general, a given RG scheme may satisfy the conditions for the family of Hamiltonians it was designed for, but will not necessarily satisfy all the desired conditions when applied to an arbitrary Hamiltonian.

% =========================================================
\subsubsection{The Block Spin Renormalisation Group Map} \label{Sec:Block_RG}
% =========================================================
We base our RG map on a blocking technique widely used in the literature to study spin systems, often called the Block Spin Renormalisation Group (BRG)\footnote{This is also sometimes called the ``quantum renormalisation group''.}~\cite{Jullien_Pfeuty_Fields_1978, Jullien_Pfeuty_1979, Penson_Jullien_Pfeuty_1982, Bhattacharyya_Sil_1999}.
Modifications and variations of this RG scheme have also been extensively studied~\cite{Martin-Delgado_et_al_1996, Wang_Kais_Levine_2002}.

The BRG is among the simplest RG schemes.
The procedure works by grouping nearby spins together in a block, and then determining the  associated energy levels and eigenstates of this block by diagonalisation.
Having done this, high energy (or otherwise unwanted) states are removed resulting in a new Hamiltonian.

As an explicit example, suppose there exists a Hamiltonian on a 1D chain
\begin{align}
  H = \sum_{i=1}^{N-1} K^{(0)}h^{(0)}_{i,i+1} + C^{(0)}\sum_{i=1}^N \1_i.
\end{align}
The BRG first groups the lattice points into pairs
\begin{align}
  H = K^{(0)}\sum_{i \ odd }^{N-1} h^{(0)}_{i,i+1} + K^{(0)}\sum_{i \ even}^{N-1} h^{(0)}_{i,i+1} + C^{(0)}\sum_{i=1}^N \1_i.
\end{align}
We then diagonalise the operators for odd $i$.
(In higher dimensional geometries we group the terms into blocks of neighbouring qudits.)
Having done this, remove all ``high energy states'' within each block, either by introducing an energy cut-off or just keeping a chosen subset of the lowest energy states.
The produces a renormalised Hamiltonian
\begin{align}
  \r^{(1)}(H) = K^{(1)}\sum_{i=1}^{N/2-1} h^{(1)}_{i,i+1}+b^{(1)} \sum_{i=1}^{N/2}h^{(1)}_i + C^{(1)}\sum_{i=1}^{N/2} \1_i.
\end{align}
For each further RG iteration the same process is repeated: the terms $h_{i,i+1}$ for odd $i$ are diagonalised and the high energy states are removed.

After $k$ iterations, the RG procedure returns a Hamiltonian of the same form, but now with different coupling constants:
\begin{align}
  \r^{(n)}(H) = K^{(n)}\sum_{i=1}^{N/2-1} h^{(n)}_{i,i+1}+ b^{(n)}\sum_{i=1}^{N/2}h^{(n)}_i + C^{(n)}\sum_{i=1}^{N/2} \1_i.
\end{align}

\paragraph{Form of the RG Mapping} ~\newline
This BRG mapping can be reformulated in terms of a series of isometries (or unitaries and subspace restrictions).
Given the local terms of some Hamiltonian, $h_{i,i+1}\in \B(\C^d\ox\C^d)$, we will consider renormalisation mappings of the form
\begin{align}
\sR: h_{i,i+1} \rightarrow V^\dagger h_{i,i+1} V
\end{align}
where $V:\C^d\rightarrow \C^{d'}$ is an isometry which will take a states in the initial set of basis states to a restricted new set of renormalised basis states.

Equivalently we can formulate this in terms of a unitary $U$ and a subspace 	$\Gamma$, as:
\begin{align}
\sR: h_{i,i+1} \rightarrow U^\dagger h_{i,i+1} U	|_\Gamma.
\end{align}
The unitary $U$ maps the original basis states to the new set (called \textit{blocking}).
This is followed by a restriction to the subspace $\Gamma$ which is the ``low-energy'' subspace: that is, all basis states which locally pick up too much energy are removed.
This subspace restriction is called \textit{truncating}.
In our particular variation of the BRG, the truncation step is not done entirely based on energy truncation, but also on overlap with a particular state.

\subsubsection{Comparison to Well Known RG Schemes}
\textit{Classical 1D Ising Model} ~\newline
A particularly famous RG scheme which satisfies \cref{Def:RG_Mapping} is the decimation scheme for the classical 1D Ising model \cite{Cardy_1996}.
Here the ground states are trivially either all $\sigma_i=1$ or $-1$.
Under the decimation RG procedure, half the spins are removed by ``averaging out'' the others.
The order parameter for the phase is the magnetisation: $M = \sum_{i=1}^{N} \sigma_i$ and it can be seen to undergo a non-analytic change between phases.
This is true even after renormalisation, thus satisfying point \ref{RG_Condition_3} of \cref{Def:RG_Mapping}.

The decimation mapping further gives a transformation of the form
\begin{align}
  \sR: J\sum \sigma_i\sigma_{i+1} + h\sum_i\sigma_i +CN \rightarrow J'\sum \sigma_i\sigma_{i+1} + h'\sum_i\sigma_i + C'N,
\end{align}
thus satisfying condition \ref{RG_Condition_4}.
It can also be shown~\cite{Cardy_1996} that the RG procedure preserves the phase of the Ising model, and hence satisfies condition \ref{RG_Condition_3}.

~\newline
\textit{MERA} ~\newline
A more recent and widely studied RG flow scheme in the quantum information literature is the multiscale entanglement renormalisation ansatz (MERA) developed in \cite{Vidal_2008}.
This is implemented by iteratively applying isometries to the local terms to produce new local Hamiltonian terms and density matrices.
This (approximately) preserves expectation values and hence can often be made to satisfy \ref{RG_Condition_3}.
Whether conditions \ref{RG_Condition_2} and \ref{RG_Condition_4} are satisfied is dependent on the Hamiltonian and isometries in question.

% =======================================================================
\subsection{Properties of the Spectral Gap Undecidability Construction} \label{Section:Properties_Of_Spec_Gap}
% =======================================================================
Constructing a mathematically rigorous RG flow for the undecidable Hamiltonian exhibited in~\cite{Cubitt_PG_Wolf_Nature, Cubitt_PG_Wolf_Undecidability} presents particular challenges, since its properties are uncomputable.
Nonetheless, we are able to by carefully analysing the local structure and properties of this Hamiltonian, which we review here.

% ----------------------------------------------------------------------
%\medskip

We start by stating the main result in~\cite{Cubitt_PG_Wolf_Undecidability}, where the authors construct a Hamiltonian depending on one external parameter, which is gapped iff a universal Turing Machine halts on an input related to the Hamiltonian parameter. The spectral gap problem for this Hamiltonian is therefore equivalent to the Halting Problem, hence undecidable.
\begin{definition}[From theorem 3 of \cite{Cubitt_PG_Wolf_Undecidability}]\label{Definition:CPW15_Hamiltonian}
  For any given universal Turing Machine (UTM), we can construct explicitly a dimension $d$, $d^2\times d^2$ matrices $A,A',B,C,D,D',\Pi$ and a rational number $\beta$ which can be as small as desired, with the following properties:
  \begin{enumerate}
  \item $A$ is diagonal with entries in $\Z$.
  \item $A'$ is Hermitian with entries in $ \Z+ \frac{1}{\sqrt{2}}\Z$,
  \item $B,C$ have integer entries,
  \item $D$ is diagonal with entries in $\Z$,
  \item $D'$ is Hermitian with entries in $\Z$.
  \item $\Pi$ is a diagonal projector.
  \end{enumerate}
  For each natural number $n$, define:
  \begin{align*}
    &\begin{aligned}
      &h_1(n)=\alpha(n)\Pi, &\qquad&  \\
      &h_{\text{col}}(n)=D +\beta D', &\qquad& \text{independent of $n$}
    \end{aligned}\\
    &h_{\text{row}}(n)=A + \beta\left(A'+e^{i\pi\varphi} B + e^{-i\pi\varphi} B^\dagger + e^{i\pi2^{-\abs{\varphi}}} C + e^{-i\pi2^{-\abs{\varphi}}} C^\dagger\right),
  \end{align*}
  where $\alpha(n)\le \beta$ is an algebraic number computable from $n$ and $\abs{\varphi}$ denotes the length of the binary representation of $\varphi$. Then:
  \begin{enumerate}
  \item  The local interaction strength is bounded by~1, i.e.\ \linebreak $\max(\norm{h_1(n)}, \norm{h_{\text{row}}(n)}, \norm{h_{\text{col}}(n)}) \leq 1$.
  \item If UTM halts on input $n$, then the associated family of Hamiltonians $\{H^{\Lambda(L)}(n)\}$ is gapped with gap $\gamma\ge 1$.
  \item If UTM does not halt on input $n$, then the associated family of Hamiltonians $\{H^{\Lambda(L)}(n)\}$ is gapless.
  \end{enumerate}
\end{definition}
\noindent We first explain the overall form of the Hamiltonian and the Hilbert space structure, and later how the individual parts fit together.

\subsubsection{Local Interaction Terms and Local Hilbert Space Structure}

The Hamiltonian $H_u(\varphi)$ is constructed such that its ground state is composed of two components: a classical ``tiling layer'' and a highly entangled ``quantum layer''.
The local Hilbert space decomposes as:
\begin{align}
  \HS_u &=  \HS_c \otimes (\HS_q\oplus\ket{e}),
\end{align}
where $\HS_c$ is the Hilbert space corresponding ot the classical tiling layer and $\HS_q\oplus \ket{e}$ is the ``quantum'' layer.
The local terms $h_u$ are constructed as
\begin{align}
  h_u = h_T^{(i,i+1)} \ox \1_{eq}^{(i)}\ox\1_{eq}^{(i+1)} + \1_{c}^{(i)}\ox\1_c^{(i+1)} \ox h_q^{(i,i+1)} + \text{ ``coupling terms''}.
\end{align}

Let $h_u^{(i,j)}\in \B(\C^d\ox\C^d)$ be the local terms of the Hamiltonian $H_u$, %constructed such that its ground state energy scales as either $\Omega(L^2)$ or $-\Omega(L)$,
$h_d^{(i,j)}\in \B(\C^2\ox\C^2)$ be the local interactions of the 1D critical XY model, and let $H_d$ be the Hamiltonian composed of XY interactions along the rows of the lattice.
This has a dense spectrum in the thermodynamic limit \cite{Lieb_Schultz_Mattis_1961}.
$h_u^{(i,j)}=h_u^{(i,j)}(\varphi)$ is designed so that $H_u(\varphi)=\sum h_u(\varphi)$ has a ground state energy which depends on whether a universal Turing Machine (UTM) halts when given on input $\varphi$ supplied in binary.
In particular, on a lattice of size $L\times L$, the ground state energy is
\begin{align}
  \lambda_0(H_u^{\Lambda(L)}) =
  \begin{cases}
    -\Omega(L) & \text{if UTM does not halt on input $\varphi$,} \\
    +\Omega(L^2) & \text{if UTM does halt on input $\varphi$.}
  \end{cases}
\end{align}
Since the halting problem is undecidable, determining which of the two ground state energies of $H_u(\varphi)$ occurs is undecidable.
% From \cref{Eq:Spectrum_of_H}, this then implies that the spectral gap is $\geq 1/2$ in the halting case, or gapless in the non-halting case.

The local Hilbert space of the overall Hamiltonian can be decomposed as:
\begin{align}
  \HS &= \ket{0}\oplus \HS_u \otimes \HS_d.
\end{align}
Here $\ket{0}$ is a zero-energy filler state, $\HS_d$ is the Hilbert space associated with the \textbf{d}ense spectrum Hamiltonian $h_d$, and $\HS_u$ is the Hilbert space associated with the Hamiltonian with \textbf{u}ndecidable ground state energy $h_u$.

The local interactions along the edges and on the sites of the lattice are act on this local Hilbert space as:
\begin{align}
  h(\varphi)^{(i,j)} &=\ket{0}\bra{0}^{(i)}\otimes (\1 -\ket{0}\bra{0} )^{(j)} + h_u^{(i,j)}(\varphi) \otimes \1_d^{(i,j)} + \1_u^{(i,j)} \otimes h_d^{(i,j)} \\
  h(\varphi)^{(1)}&= -(1+\alpha_2)\Pi_{ud},
\end{align}
where $\Pi_{ud}$ is a projector onto $\HS_{u}\ox \HS_{d}$, and $\alpha_2=\alpha_2(|\varphi|)$ is a constant depending only on $|\varphi|$.
Importantly, the spectrum of the overall lattice Hamiltonian composed of these local interactions is
\begin{align} \label{Eq:Spectrum_of_H}
  \spec{H(\varphi)} = \{0\}\cup \left\{ \spec(H_u(\varphi)) + \spec(H_d)   \right\}\cup S,
\end{align}
for a set $S$ with all elements $>1$. This means that if $\lambda_0(H_u^{\Lambda(L)})\rightarrow -\infty$ then the overall Hamiltonian has a dense spectrum, while if $\lambda_0(H_u^{\Lambda(L)})\rightarrow +\infty$ the overall Hamiltonian has a spectral gap $>1$.

In the $\lambda_0(H_u^{\Lambda(L)}(\varphi))=+\Omega(L^2)$ case, the ground state of the entire Hamiltonian is $\ket{0}^{\Lambda}$. In the $\lambda_0(H_u^{\Lambda(L)}(\varphi))=-\Omega(L)$ case, the overall ground state is $\ket{\psi_u}\ox \ket{\psi_d}$ where $\ket{\psi_u}$  and $\ket{\psi_d}$ are the ground states of $H_u(\varphi)$ and $H_d=\sum_{i \in \Lambda}h_d^{i,i+1}$ respectively.

% \subsubsection{The Ground State of $H_u$} \label{Sec:CPW15_Ground_State}

\bigskip

We now explain the terms $h_T$ and $h_q$ as well as the cumulative effects of the coupling terms.

\paragraph{The Tiling Hamiltonian}
Wang tiles are square tiles of unit length with markings on each side, together with rules stipulating that a pair of tiles can only be placed next to each other if the markings on their adjacent sides match.
In \cite{Cubitt_PG_Wolf_Undecidability} the tile set is chosen to be a slightly modified version the Robinson tiles from \cite{robinson1971undecidability}, shown in \cref{Fig:Modified_Robinson_Tiles}.
When placed on a 2D grid such that the tiling rules are satisfied, the markings on the tiles form an aperiodic tiling consisting of a series of nested squares of sizes $4^n+1$, for all $n\in \N$, as shown in \cref{Fig:Robinson_Tiling_Pattern}.

\begin{figure}
  \hskip8pt \subfloat[The modified Robinson tiles used in \cite{Cubitt_PG_Wolf_Undecidability}.]{\includegraphics[width=0.43\textwidth]{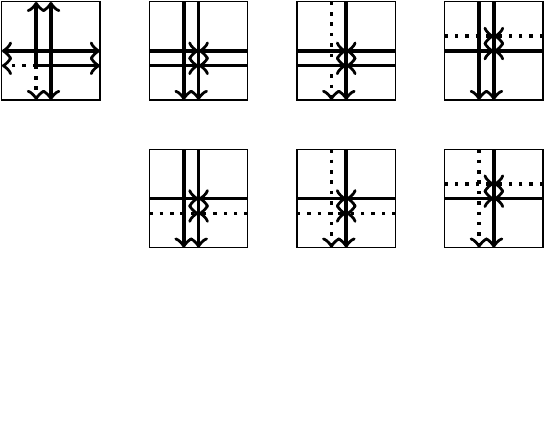}\label{Fig:Modified_Robinson_Tiles}}
  \hfill
  \subfloat[Ground state $\ket{T}_c$ of the classical Hamiltonian.]{\includegraphics[width=0.47\textwidth]{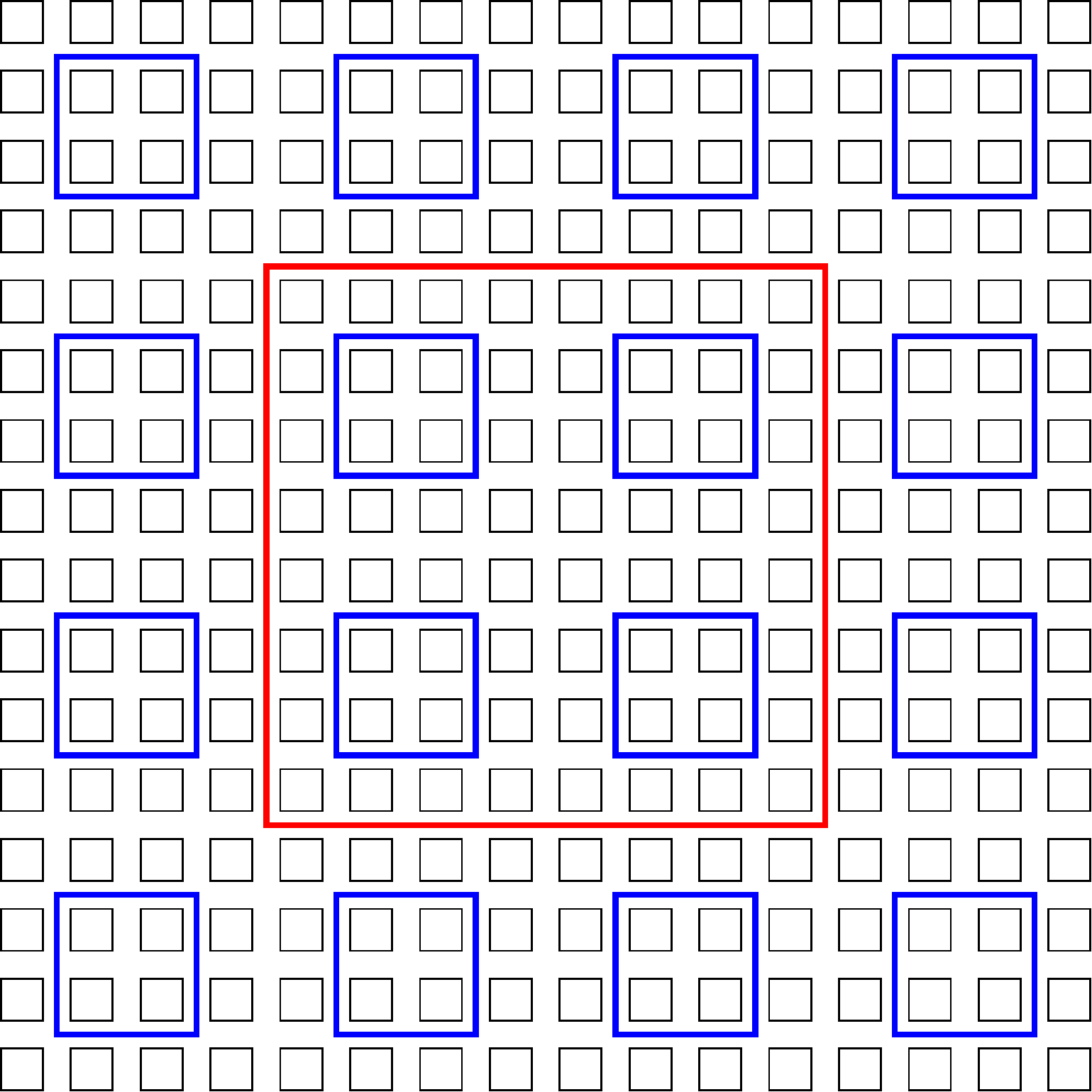}\label{Fig:Robinson_Tiling_Pattern}} \hskip10pt
  \caption{}
\end{figure}

This set of tiles can then be mapped to a 2D, translationally invariant, nearest neighbour, classical Hamiltonian by simply mapping each tile type to a state in the local Hilbert space and introducing local interactions that apply an energy penalty to neighbouring pairs which do not satisfy the tiling rules.
That is, the local terms are defined as $(h_T)_{i,i+1}:=\sum_{(t_\alpha,t_\beta)\not \in A} \ketbra{t_\alpha,t_\beta}_{i,i+1}$ where $A$ is the set of allowed neighbouring tiles.

Then, the ground state of the entire 2D lattice, $\ket{T}_c$, corresponds to the Robinson tiling pattern as shown in \cref{Fig:Robinson_Tiling_Pattern}.
Any other configuration must violate a tiling rule and thus receives an energy penalty.

\paragraph{The Quantum Hamiltonian}
$H_u(\varphi)$ is constructed so that its ground state energy encodes the halting or non-halting of a computation.
The fundamental ingredient required is the ``QTM-to-Hamiltonian'' mapping~\cite{Gottesman-Irani,Cubitt_PG_Wolf_Undecidability}.
This takes a given quantum QTM and creates a corresponding Hamiltonian which has a ground state which encodes its evolution.
This quantum state is called a \emph{history state}.
Let $\ket{\psi_t}$ be the state describing the configuration of the QTM after $t$ steps.
Then the history state takes the general form
\begin{equation}\label{eq:history_state}
  \ket{\Psi_{hist}} = \frac{1}{\sqrt{T}}\sum_{t=1}^T \ket{\psi_t}\ket{t},
\end{equation}
where $\ket{t}$ is a state labelling which step of the computation $\ket{\psi}$ corresponds to.

It is then possible to add a local projector term to the Hamiltonian which gives an additional energy penalty to certain outcomes of the computation.
In particular, \cite{Cubitt_PG_Wolf_Undecidability} penalise the halting state, so that if the QTM halts at some point, the Hamiltonian defined by $h_q$ picks up an additional energy contribution.
As a result, the energy of the ground state differs depending on whether or not the QTM halts within time $T$.

In particular,
\cite{Cubitt_PG_Wolf_Undecidability} adapt the QTM-to-Hamiltonian mapping originally developed by Gottesman and Irani~\cite{Gottesman-Irani}, which takes a QTM and maps its evolution to 1D, translationally invariant, nearest neighbour, Hamiltonian.
By $H_q$ we denote this modified version of the \emph{Gottesman-Irani Hamiltonian} (cf.~\cref{sec:Gottesman_Irani_Hamiltonian}).

The length of the computation encoded on a chain of length $L$ is $T(L)\sim \poly(L)2^L$, and the associated ground state energy is
\begin{align}
  \lambda_0(H_{q}(L))=
  \begin{cases}
    0 & \text{if QTM is non-halting within time $T(L)$,} \\
    \theta(1/T^2) & \text{if QTM halts within time $T(L)$.}
  \end{cases}
\end{align}
We give a more detailed analysis of the construction at the beginning of \cref{Sec:Quantum_RG}.

\paragraph{Combining $h_T$, $h_q$ and the Coupling Terms}
The terms $h_u$ are designed so that all eigenstates of $H_u^{\Lambda(L)}$ are product states $\ket{T}_c\ox \ket{\psi}_{eq}$ where $\ket{T}\in \HS_c^{\ox (L\times L)}$ and $\ket{\psi}\in \HS_{eq}^{\ox (L\times L)}$~\cite[Lemma~51]{Cubitt_PG_Wolf_Undecidability}.

Furthermore, the coupling terms are chosen such that the ground state has the following properties:
\begin{enumerate}
	\item the classical part of the ground state $\ket{T}_c$ corresponds to a perfect Robinson tiling.
	  The pattern created has a series of nested red Robinson squares as per \cref{Fig:Robinson_Tiling_Pattern}.
	\item the quantum part of the ground state $\ket{\psi}_{eq}$ has the following structure: along the top of every red Robinson square there is a history state (as defined in \cref{eq:history_state}); everywhere which is not along the top of a square is in the zero energy filler state $\ket{e}_e$.
\end{enumerate}

The consequence of this is that ground states of $H_q(\ell)$ of all lengths appear with a constant density across the lattice.
If, for any length, the encoded computation halts, then the ground state picks up a constant energy density, so that the energy scales as $\Omega(L^2)$.
However, if the encoded computation never halts, then for all lengths the ground state of the Gottesman-Irani Hamiltonian has zero energy, and (due to boundary effects), the ground state has energy $-\Omega(L)$~\cite{Cubitt_PG_Wolf_Undecidability}.

\subsection{The Gottesman-Irani Hamiltonian}\label{sec:Gottesman_Irani_Hamiltonian}
The particular circuit-to-Hamitonian mapping used in the previous section will be important when it comes to renormalising the overall Hamiltonian.
The overall structure used in \cite{Cubitt_PG_Wolf_Undecidability} is a modification of the one used in \cite{Gottesman-Irani}.

We start by defining one of the core concepts behind the construction an behaviour of the Hamiltonian of~\cite{Cubitt_PG_Wolf_Undecidability}: the Quantum Turing Machine (QTM).
\begin{definition}[Quantum Turing Machine~\cite{Bernstein1997}]
	A \emph{quantum Turing Machine} (QTM) is defined by a triplet $(\Sigma; \mathcal Q; \delta)$ where $\Sigma$ is a finite alphabet with an identified blank symbol $\#$, $\mathcal Q$ is a finite set of states with an identified initial state $q_0$ and final state $q_f\neq q_0$ , and $\delta$  is the \emph{quantum transition function}
	\begin{equation}
	\delta : \mathcal Q \times \Sigma \rightarrow \mathds{C}^{\Sigma \times \mathcal Q \times [L,R]}
	\end{equation}
	The QTM has a two-way infinite tape of cells indexed by $\mathds{Z}$ and a single read/write tape head that moves along the tape.
	A \emph{configuration} of the QTM is a complete description of the contents of the tape, the location of the tape head and the state $q \in \mathcal Q$ of the finite control.
	At any time, only a finite number of the tape cells may contain non-blank symbols.
	The initial configuration satisfies the following conditions: the head is in cell $0$, called the starting cell, and the machine is in state $q_0$.

	We say that an initial configuration has input $x \in (\Sigma \backslash \{\#\})^\ast$ if $x$ is written on the tape in positions $0, 1, 2, \dots$ and all other tape cells are blank.
	The QTM halts on input $x$ if it eventually enters the final state $q_f$.
	The number of steps a QTM takes to halt on input $x$ is its \emph{running time on input $x$}.\newline %
	Let $\mathcal S$ be the inner-product space of finite complex linear combinations of configurations of the QTM $M$ with the Euclidean norm.
	We call each element $\phi \in \mathcal S$ a \emph{superposition} of $M$.\newline %
	The QTM $M$ defines a linear operator $U_M : \mathcal S \rightarrow \mathcal S$, called the \emph{time evolution operator of $M$}, as follows: if $M$ starts in configuration $c$ with current state $p$ and scanned symbol $\sigma$, then after one step $M$ will be in superposition of configurations $\psi = \sum_j \alpha_j c_j$, where each non-zero $\alpha_j$ corresponds to the amplitude $\delta(p; \sigma; \tau; q; d)$ of $\ket{\tau}\ket{q}\ket{d}$ in the transition $\delta(p; \sigma)$ and $c_j$ is the new configuration obtained by writing $\tau$, changing the internal state to $q$ and moving the head in the direction of $d$.
	Extending this map to the entire $\mathcal S$ through linearity gives the linear time evolution operator $U_M$.
\end{definition}

\medskip

Following~\cite{Gottesman-Irani}, the QTM can be encoded into a 1D, translationally-invariant, nearest-neighbour Hamiltonian, which we refer to as \emph{a Gottesman-Irani Hamiltonian}, denoted by $H_q(L)\in \B((\C^d)^{\ox L})$. %(cf.\cite{Cubitt_PG_Wolf_Nature,Cubitt_PG_Wolf_Undecidability}~).
This is summarised by theorem 32 of~\cite{Cubitt_PG_Wolf_Undecidability}; we write out a slightly simpler version here as the specific details are not important for our purposes.
These constructions will be needed in order to formulate~\cref{Lemma:RG_GI_Properties} for the block-renormalisation of the quantum Hamiltonian.
\begin{theorem}[Informal Version of Theorem 32 of \cite{Cubitt_PG_Wolf_Undecidability}]
	\label{QTM_in_local_Hamiltonian} \hfill\\
	Let $\C^d = \C^C\otimes\C^Q$ be the local Hilbert space of a 1\nobreakdash-dimensional chain of length $L$, with special marker states $\ket{\leftend},\ket{\rightend}$.
	Denote the orthogonal complement of $\spann(\ket{\leftend}, \ket{\rightend})$ in $\C^d$ by $\C^{d-2}$.
	Let $d, Q$ and $C$ all be fixed.

	For any well-formed unidirectional Quantum Turing Machine $M = (\Sigma,Q,\delta)$ and any constant $K>0$, we can construct a two-body interaction $h \in \cB(\C^d\ox\C^d)$ such that the 1\nobreakdash-dimensional, translationally-invariant, nearest-neighbour Hamiltonian $H(L)=\sum_{i=1}^{L-1} h^{(i,i+1)} \in \cB(\HS(L))$ on the chain of length $L$ has the following properties:
	\begin{enumerate}
		\item \label[part]{QTM_in_local_Hamiltonian:local_dim}%
		$d$ depends only on the alphabet size and number of internal states of $M$.

		\item \label[part]{QTM_in_local_Hamiltonian:FF}%
		$h \geq 0$, and the overall Hamiltonian $H(L)$ is frustration-free for all $L$.

		\item \label[part]{QTM_in_local_Hamiltonian:gs}%
		Denote $\HS(L-2) := (\C^{d-2})^{\ox L-2}$ and define $\mathcal{S}_{br}=\spann(\ket{\leftend}) \ox \HS(L-2) \ox \spann(\ket{\rightend})\subset\HS$.
		Then the unique ground state of $H(L)|_{\mathcal{S}_{br}}$ is a computational history state (cf. \cref{eq:history_state} for a definition) encoding the evolution of $M$.

	\end{enumerate}
	Moreover, the action of $M$ satisfies:
	\begin{enumerate}%\setcounter{enumi}{\value{tmpcounter}}
		\item \label[part]{QTM_in_local_Hamiltonian:halt}%
		The computational history state always encodes $\Omega(2^L)$ time-steps. If $M$ halts in fewer than the number of encoded time steps, exactly one $\ket{\psi_t}$ has support on a state $\ket{\top}$ that encodes a halting state of the QTM. The remaining time steps of the evolution encoded in the history state leave $M$'s tape unaltered, and have zero overlap with $\ket{\top}$.

		\item \label[part]{QTM_in_local_Hamiltonian:out-of-tape}%
		If $M$ runs out of tape within a time $T$ less than the number of encoded time steps, the computational history state only encodes the evolution of $M$ up to time $T$. The remaining steps of the evolution encoded in the computational history state leave $M$'s tape unaltered.
	\end{enumerate}
\end{theorem}

We provide in the following a more detailed sketch of how the modified Gottesman-Irani construction works, and refer the reader to \cite{Cubitt_PG_Wolf_Undecidability,Gottesman-Irani} for a detailed overview.
We begin by considering the general setup.
Our basis states for $(\C^{ d})^{\ox L}$ (i.e. the chain of length $L$) will have the following structure:
\begin{center}
	\begin{tabular}{|l|clc|r|}
		\hline
		$\leftend$ & $\cdots$ & Track 1: Clock oscillator & $\cdots$ & $\rightend$ \\
		\hline
		$\leftend$ & $\cdots$ & Track 2: Counter TM head and state & $\cdots$ & $\rightend$\\
		\hline
		$\leftend$ & $\cdots$ & Track 3: Counter TM tape & $\cdots$ & $\rightend$\\
		\hline
		$\leftend$ & $\cdots$ & Track 4: QTM head and state & $\cdots$ & $\rightend$\\
		\hline
		$\leftend$ & $\cdots$ & Track 5: QTM tape & $\cdots$ & $\rightend$\\
		\hline
		$\leftend$ & $\cdots$ & Track 6: Time-wasting tape & $\cdots$ & $\rightend$\\
		\hline
	\end{tabular}
\end{center}
The local Hilbert space at each site is the tensor product of the local Hilbert space of each of the six tracks $\HS = \bigotimes_{i=1}^6\HS_{i}$.

The outline of the construction is the following: tracks 1 encodes the evolution of an oscillator which goes back and forth along its track as per \cref{Fig:Clock-Oscillator}
\begin{figure}[b]
	\centering
	\includegraphics[width=0.2\textwidth]{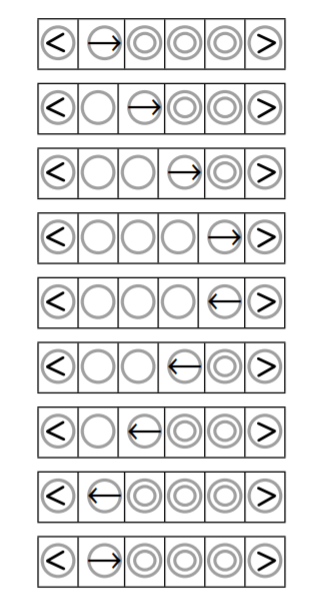}
	\caption{Evolution of the Track~1 clock oscillator.}
	\label{Fig:Clock-Oscillator}
\end{figure}
Tracks 2 and 4 contain the heads of a classical and quantum TM respectively.
These heads are only able to move when the oscillator on track 1 passes by their heads -- in this way their evolution can be encoded with only local Hamiltonian terms.
Tracks 3 and 5 are the read/write tapes for the respective TMs.

The classical TM encoded by the track 2 head will be a simple counter: it will write out binary number on its tape (on track 3) and then increment it by one to the next binary number.
This continues until the tape is filled, at which point it halts along with the clock oscillator.

The QTM on tracks 4 and 5 will be a generic QTM.
The QTM evolves as per its transition rules until either: (a) the counter TM runs out of space and hence the oscillator stops, or (b) the QTM finishes its computation and halts.
If the QTM halts before the counter TM runs out of steps, it places a halting marker on track 5.
The head then moves to track 6 where it performs some arbitrary time wasting computation which is guaranteed not to halt before the counter TM.

We also note that tracks 1-3 evolved entirely classically whereas tracks 4-6 will contain quantum states.
As such, we decompose the local Hilbert space into a classical and quantum part $\C^C\ox \C^Q$.

%-----------------------------------------------------------------------
\subsection{Order Parameters}\label{Sec:Order_Parameter}

Back to claim~\ref{RG_Condition_3} of \cref{Def:RG_Mapping}, we now discuss order parameters in more detail.
As noted in \cite{Bausch_Cubitt_Watson}, the two phases\footnote{Phase in this context refers to the state of matter, not a quantum mechanical phase factor (of the form $e^{i\theta}$).} of the Hamiltonian (which we label A and B for convenience) can be distinguished by an order parameter $O_{A/B}$, defined as
\begin{align}
  O_{A/B} = \frac{1}{|\Lambda|}\sum_{i\in \Lambda}\ketbra{0}^{(i)}.
\end{align}
In particular, upon moving from one phase to another, the expectation value of the order parameter is expected to undergo a non-analytic change.
In the case $\lambda_0(H_u^{\Lambda(L)}(\varphi))=+\Omega(L^2)$ the ground state of the entire Hamiltonian is then $\ket{0}^{\Lambda}$ and hence $\langle O_{A/B} \rangle = 1$, and otherwise $\langle O_{A/B} \rangle =0$.
This is true even if we restrict $O_{A/B}$ to subsections of the lattice, hence $O_{A/B}$ is a local order parameter (as opposed to the global order parameters required to distinguish topological phases).
Thus $O_{A/B}$ undergoes a non-analytic change between phases, which itself demonstrates a phase transition.
More generally for a ball $B(r)$ of radius $r$, and for a state $\ket{\nu}\in \HS^{\ox \Lambda}$ we can define a local observable
\begin{align}\label{Eq:Order_Parameter_Definition}
O_{A/B}(r) = \frac{1}{|B(r)|}\sum_{i\in B(r)} \ketbra{0}^{(i)},
\end{align}
which acts as a local order parameter.

% ===================================================================
% ===================================================================
\section{Main Results and Overview of RG Procedure} \label{Sec:Main_Results}
% ===================================================================
% ===================================================================

\begin{theorem}[Exact RG flow for Undecidable Hamiltonian]\label{Theorem:Undecidability_of_RG_Flows_Main_Results}
  Let $H(\varphi)$ be the Hamiltonian defined in \cite{Cubitt_PG_Wolf_Undecidability}.
  We construct a renormalisation group procedure for the Hamiltonian which has the following properties:
  \begin{enumerate}
  \item \label{RG_Condition_1_3} $\sR$ is computable.
  \item \label{RG_Condition_2_3} If $H(\varphi)$ is gapless, then $\Rk(H(\varphi))$ is gapless, and if $H(\varphi)$ is gapped, then $\Rk(H(\varphi))$ is gapped (where gapped and gapless are defined in \cref{Def:gapped} and \cref{Def:gapless}).
  \item \label{RG_Condition_3_3} For the order parameter $O_{A/B}(r)$ (as defined in \cref{Eq:Order_Parameter_Definition}) which distinguishes the phases of $H^{\Lambda(L)}$ and is non-analytic at phase transitions, there exists a renormalised observable $\Rk(O_{A/B}(r))$ which distinguishes the phases of $\Rk(H)^{\Lambda(L)}$ and is non-analytic at phase transitions.
  \item \label{RG_Condition_4_3} Under an arbitrary number of iterations, the renormalised local interactions belong to a family $\mathcal{F}(\varphi, \tau_1, \tau_2, \{\alpha_i\}_i,  \{\beta_i\}_i)$, and for any finite $k$ all of the parameters are computable.
  \item \label{RG_Condition_5_3} If $H(\varphi)$ initially has algebraically decaying correlations, then $\Rk(H(\varphi))$ also has algebraically decaying correlations.
    If $H(\varphi)$ initially has zero correlations, then $\Rk(H(\varphi))$ also has zero correlations.
  \end{enumerate}
\end{theorem}

\begin{theorem}[Uncomputability of RG flows]\label{Theorem:tau_2_Divergence_Main_Results}
  Let $h(\varphi)$, $\varphi\in \mathbb{Q}$, be the full local interaction of the Hamiltonian from \cite{Cubitt_PG_Wolf_Undecidability}.
  $H(\varphi):=\sum h(\varphi)^{(i,j)}$ is gapped if the UTM corresponding to $h{(\varphi)}$ halts on input $\varphi$, and gapless if the UTM never halts, where gapped and gapless are defined in \cref{Def:gapped} and \cref{Def:gapless}.
  Consider $k$ iterations of the RG scheme (defined later in \cref{Def:Full_RG_Mapping}) acting on $H(\varphi)$, such that the renormalised local terms are given by  $\Rk(h(\varphi))$, which can be parameterised as part of the family $\mathcal{F}(\varphi, \tau_1, \tau_2, \{\alpha_i\}_i,  \{\beta_i\}_i)$ (as per \cref{Corollary:Uncomputable_Parameters}).
  Then, if the UTM is non-halting on input $\varphi$, for all $k>k_0(\varphi)$,  $\tau_2(k)=-2^{k}$, for some computable $k_0(\varphi)$.
  If the UTM is halting on input $\varphi$, then there exists an uncomputable $k_h(\varphi)$ such that for $k_0(\varphi)<k< k_h(\varphi)$, $\tau_2(k)=-2^{k}$, and for all $k>k_h(\varphi)$ then $\tau_2(k)=-2^{k}+\Omega(4^{k-k_h(\varphi)})$.
\end{theorem}
\noindent A direct consequence of this is:
\begin{corollary}\label{Corollary:Fixed_Point_Undecidability}
  Determining which fixed point the Hamiltonian flows to under this RG scheme is undecidable.
\end{corollary}

The overall RG scheme is explicitly given in \cref{Def:Full_RG_Mapping}, and the family \newline $\mathcal{F}(\varphi, \tau_1, \tau_2, \{\beta_i\})$ which the renormalised Hamiltonians belong to is given in \cref{Corollary:Family}.
One of the consequences of \cref{Theorem:tau_2_Divergence_Main_Results} is that the Hamiltonian is guaranteed to flow towards one of two fixed points.
However, determining which fixed point it flows to for a given value of $\varphi$ is undeciable.
% We will see if $H(\varphi)$ is gapped, we see the Hamiltonian flows towards an Ising-type fixed point.
% If it is gapless, it flows towards a Hamiltonian in the same class as the critical XY-model.

The undecidability of the fixed point follows implicitly from undecidability of the spectral gap \cite{Cubitt_PG_Wolf_Undecidability, Cubitt_PG_Wolf_Nature}, since the fixed point depends on the gappedness of the unrenormalised Hamiltonian.
However, \cref{Theorem:tau_2_Divergence_Main_Results} shows precisely how the trajectory of the Hamiltonian in parameter space diverges in an uncomputable manner under RG flow.

% =====================================================================
% =====================================================================
\subsection{Overview of the proof of the main results}
% =====================================================================

The renormalisation group scheme we will employ will be a variant of the BRG described in \cref{Sec:Block_RG},
where we block $2\times 2$ groups of spins to a single ``super-spin'' which preserves some of the properties of the original set.
Due to the complexity of the Hamiltonian in consideration, we will first renormalise the different parts $h_u, h_d, \ket{0}$ of the Hamiltonian separately, then combine these RG maps into the complete map.
For a finite size lattice, $h_u$ has a ground state which is product between $\HS_C$ and $\HS_q\oplus \ket{e}$.
This key property allows us to essentially renormalise the tiling Hamiltonian and the Gottesman-Irani Hamiltonian separately.

\paragraph{Renormalising the Tiling Hamiltonian %$h_T$
}~\newline
\Cref{Fig:Robinson_Tiling_Pattern} shows that the ground state of the tiling Hamiltonian corresponds to a particular pattern; notably the Robinson tiling creates a self-similar pattern for across all sizes of squares, where smaller squares are nested within larger ones.
We design a blocking procedure which takes a set of $2\times 2$ Robinson tiles, then maps them onto a single new tile which has the same markings and tiling rules as one in the original set of Robinson tiles.
Doing this we recover a set of tiles which recreate the Robinson tiling pattern, but now with the smallest squares ``integrated out''.
Repeated iterations of this process still preserve the Robinson tiling pattern.
The details are give in \cref{Sec:Robinson_RG}.

\paragraph{Renormalising the Gottesman-Irani Hamiltonian %$h_q$
}~\newline
The Gottesman-Irani Hamiltonian $h_q$ is a 1D Hamiltonian which serves as a QTM-to-Hamiltonian map.
As noted in section \cref{Section:Properties_Of_Spec_Gap}, in the ground state of $\sum h_u$, ground states of Gottesman-Irani Hamiltonians appear along the top edge of the Robinson tiles.
We aim to design an RG scheme such that the energy of the Gottesman-Irani ground state attached to a square remains the same even when the square size is halved.
To do this, we map pairs of spins to a new ``combined spin'' which now has local Hilbert space dimension $d^2$ if the original dimension is $d$.
As with the BRG, we consider the new 1-local terms and diagonalise them. Since we know the form of the ground state explicitly, it is possible to identify states which pick up too much energy to have overlap with the ground state.
We can truncate the local Hilbert space by removing these states and hence reduce the dimension of the combined spin to something $<d^2$ (but still $>d$).
This blocking procedure will preserve whether the Hamiltonian has a zero energy ground state or a ground state with energy $>0$.

In mathematical terms, the procedure is implemented by a series of isometries
which are used to map the original states to the new blocked states, and then subspace restrictions which remove the high energy states.
This is summarised in \cref{Lemma:RG_GI_Properties}.
We refer the reader to \cref{Sec:Quantum_RG} for full details.

%---------------------------------------------------------------
\paragraph{Renormalising $h_u$}

Since $h_u = h_T^{(i,i+1)} \ox \1_{eq}^{(i)}\ox\1_{eq}^{(i+1)} + \1_{c}^{(i)}\ox\1_c^{(i+1)} \ox h_q^{(i,i+1)} + \text{ coupling terms} $, to renormalise it, we do the following:
\begin{itemize}
\item Choose a $2\times 2$ block of spins.
\item Renormalise the classical tiling part of the Hamiltonian as above.
\item To renormalise the quantum part of the Hamiltonian, break the $2\times 2 $ block into two $2\times 1$ blocks.
  Renormalise these two sections as the above renormalisation for the Gottesman-Irani Hamiltonian.
  The $2\times 2$ block is now a $2\times 1 $ block.
\item Trace out part of the Hilbert space such the $2\times 1 $ block is now a single site in the renormalised Hilbert space such that we are left with 1-local and 2-local projector terms which introduce an energy shift.
This energy shift exactly compensates for any energy lost in the integrating out operation.
\end{itemize}
The above can be shown to preserve the ground state energy in the desired way.
See \cref{Def:h_u_RG_Mapping} in \cref{sec:renormlaising_3_H} for the complete description.

\paragraph{Renormalising the Entire Hamiltonian }

We renormalise $h_d$ and $\ketbra{0}$ in a trivial way such that their properties are preserved.
Thus the overall renormalisation scheme acts on $h_u$ as above, and essentially leaves $h_d$ and $\ketbra{0}$ unchanged.

Since $h_u$, $h_d$ and $\ketbra{0}$ have their respective ground state energies preserved (approximately), whether the ground state is $\ket{0}^{\Lambda}$ or the more complex ground state of the tiling+quantum Hamiltonian, is preserved.
Importantly it can be shown the spectral gap of both cases is preserved.
The RG process can then be iterated arbitrarily many times: we show the relevant properties are preserved throughout.
Determining the properties of the ground state and spectral gap are undecidable for the unrenormalised Hamiltonian, and since these properties are preserved by the RG mapping, it is also undecidable for the renormalised Hamiltonians.

The renormalisation of the entire Hamiltonian is given in detail in \cref{sec:all_together}.

% =====================================================================
% =====================================================================
\section{Renormalisation of the Robinson tiling Hamiltonian} \label{Sec:Robinson_RG}
% =====================================================================
% =====================================================================

In the following we will construct an RG map under which the two graphs representing respectively the \emph{adjacency relations} (roughly speaking, the rules telling us what tiles can stay above / below / left / right of a given tile) for the Robinson tiles and for a specific subset of $2 \times 2$ supertiles are isomorphic.
This implies that the pattern produced by the tiling of the 2D plane using Robinson tiles is scale-invariant.
This property is crucial in order to ensure that the density of the Gottesman-Irani ground states (corresponding to the top edges of the squares appearing in the pattern) which encode the QTM is preserved under the renormalisation procedure.

\medskip

More formally, we have that

\begin{theorem}(Adjacency Rules Isomorphism)\label{thm:tiles_graph_isomorphism}
  Let $T_1$ be the set of Robinson tiles and $A_1$ be the corresponding adjacency rules.
  Let $T_2$ be the set of $2 \times 2$ supertiles, obtained from all combinations allowed by $A_1$ of four Robinson tiles placed in a $2 \times 2$ square, and $A_2$ be the adjacency rules of $T_2$, derived from the principle that two supertiles can be placed next to each other only if the Robinson tiles on the edges that are put adjacent respect $A_1$.
  Then there exists a subset $T'_2\subset T_2$, $|T'_2|=|T_1|=56$, with tiling rules $A_2'=A_2|_{T_2'}$, and a bijection $T_2'\rightarrow T_1$ under which $A_1$ and $A_2'$ are equivalent.
\end{theorem}

From this result it follows that (cf.\ \cref{appendix:pattern})

\begin{corollary}(Scale Invariance of the Robinson Tiling)\label{cor:Robinson_pattern}
  Under the bijection in \cref{thm:tiles_graph_isomorphism}, the Robinson tiling pattern is preserved under the $2\times 2\rightarrow1\times 1$ renormalisation of the grid.
\end{corollary}

We can then translate this scale invariance into a statement about the properties of the Hamiltonian which describes the Robinson tiling, i.e.,
\begin{theorem}[(Informal) Robinson Tiling Hamiltonian Renormalisation]
  Let $h_T\in \C^T \otimes \C^T$ be the local interactions which describe the Robinson tiling Hamiltonian.
  Then there exists a renormalisation group mapping $\sR_T$ satisfying
  $R_T(h_T)=h'_T$, where $h_T'\in \C^T \otimes \C^T$, such that $R_T(h_T)$ preserves both the ground state energy and the tiling pattern.
\end{theorem}

\vspace{0.3cm}

Before the explicit construction of the re-scaling transformation, we shall recall the Robinson tiles and their adjacency rules.

% ==============================================================
\subsection{Robinson Tiling}
% ==============================================================

\begin{figure}
	\begin{center}
		\includegraphics[width=0.9\textwidth]{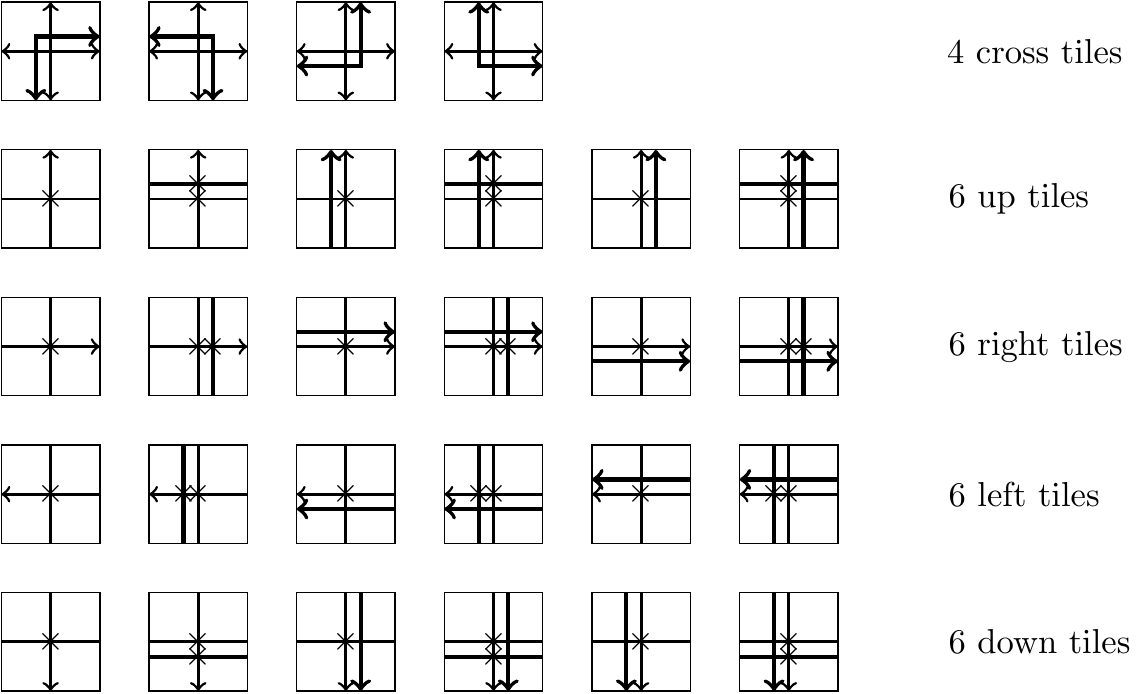}
	\end{center}
	\caption{}
	\label{fig:Robinson_tiles}
\end{figure}
%---------------------------------------------------------------

Two tiles can be placed adjacent to each other only if their arrows are compatible.
That is, the head(s) of the arrow(s) in one tile and the tail(s) of the arrow(s) in the other tile must match exactly on the edges put into contact.
We refer to~\cite{robinson1971undecidability} for a complete description.

\smallskip

Recall that there are 28 different arrow markings in the Robinson tiles set, which we list in \cref{fig:Robinson_tiles}.
Following Robinson, these arrow markings are augmented with 4 parity tiles in a way that gives rise to 56 total different tiles.
More precisely, we consider the coloured tiles given in \cref{Fig:Parity_Tiles}, which following Robinson we call the \emph{parity tiles}, satisfying the tiling rules stating that only borders with the same colour can be placed next to each other.
\begin{figure}[b]
  \begin{center}
    \includegraphics[width=1\textwidth]{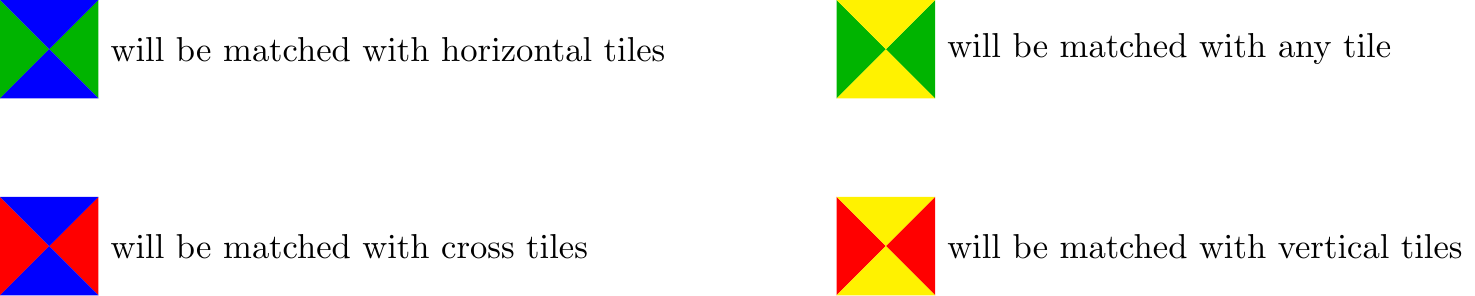}
  \end{center}
  \caption{}
  \label{Fig:Parity_Tiles}
\end{figure}
Each parity tile can be thought of as being attached to a Robinson tile in another layer.
Thus, tiles are only allowed to be placed next to each other if both their Robinson markings and parity markings match along the edge in contact.

\medskip

We will use the following terminology from~\cite{robinson1971undecidability}: cross tiles matched with the red/blue parity tile will be called ``parity crosses'';
horizontal tiles coupled with green/blue tiles will be denoted as ``parity horizontal'' and analogously vertical tiles linked to red/yellow parity will be called ``parity vertical''.
Conversely, any tile associated to the green/yellow tile will be called ``free'', so we will have ``free crosses / horizontal / vertical'' tiles.\newline
When building adjacency rules, both arrow and parity rules must be obeyed.

\medskip

Parity tiles will force the following structure.
Considering the plane as a grid of cells where the tiles are to be placed, then parity cross tiles will appear in alternating rows and alternating columns.
The same applies for parity horizontal and parity vertical tiles.
Thus, if we consider a grid of $2\times 2$ blocks over the plane, each $2\times 2$ supertile will have the same inner parity structure.
Depending on where we place the grid, we will obtain one of the configurations illustrated in~\cref{fig:inner_parity}, repeated over the whole plane.

%\begin{figure}
%  \begin{center}
%    \includegraphics[width=0.8\textwidth]{parity_structures}
%  \end{center}
%  \caption{}
%  \label{fig:inner_parity}
%\end{figure}

\begin{figure}
	\hskip8pt \subfloat[The four parity structures of a $2 \times 2$ cell]{\includegraphics [width=0.4\textwidth]{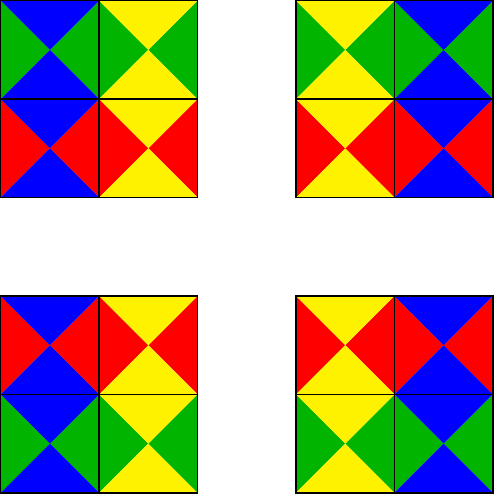}\label{fig:inner_parity}}
	\hfill
	\subfloat[pattern of the parity structure  of the plane]{\includegraphics[width=0.47\textwidth]{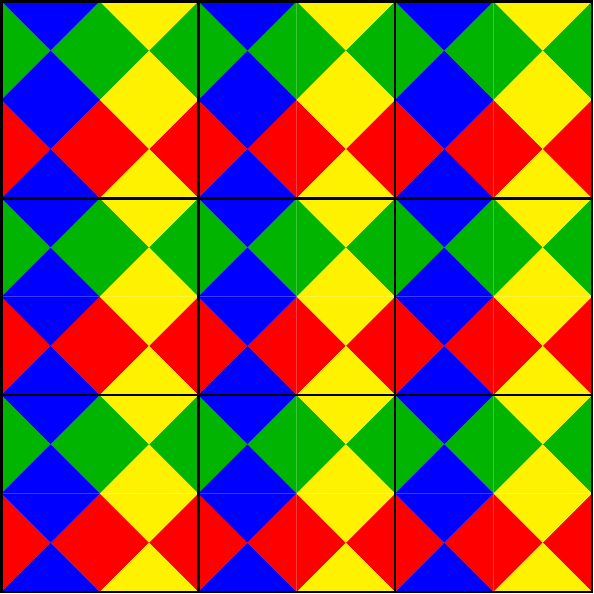}\label{Fig:Parity_Pattern}} \hskip10pt
	\caption{}\label{supertiles_parity}
\end{figure}

% ===============================================================
\subsection{Tiling Renormalisation} \label{Sec:Tiling_Renormalization}
% ===============================================================

In this section we will provide the proof of \cref{thm:tiles_graph_isomorphism}.
When changing the grid size, we go from $1\times 1$ Robinson tiles to $2\times 2$ supertiles. As we noted above, depending on the positioning of the grid, we will obtain one of the inner parity structures given in~\cref{fig:inner_parity} .
From this point we will consider  the first supertile on the top-left in~\cref{fig:inner_parity}, that is, the one with the parity cross on the bottom-left.
The parity structure of the plane will then look as shown in \cref{Fig:Parity_Pattern}.

%\begin{figure}
%	\begin{center}
%		\includegraphics[width=0.6\textwidth]{parity_pattern}
%	\end{center}
%	\caption{}\label{Fig:Parity_Pattern}
%\end{figure}

With this parity structure in mind, we generate all $2\times 2$ supertiles permitted by the arrows rules.
There are a total of 68 such supertiles, that we will call \emph{allowed supertiles}. %(see~\cref{appendix:mathematica} for the explicit construction).
Our aim is to identify a bijection between a subset of these 56 supertiles and the Robinson tiles that leads to \cref{thm:tiles_graph_isomorphism}.
In other words, we consider the adjacency relations of the $2\times 2$ tiles: they will generate a \emph{directed graph}. We want to prove that this graph is isomorphic to the one describing the relations of the original Robinson tiles.

\bigskip

Interestingly, from the approach that aims to replicate the Robinson pattern with supertiles described in \cref{appendix:pattern}, we observe that we can formulate the projection from $2\times 2$ to $1\times 1$ tiles by looking at the two tiles that occupy the bottom-left and the top-right position of the supertile.
Indeed, once the tiles on the bottom-left and top right position are placed, there is only one possible choice for the two remaining tiles, which must also obey the inner parity structure of the supertile (see~\cref{fig:inner_parity}).
This fact leads to the following definition of the renormalisation map.

% ===definition RENORMALIZATION MAP==============================
\begin{definition}[Renormalisation Map]\label{def:renormalization_map}
  Given an allowed $2\times 2$ supertile, we consider its top-right tile with free parity, that we denote by $T$, and the parity cross on the bottom-left position, that we call $C$.
  The associated Robinson tile under the renormalisation map has the same marking as $T$ and parity characterised by $C$ according to the correspondence given in~\cref{fig:crosses_parity}.
\end{definition}

\begin{figure}[h]
	\begin{center}
	\includegraphics[width=0.9\textwidth]{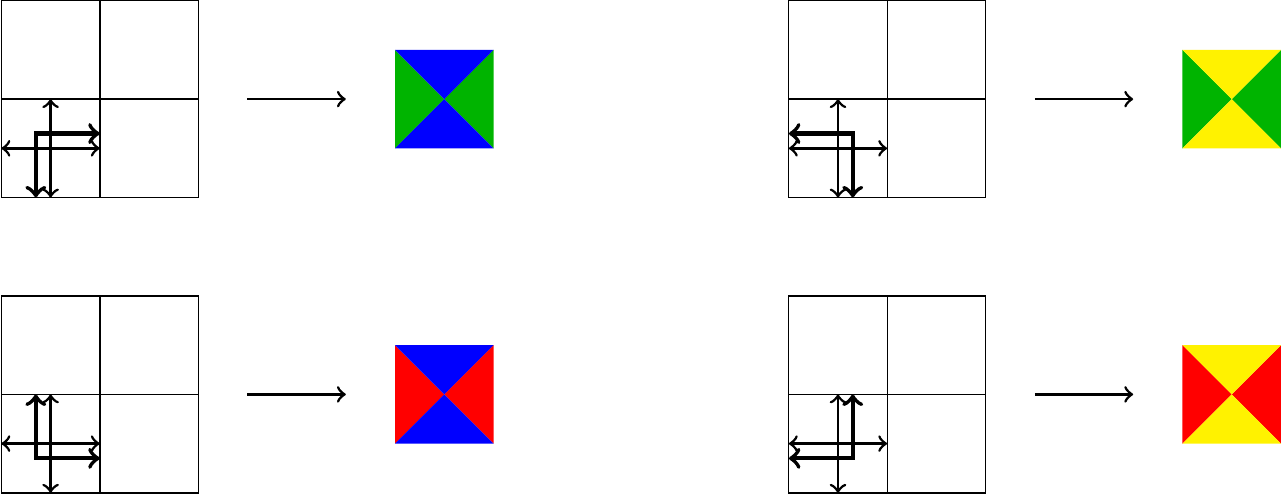}
	\end{center}
	\caption{}\label{fig:crosses_parity}
\end{figure}

We have verified in a \textit{Mathematica notebook}\footnote{The notebook is included in the supplementary material of the arXiv submission} that, under the map in \cref{def:renormalization_map}, the adjacency relations of the Robinson tiles and the ones of a subset of 56 allowed supertiles are equivalent, which proves \cref{thm:tiles_graph_isomorphism}. Refer to \cref{appendix:mathematica} for more details.

\medskip

Under this projection the supertiles that do not appear in the tiling of the plane, illustrated in~\cref{sec:not_appearing}, are not mapped to any Robinson tile. The reason for this is that there does not exist a Robinson tile with matching of arrows and parity: the supertiles of type 1 in \cref{Fig:Two_Structures} would be mapped to vertical arms with horizontal parity, and conversely type 2 supertiles would correspond to horizontal arms with vertical parity.

% ========================================================================
\subsection{Allowed but not appearing supertiles}\label{sec:not_appearing}
% ========================================================================

Consider the subset of supertiles allowed by the adjacency rules which have a $1\times 1$ parity cross in the bottom-left.
There are 68 such tiles, however, there are only 56 Robinson tiles.
The result of this is that 12 tiles cannot be mapped under the renormalisation procedure.
These have two distinct structures, as shown in \cref{Fig:Two_Structures}, where we have used the abbreviated notation used in \cite{robinson1971undecidability}, indicating only the direction of the arms.

\begin{figure}[h]
  \begin{center}
    \includegraphics[width=0.8\textwidth]{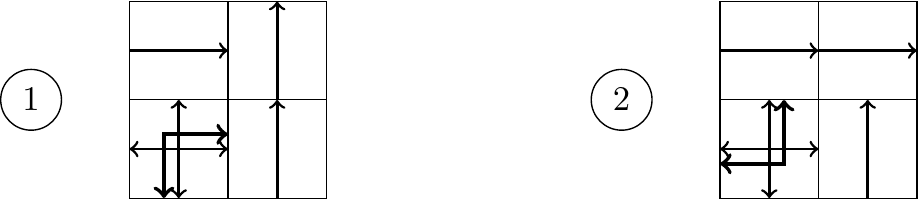}
  \end{center}
  \caption{}
  \label{Fig:Two_Structures}
\end{figure}

For each of these two structures, we have 6 possible combinations.
Those supertiles cannot appear in any tiling of the plane.
Consider the structure 1.
By imposing the parity rules for supertiles that we described previously, a supertile of type~1 there must have one of the parity cross supertiles above it, which has the  (abbreviated) form shown in \cref{Fig:Single_Tile}.
Clearly, no supertile with the structure 1 can be placed below a parity cross supertile because of the arrow rules.
Analogously, by parity rules, on the right of a supertile with structure 2 must lie a parity cross supertile.
Again, it is clear that this is not allowed by the arrow rules.

\begin{figure}[h!]
  \begin{center}
    \includegraphics[width=0.2\textwidth]{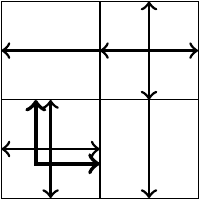}
  \end{center}
  \caption{}
  \label{Fig:Single_Tile}
\end{figure}

% ==========================================================
\subsection{Shifting the Grid}
% ==========================================================

The previous analysis was done by placing a $2\times2$ grid over the Robinson pattern on the plane with the parity cross lying on the bottom-left of each $2\times2$ cell.
Naturally, there are 4 possible ways that we could place our $2\times 2$ grid.
The same investigation has been performed for all other three cases when shifting the $2\times2$ grid one cell right, upwards, and diagonally, so that the parity cross will occupy the bottom-right, top-left and top-right position of the supertiles, respectively.
The analysis for these other settings is completely equivalent to the case we have discussed.
For any of the four positioning of the parity cross in the $2\times2$ grid, there exist 68 allowed supertiles, and a subset of 56 will tile the plane and build adjacency relations isomorphic to the Robinson tiles.

\begin{remark}
  The four possible placements of the $2\times2$ grid give rise to four completely disjoint sets of 68 supertiles each.
  This is a direct consequence of the different inner parity structure of the supertiles illustrated in~\cref{fig:inner_parity}.
\end{remark}
Let the set of possible renormalised supertiles, but restricted to those which actually appear, be denoted:
\begin{align}
  \mathcal{T}_1 \oplus \mathcal{T}_2 \oplus \mathcal{T}_3 \oplus \mathcal{T}_4.
\end{align}
which each $\mathcal{T}_i$ corresponding to a different supertile parity structure.
Then the set of supertiles that occurs for our choice of basis, $T_2'$, is equal to one of these four disjoint sets.
Which set occurs depends on where the $2\times 2 $ grid is placed.
% with respect to the tiling pattern.

% ==========================================================
\subsection{Renormalising the Classical Hamiltonian}
% ==========================================================
We are now in a position to show that there is an RG transformation on the tiling Hamiltonian which preserves the ground state.
In terms of the Hamiltonian, the RG scheme takes the form of restricting to sets of allowed $2\times 2$ blocks, and then applying an isometry mapping these $2\times 2$ blocks to new supertiles belonging to the set $T_2'$.

\paragraph{The Initial Tiling Hamiltonian} ~\newline
Let $h_T^{row}=\sum_{(t_i,t_j)\not\in A_1^H}\ket{t_it_j}\bra{t_it_j}$ and $h_T^{row}=\sum_{(t_i,t_j)\not\in A_1^V}\ket{t_it_j}\bra{t_it_j}$ be the local interaction terms of the tiling Hamiltonian, where $A_1^H$ and $A_1^V$ are, respectively, the horizontal and vertical adjacency rules for tiles in $T_1$.
Then the ground state of $H=\sum_{i\in\Lambda( L)}h_{T,i,i+1}^{row}+ \sum_{j\in \Lambda( L)}h_{T,j,j+1}^{col}$ has a zero energy ground state which is given by the tiling of the plane according to the Robinson pattern.

We now consider the RG scheme for the Hamiltonian:
\begin{definition}[Tiling Renormalisation Isometry]
  Let $T_2'$ be one of the disjoint subsets of $2\times 2 $ Robinson tiles which appear in the previously described renormalisation scheme.
  Let $V_{(i,i+1),(j,j+1)}: T_1^{\otimes 2\times 2} \rightarrow T_2'$ be the isometry mapping these $2 \times 2 $ blocks to allowed supertiles which appear in the Robinson pattern,
  \begin{align}
    V_{(i,i+1),(j,j+1)} = \sum_{\ket{T_\alpha} \in T_2'} \ket{T_\alpha}_{I,J}\bra{t_a}_{i,j}\bra{t_b}_{i+1,j}\bra{t_k}_{i,j+1} \bra{t_l}_{i+1,j+1},
  \end{align}
  where $\ket{t_m}\in T_1$, with the set of $T_2'$ tiles that maps to the Robinson tiles as described in \cref{def:renormalization_map}.
\end{definition}

\begin{definition}[Tiling Hamiltonian Renormalisation] \label{Def:Tiling_Hamiltonian_RG}
  Let $h_T^{col}, h_T^{row}\in \B(\C^T\otimes \C^T)$ be the local interactions describing the tiling Hamiltonian.
  Let $h_{i,i+1}^{row}(j)$ denote the row interaction between sites $(i,j),(i+1,j)$ and similarly let $h_{j,j+1}^{col}(i)$ be the interaction between $(i,j),(i,j+1)$.
  Let the $2\times 2$ supertiles be assigned at $(i,j), (i+1,j), (i,j+1), (i+1,j+1)$ and sites consistent with it.
  Then the renormalised Hamiltonian has local terms $\r(h_T^{col}),\r(h_T^{col})\in \mathcal{B}(\C^T \otimes \C^T)$.
  \begin{adjustwidth}{-2cm}{-2cm}
    \begin{align}
      \r(h_T^{col})_{\lceil j/2-1\rceil,\lceil j/2-1\rceil+1} &= V_{(i,i+1),(j+2,j+3)}V_{(i,i+1),(j,j+1)}\left( h_{T,j+1,j+2}^{col}(i) + h_{T,j+1,j+2}^{col}(i+1) \right) \big|_{T_2'} \\
                                                              &\times V_{(i,i+1),(j,j+1)}^\dagger V_{(i,i+1),(j+2,j+3)}^\dagger \\
      \r(h_T^{row})_{\lceil i/2-1\rceil,\lceil i/2-1\rceil+1} &= V_{(i+2,i+3),(j,j+1)}V_{(i,i+1),(j,j+1)}\left(h_{T,i+1, i+2}^{row}(j)+h_{T,i+1, i+2}^{row}(j+1)\right) \big|_{T_2'} \\
      &\times V_{(i,i+1),(j,j+1)}^\dagger V_{(i+2,i+3),(j,j+1)}^\dagger
    \end{align}
  \end{adjustwidth}
  In the above we have used the standard abbreviation that each local term is implicitly tensored with the appropriate identity terms, e.g.  $h_{T,j+1,j+2}^{col}(i)$ is actually  $\1_{i,j}\otimes  \1_{i+1,j}\otimes h_{T,j+1,j+2}^{col}(i)  \otimes \1_{i,j+3} \otimes \1_{i+1,j+3}$.
\end{definition}

Note that this renormalisation map is computable, as each $V$ simply describes the mapping of tiles in the initial set to those in the new set in the way illustrated previously.

\medskip

We now prove that the local Hamiltonian terms are mapped back onto themselves when this RG transformation is applied.

\begin{lemma}\label{Lemma:Classical_RG_Hamiltonian}
  The matrix form of the initial and renormalised Hamiltonian are the same, i.e.,
  \begin{equation}
    \r(h_T^{row})_{i,i+1} = h_{T,i,i+1}^{row}
    \qquad \text{and} \qquad
    \r(h_T^{col})_{j,j+1} = h_{T,j,j+1}^{col}.
  \end{equation}
\end{lemma}
\begin{proof}
  We consider two neighbouring $2\times 2$ blocks $(i,j), (i+1,j), (i,j+1), (i+1,j+1)$ and $(i+2,j), (i+3,j), (i+2,j+1), (i+3,j+1)$, and determine how the row and column interactions transform under this renormalisation process.
  We can then write
  \begin{align}
    h_{i,i+1}^{row}(j)=\sum_{(t_k,t_l\in H)}^{}\ket{t_k}_{i,j}\ket{t_l}_{i+1,j} \bra{t_k}_{i,j}\bra{t_l}_{i+1,j}
  \end{align}
  and, with $\ket{T_\alpha}\in T_2'$, then
  \begin{align}
    V_{(i,i+1),(j,j+1)} = \sum_{\ket{T_\alpha} \in T_2'} \ket{T_\alpha}\bra{t_a}_{i,j}\bra{t_b}_{i+1,j}\bra{t_k}_{i,j+1} \bra{t_l}_{i+1,j+1},
  \end{align}
  where $\ket{t_a}_{i,j}\ket{t_b}_{i+1,j}\ket{t_k}_{i,j+1}\ket{t_l}_{i+1,j+1}$ is an allowed $2\times 2$ supertile, and we sum over all such allowed $2\times 2 $ blocks.

  Now consider the two blocks: there are 6 relevant row interactions:
  \begin{align}
    &h_{i,i+1}^{row}(j)+ h_{i,i+1}^{row}(j+1) \\
    + &h_{i+1,i+2}^{row}(j)+ h_{i+1,i+2}^{row}(j+1) \\
    + &h_{i+2,i+3}^{row}(j) + h_{i+2,i+3}^{row}(j+1).
  \end{align}
  Restrict to the set of appearing $2\times 2 $ supertiles, $T_2'$, which are centred on the  $2\times 2$ blocks $(i,j), (i+1,j), (i,j+1), (i+1,j+1)$ and $(i+2,j), (i+3,j), (i+2,j+1), (i+3,j+1)$.
  In this case we see that by enforcing only allowed supertiles, then  $(h_{i,i+1}^{row}(j)+ h_{i,i+1}^{row}(j+1))|_{T_2'}=0$ and $ (h_{i+2,i+3}^{row}(j) + h_{i+2,i+3}^{row}(j+1))|_{T_2'}=0$.

  Finally we need to consider the terms
  \begin{align}
    &V_{(i+2,i+3),(j,j+1)}V_{(i,i+1),(j,j+1)}(h_{i+1,i+2}^{row}(j)+ h_{i+1,i+2}^{row}(j+1))|_{T_2'} \\
    &\times V_{(i,i+1),(j,j+1)}^\dagger V_{(i+2,i+3),(j,j+1)}^\dagger.
  \end{align}
  The application of the isometries maps the tiles to supertiles.
  Hence we can write
  \begin{adjustwidth}{-2cm}{-2cm}
    \begin{align}
      \r(h^{row})_{i/2,i/2+1}(j) = &V_{(i,i+1),(j,j+1)} V_{(i+2,i+3),(j,j+1)}\left(h_{i+1,i+2}^{row}(j)+ h_{i+1,i+2}^{row}(j+1)\right)|_{T_2'} \times \\
                                  &\times V_{(i,i+1),(j,j+1)}^\dagger V_{(i+2,i+3),(j,j+1)}^\dagger.
    \end{align}
  \end{adjustwidth}
  Note that $\r(h^{row})_{i,i+1}(j)$ acts on $T_2'$ and we see that there is an energy assigned to a particular term in $\r(h^{row})_{i,i+1}(j)$ iff there is a corresponding term in $h^{row}_{i,i+1}$.
  Furthermore $\r(h^{row})_{i,i+1}(j)$ is the same for all $j$, hence $\r(h^{row})_{i,i+1}=h_{i,i+1}^{row}$.

\end{proof}

\begin{corollary}
  The Hamiltonian with local terms $\Rk(h_T^{row}),\Rk(h_T^{col})\in \mathcal{B}(\C^T \otimes \C^T)$, has the same ground state energy and excited state energies as the unrenormalised Hamiltonian, for any $k\geq 0$.
\end{corollary}

% ==========================================================
% ==========================================================
\section{Renormalisation of the Quantum Hilbert Space} \label{Sec:Quantum_RG}
% ==========================================================
% ==========================================================

In this section we will deal with the renormalisation of the quantum Hamiltonian. For this, we will need a number of definitions from~\cite{Cubitt_PG_Wolf_Undecidability}.

\begin{definition}[Standard Basis States]
  Let the single site Hilbert space be $\mathcal{H}=\otimes_i \mathcal{H}_i$ and fix some orthonormal basis for the single site Hilbert space.
  Label the set of single site basis states for site $i$ as $\frk{B}^{(i)}_q$.
  Then a \emph{standard basis state} for $\mathcal{H}^{\otimes L}$ are product states over the single site basis.
\end{definition}

\begin{definition}[Penalty Terms and Transition Rules]
  The two-local quantum Hamiltonian will contain two types of terms: \emph{penalty terms} and \emph{transition rule} terms. Penalty terms have the form $\dyad{ab}{ab}$ where $\ket{a}$ and $\ket{b}$ are standard basis states. This adds a positive energy contribution to any configuration containing the state $\ket{ab}$, which we call an \emph{illegal pair}.
  Transition rule terms take the form $\frac{1}{2} (\ket{ab}-\ket{cd})(\bra{ab}-\bra{cd})$ with $\ket{ab} \neq \ket{cd}$, where $\ket{ab}$ and $\ket{cd}$ act on the same pair of adjacent sites.
\end{definition}

\begin{definition}[Legal and Illegal States]
  We call a standard basis state \emph{legal} if it
  does not contain any illegal pairs, and \emph{illegal} otherwise
\end{definition}

% -------------------------------------------------------------------
We then define a standard form Hamiltonian on the joint system
\begin{equation}
  \mathcal{H}_C\otimes\mathcal{H}_Q \coloneqq (\mathds{C}^C\otimes\mathds{C}^Q)^{\otimes L} = (\mathds{C}^C)^{\otimes L}\otimes(\mathds{C}^Q)^{\otimes L}.
\end{equation}

\begin{definition}[Standard-Form Hamiltonian~\cite{Cubitt_PG_Wolf_Undecidability,Watson_Hamiltonian_Analysis}]
  \label{Def:Standard-form_H}
  We say that a Hamiltonian $H = H_{trans} + H_{pen} + H_{in} + H_{out}$ acting on $\mathcal{H}_C\otimes\mathcal{H}_Q$ is of \emph{standard form} if it takes the form
  \begin{equation}
    H_{trans,pen, in, out} = \sum_{i=1}^{L-1} h_{trans,pen, in, out}^{(i,i+1)}
  \end{equation}
  where the local interactions $h_{trans,pen, in, out}$ satisfy the following conditions:
  \begin{enumerate}
  \item $h_{trans} \in \mathcal{B}\left((\mathds{C}^C\otimes\mathds{C}^Q)^{\otimes 2}\right)$ is a sum of transition rule terms, where all the transition rules act diagonally on $\mathds{C}^C\otimes\mathds{C}^C$ in the following sense. Given standard basis states $a,b,c,d\in\mathds{C}^C$, exactly one of the following holds:
    \begin{itemize}
    \item there is no transition from $ab$ to $cd$ at all; or
    \item $a,b,c,d\in\mathds{C}^C$ and there exists a unitary $U_{abcd}$ acting on $\mathds{C}^Q\otimes\mathds{C}^Q$ together with an orthonormal basis $\{\ket{\psi_{abcd}^i}\}_i$ for $\mathds{C}^Q\otimes\mathds{C}^Q$, both depending only on $a,b,c,d$, such that the transition rules from $ab$ to $cd$ appearing in $h_{trans}$ are exactly $\ket{ab}\ket{\psi^i_{abcd}}\rightarrow \ket{cd}U_{abcd}\ket{\psi^i_{abcd}}$ for all $i$. There is then a corresponding term in the Hamiltonian of the form
      $(\ket{cd}\otimes U_{abcd} - \ket{ab})(\bra{cd}\otimes U_{abcd}^\dagger - \bra{ab})$.
    \end{itemize}
    \label{standard-form_H:transition_terms}
  \item $h_{pen} \in \mathcal{B}\left((\mathds{C}^C\otimes\mathds{C}^Q)^{\otimes 2}\right)$ is a sum of penalty terms which act non-trivially only on $(\mathds{C}^C)^{\otimes 2}$ and are diagonal in the standard basis, such that $h_{pen} = \sum_{ab \ \mathrm{illegal}  } \ket{ab}\bra{ab}_C \otimes \mathds{1}_{Q}$, where $\ket{ab}$ are members of a disallowed/illegal subspace.
    \label{standard-form_H:penalty_terms}
  \item $h_{in}=\sum_{ab} \ket{ab}\bra{ab}_C \otimes \Pi_{ab}$, where $\ket{ab}\bra{ab}_C \in (\mathds{C}^C)^{\otimes2}$ is a projector onto $(\mathds{C}^C)^{\otimes 2} $ basis states, and $\Pi_{ab}^{(in)} \in (\mathds{C}^{Q})^{\otimes 2}$ are orthogonal projectors onto $(\mathds{C}^{Q})^{\otimes 2}$ basis states.
  \item $h_{out}= \ket{xy}\bra{xy}_C \otimes \Pi_{xy}$, where $\ket{xy}\bra{xy}_C \in (\mathds{C}^C)^{\otimes2}$ is a projector onto $(\mathds{C}^C)^{\otimes 2} $ basis states, and $\Pi_{xy}^{(in)} \in (\mathds{C}^{Q})^{\otimes 2}$ are orthogonal projectors onto $(\mathds{C}^{Q})^{\otimes 2}$ basis states.
  \end{enumerate}
\end{definition}

Importantly the Gottesman-Irani Hamiltonian we will be considering will be of standard form.

\bigskip

The 1D Gottesman-Irani Hamiltonian$H_q(L)\in \B(\C^d)^{\otimes L}$ is a standard-form Hamiltonian according to the above definition, and is given by
\begin{align}
  H_q = H_{trans} + H_{in} + H_{pen}+H_{halt},
\end{align}
where $H_{trans}$ contains transition rule terms, $H_{pen}$ is a set of penalty terms which penalise states that should not appear in correct history states, $H_{in}$ penalises states which are incorrectly initialised, and $H_{halt}$ penalises states which encode a halting computation. Moreover, it has a six-fold tensor product form
\begin{equation}\label{6_tracks_TP_Hamiltonian}
  \HS_q = \bigotimes_{j=1}^6 (\HS_q)_{j}.
\end{equation}
where each $(\HS_q)_{j}$ is identified with a different track.\\

Lemma 43 of \cite{Cubitt_PG_Wolf_Undecidability} identifies three subspaces of states, which are closed under the action of $H_q$.
\begin{enumerate}
\item \textbf{Illegal Subspace}, $\calS_1$: All $\ket{x}\in \calS_1\subset \frk{B}^{\ox L}$ are in the support of $H_{pen}$ and hence $\bra{x}H\ket{x}\geq 1$.
 By \cite{Cubitt_PG_Wolf_Undecidability} Lemma 43, the minimum eigenvalue of these subspaces is
  \begin{align}
    \lambda_0(H|_{\calS_1})\geq 1.
  \end{align}
\item \textbf{Evolve-to-Illegal Subspace}, $\calS_2$: All standard basis states $\ket{x}\in \calS_2\subset \frk{B}^{\ox L}$ will evolve either forwards or backwards in time to an illegal state in $O(L^2)$ steps under the transition rules.
  As per lemma 5.8 of \cite{Watson_Hamiltonian_Analysis}, the minimum eigenvalue of these subspaces is
  \begin{align}
    \lambda_0(H|_{\calS_2})= \Omega(L^{-2}).
  \end{align}
\item \textbf{Legal Subspace}, $\calS_3$: all standard basis states in $\calS_3$ are legal and \emph{do not} evolve to illegal states.
By \cite{Cubitt_PG_Wolf_Undecidability} lemma 43, they have zero support on $ H_{pen}$ or $ H_{in}$.
\end{enumerate}

In our renormalisation procedure we seek to preserve only the low energy subspace, hence at any point where we can locally identify states as being in subspace $\calS_1$ or $\calS_2$, we will remove them from the state space in the renormalisation step.

However, we note that in the general case we cannot locally identify all such states in $\calS_2$.
That is, determining the whether a state evolves to an illegal under the action of the transitions may be impossible if we only look at what the state looks like on a $O(1)$-subset of the sites.

% ====================================================================
\subsubsection{The Ground States}
% ====================================================================
From \cite{Cubitt_PG_Wolf_Undecidability} we know that there are two cases we need to consider: the QTM encoded in $H_q(L)$ halts or does not halt.
\begin{lemma} \label{Lemma:Ground_State_GI_Form}
 % $H_q(L)$ has a ground state energy that is either $0$ or $1- \cos(\frac{\pi}{2T})$.
  Let a given UTM be encoded in the Gottesman-Irani Hamiltonian $H_q(L)$.
  Then $H_q(L)$ has a ground state energy that is either $0$ if the UTM does not halt within time $T(L)$ or $1- \cos(\frac{\pi}{2T})$ if the UTM does halt within $T(L)$.
  $T(L)$ is a fixed, predetermined function.
  In the non-halting case, the ground state is
  \begin{align}
    \ket{\Psi_{hist}(L)}=\frac{1}{\sqrt{T}}\sum_{t=1}^{T(L)}\ket{t}\ket{\psi_t},
  \end{align}
  and in the halting case it is
  \begin{align}
    \ket{\Psi_{halt}(L)}=\sum_{t=1}^{T(L)}2\cos(\frac{(2t+1)\pi t }{4T})\sin(\frac{\pi }{4T})\ket{t}\ket{\psi_t},
  \end{align}
  where $\ket{t}$ is the state of the clock register and $\ket{\psi_t}=\prod_{j=1}^t U_j\ket{\psi_0}$ and $\ket{\psi_0}$ is the initial state of the computational register and the $\{U_t\}$ represent the action of the QTM at time step $t$.
\end{lemma}
\begin{proof}
  Combine the standard form property of $H_q$ from \cite{Cubitt_PG_Wolf_Undecidability} with Lemma 5.10 of \cite{Watson_Hamiltonian_Analysis}.
\end{proof}

% ========================================================================
\subsection{Block Renormalisation of the Gottesman-Irani Hamiltonian} \label{Sec:GI_Blocking}
% ========================================================================
In this section we will construct a renormalisation scheme for the Gottesman-Irani Hamiltonian.
For a given spin at site $i$, we write each possible conventional basis state (i.e. basis state before the RG procedure has started) as
$\ket{\twocellsvert{a}{\alpha}}_{(i)}\in \C^C\otimes \C^Q $, where the top cell indicates the classical tracks of the construction encoded in \cite{Cubitt_PG_Wolf_Undecidability}, while the bottom cell indicates the quantum tracks (see~\cref{sec:Gottesman_Irani_Hamiltonian}).

We then define a pair of operations: the blocking operation $\B_q$ and the truncation operation $\T_q$.
Given a line of qudits $\B_q$ will essentially combine two lattice sites into a single site with a larger local Hilbert space dimension, while $\T_q$ will remove any of the new single site states which can be locally detected to have non-zero overlap with the ground state.
Thus $\T_q$ reduces the local Hilbert space dimension.

We note that we do not truncate all high energy states since in the halting case this would remove the ground state of the Gottesman-Irani Hamiltonian.
Instead, we removed states based on a combination of high energy and a priori knowledge of the ground state.

\paragraph{Blocking $\B_q$} ~\newline
The blocking part of the renormalisation procedure is defined as follows.
\begin{definition}[Gottesman-Irani Blocking, $\B_q$] \label{Def:Gottesman-Irani_Blocking}
  Let $\ket{\psi}\in \HS^{(i)}_q\ox \HS^{(i+1)}_q$, $i\in \N$.
  The blocking operation, $\B_q : \HS_q^{(i)} \times \HS_q^{(i+1)} \rightarrow \HS_q^{\prime(i/2)}$, is given by the action of the unitary $U_{i,i+1}: \HS_q^{(i)} \times \HS_q^{(i+1)} \rightarrow \r(\HS_q)^{\prime}$ as
  \begin{align}
    \B_q^{(i,i+1)}: \ket{\psi} \mapsto U_{i,i+1}\ket{\psi}
  \end{align}
  where
  \begin{align}
    U_{i,i+1}= \sum_{\ket{x},\ket{y}\in\frk{B} } \ket{xy}_{i/2}\bra{x}_i\bra{y}_{i+1}.
  \end{align}
  We extend this to $\ket{\chi}\in \HS_q^{\ox L}$ as
  \begin{align}
    \B_q: \ket{\chi} \mapsto U\ket{\chi},
  \end{align}
  where $U=\bigotimes_{i\in 2\N}^{i\leq L/2}U_{i,i+1}$.
\end{definition}

~\newline
This can be expressed more intuitively in terms of basis states
\begin{align}
  \B_q^{(i,i+1)}: \ket{\twocellsvert{a}{\alpha}}_{(i)} \otimes \ket{\twocellsvert{b}{\beta}}_{(i+1)}
  \longrightarrow \ket{\fourcells{  a }{ \alpha }{b}{\beta}}_{(i/2)}.
\end{align}
Note that $\B_q$ is just a relabelling of the space, so the local Hilbert space dimension is now $\C^{d^2}$ and part of the tensor product structure is lost.
We denote by  $\HS_q'$ this new local Hilbert space spanned by the basis $\frk{B}^{\prime(1)}$.\\

\paragraph{Truncation $\T_q$} ~\newline
The truncation part of the RG map truncates the local Hilbert space to discard those states which locally have support on the penalty terms.

\begin{definition}[Gottesman-Irani Truncation Mapping, $\T_q$] \label{Def:Gottesman-Irani_Truncation}
  Let $\frk{B}^{(1)}$ be the set of basis states defined by $\B_q$ such states with a preimage $\ket{a}\ket{b}$, such that $\ket{a},\ket{b}\in \frk{B}$ cannot be locally identified as being in subspace $\calS_1$ or $\calS_2$.
  That is
  \begin{align}
    \bra{a}\bra{b}h_{pen}^{i,i+1}\ket{a}\ket{b} = \bra{a}\bra{b}h_{in}^{i,i+1}\ket{a}\ket{b} &=0, \\
    \bra{a}\bra{b}h_{trans}^{(i,i+1)}h_{pen}^{(i,i+1)}h_{trans}^{(i,i+1)}\ket{a}\ket{b}&= 0.
  \end{align}
  The truncation mapping is then $\T_q^{(i,i+1)}: \r(\HS_q)^{\prime} \rightarrow \r(\HS_q)$ for $\r(\HS_q)=\spann\{\frk{B}_q^{(1)}\}\subset \r(\HS)_q^{\prime}$.
  Then the full restriction is $\T_q: \HS_q^{\prime\ox L/2} \rightarrow \r(\HS_q)^{\ox L/2}$.

\end{definition}

%The set of states removed is designed to cut out:
%\begin{itemize}
%\item States with two or more QTM heads.
%\item States with two or more oscillating counters.
%\item States which have a counter marker in the zero phase, but have a non-zero state on the QTM track.
%\item States which have a $\leftend$ or $\rightend$ marker appearing anywhere except the left or right end of the state, respectively.
%\item States which necessarily evolve (within their spacial location) to an illegal state.
%\end{itemize}

We now combine the unitary and subspace restriction to give an isometry which implements $\T_q\circ \B_q$.

\begin{lemma}[Renormalisation Unitary Structure] \label{Lemma:Renormalisation_Isometry}
  Let the renormalisation isometry $V^{GI}_{i,i+1}$ be the unitary map follow by subspace restriction previously described.
  Define $V^{GI}: \HS_q^{\otimes L}\rightarrow \r(\HS_q)^{\otimes L/2} $ to implement the mapping $\T_q \circ \B_q$ on a state in $\HS_q^{\ox L}$, as
  \begin{align}
  \T_q \circ \B_q: \ket{\chi} \mapsto  U\ket{\chi}|_{\r(\HS_q)^{\ox L/2}} =: V^{GI}\ket{\chi}.
  \end{align}
  where $U$ is defined in \cref{Def:Gottesman-Irani_Blocking} and $\r(\HS_q)$ is defined in \cref{Def:Gottesman-Irani_Truncation}.
  Then $V^{GI}$ can be defined as and decomposed as
  \begin{equation}\label{eq:Isometry}
    V^{GI} :=\bigotimes_{i\in 2\N}^{i\leq \lfloor L/2\rfloor}V^{GI}_{i,i+1}= \bigotimes_{i\in 2\N}^{i\leq \lfloor L/2\rfloor}\left( \bigotimes_{j=1}^6 V_{i,i+1}^{GI \ (j)}  \right),
  \end{equation}
  with
  \begin{equation}
    V^{GI}_{i,i+1}: \HS_q^{\otimes 2} \rightarrow \r(\HS_q)
  \end{equation}
  and where each part of the decomposition acts on one of the six different tracks,
  \begin{equation}
    V_{i,i+1}^{GI \ (j)}: \HS_{q,j}^{\otimes 2} \rightarrow \r(\HS_{q})_j.
  \end{equation}
\end{lemma}

\begin{proof}
  The decomposition $V^{GI}=\bigotimes_{i\in 2\N}^{i\leq \lfloor L/2\rfloor}V^{GI}_{i,i+1}$ is evident from the block procedure.
  The decomposition $V^{GI}_{i,i+1}=  \bigotimes_{j=1}^6 V_{i,i+1}^{GI \ (j)}$ arises from the fact that the procedure keeps each basis state as a product across the different tracks and hence the different $\HS_{q,j}$.
\end{proof}

We now need to define how the Hamiltonian acts with respect to the RG procedure.
We want to break down the Hamiltonian into different subspaces and renormalise them separately while preserving the ground state (in both the halting and non-halting cases) and its energy.

\begin{lemma}[Renormalised Gottesman-Irani Hamiltonian] \label{Lemma:RG_GI_Properties}
  Let $h_{q}$ be the local terms of a nearest neighbour, translationally invariant Hamiltonian
  \begin{equation}
    H_q(L)= \sum_{i=1}^{L}h_q^{(i,i+1)}=H_{trans}+H_{pen}+H_{in}+H_{out},
  \end{equation}
  such that $H(L)$ is standard form.
  Let $V:\C^d\otimes \C^d \rightarrow \C^{f(d)}$, be the isometry from~\cref{Lemma:Renormalisation_Isometry}.
  %	Define $\Gamma$ to be the subspace $\Gamma=\bigotimes_{i=1}^{L} \Gamma_{i} $, where $\Gamma_{i}=span \{ \ket{a}_i : \ket{a}\in\frk{B}^{(1)} \}$.
  Then the renormalised Hamiltonian, defined as
  \begin{equation}\label{eq:renormalized_GI_Hamiltonian}
    \sR(H_q(L))= V^{GI}H_q(L) V^{GI \dagger} = \sum_{i=1}^{L/2}V^{GI}h_q^{(i,i+1)} V^{GI\dagger}=\r(H_q)(L),
  \end{equation}
  is a translationally invariant, nearest-neighbour Hamiltonian with local interactions $\r(h_q)^{(i/2,i/2+1)}=V^{GI}(h_q^{(i-1,i)}+h_q^{(i+1,i+2)})V^{GI\dagger}$ and $\r(h_q)^{i/2}=V^{GI}h_q^{(i,i+1)}V^{GI\dagger}$.
  Furthermore, $\r(H_q)(L)$ has the following properties:
  \begin{enumerate}
  \item $\r(H_q)(L)$ is a standard form Hamiltonian. \label{RG-GI:Standard-Form}
  \item $\r(H_{trans})$ encodes a transition $V^{GI}(\ket{ab}\ket{\psi_{abcd}})\rightarrow V^{GI}(\ket{cd}U_{abcd}\ket{\psi_{abcd}})$ iff $H_{trans}$ encodes the transition $\ket{ab}\ket{\psi_{abcd}}\rightarrow \ket{cd}U_{abcd}\ket{\psi_{abcd}}$. \label{RG-GI:Transition_Rules}
  \item $\r(H_{pen}),\r(H_{in}),\r(H_{out})$ have support on a renormalised basis state $V^{GI}(\ket{ab}\ket{\psi})$ iff $H_{pen},H_{in},H_{out}$ respectively have non-zero support on $\ket{ab}\ket{\psi}$.\label{RG-GI:Penalty_Terms}
  \item $\lambda_0(H_q(L))=\lambda_0(\r(H_q)(L/2))$ (the ground state energy is preserved). \label{RG-GI:GS_Energy}
  \item $\r(H_q)$ maintains the six-fold tensor product structure of the original Hamiltonian $H_q$ in~\cref{6_tracks_TP_Hamiltonian}, that is,
    $\r(\HS_q) = \bigotimes_{j=1}^6 \r(\HS_q)_j$. \label{RG-GI:Ren_H_6_structure}
  \end{enumerate}
\end{lemma}

\begin{proof}

  First note that for all $i\in 2\N$, $V^{GI}_{i,i+1}h^{(i,i+1)} V_{i,i+1}^{GI\dagger} \in \mathcal{B}(\C^{f(d)})$ is now a $1$-local term in the new renormalised Hamiltonian.
  However, $V^{GI}_{i+2,i+3}V^{GI}_{i,i+1}h^{(i+1,i+2)} V_{i,i+1}^{GI\dagger} V_{i+2,i+3}^{GI\dagger}\in \mathcal{B}(\C^{f(d)}\otimes \C^{f(d)})$

  \paragraph{Claims~\ref{RG-GI:Standard-Form} and \ref{RG-GI:Transition_Rules}} ~\newline

  From the linearity of $V^{GI}$, we see that $\sR(H_q(L))=\r(H_{trans})+\r(H_{pen})+\r(H_{in})+\r(H_{out})$.
  It is trivial to see that $\r(H_{trans})=V^{GI}H_{trans}V^{GI\dagger}=\sum_{ab\rightarrow cd}(V^{GI}\ket{cd}\otimes U_{abcd} - V^{GI}\ket{ab})(\bra{cd}\otimes U_{abcd}^\dagger V^{GI\dagger} - \bra{ab}V^{GI\dagger})$, and hence encodes transitions between the renormalised states.
  This also shows $\r(H_{trans})$ satisfies Claim~\ref{RG-GI:Transition_Rules}.
  Due to the decompositional properties of $V^{GI}$, as shown in \cref{Lemma:Renormalisation_Isometry}, we preserve that $H_{trans}$ acts diagonally on the states in $\C^C$.
  Likewise, it preserves the form of $H_{pen},H_{in}, H_{out}$ as projectors onto a subset of states.
 % It should be noted that in the case $L=2$, we see that the state is renormalised to a history state and $H_{trans}=0$.

  \paragraph{Claim \ref{RG-GI:Penalty_Terms}:}
  Consider the penalty terms: given a renormalised state
  $V^{GI}\ket{\psi}$, it is clear that $$(\bra{\psi}V^{GI\dagger})V^{GI}H_{pen}V^{GI\dagger}  (V^{GI}\ket{\psi}) = \bra{\psi}H_{pen} \ket{\psi}=1,$$
  hence $V^{GI}\ket{\psi}$ is penalised by the renormalised Hamiltonian iff $\ket{\psi}$ is penalised by the unrenormalised Hamiltonian.
  The same applied to $H_{in}$ and $H_{out}$.

  \paragraph{Claim \ref{RG-GI:GS_Energy}:}
  First note that any state
  \begin{equation}
    \ket{\Psi\{a_t\}} = \sum_{t=1}^{\tau} a_t(\ket{t}\ket{\psi_t}).
  \end{equation}
  which encodes a valid evolution is in the kernel of $H_{in}, H_{pen}$, and is contained in subspace $\mathcal{S}_3$.
  Thus, $V^{GI}\ket{\Psi\{a_t\}}\in \r(\HS)^{\ox L/2}$, and after the RG procedure $\T_q \circ \B_q$ the corresponding renormalised state is
  \begin{equation}
    \ket{\Psi'\{a_t\}} = \sum_{t=1}^{\tau} a_tV^{GI}(\ket{t}\ket{\psi_t}).
  \end{equation}

  To see the energy of such states is preserved note
  \begin{equation}
    \bra{\Psi'\{a_t\}}V^{GI}H_q(L)V^{GI\dagger} \ket{\Psi'\{a_t\}} = \bra{\Psi\{a_t\}}H_q(L)\ket{\Psi\{a_t\}}.
  \end{equation}
  From \cref{Lemma:Ground_State_GI_Form} the ground states are of the form $\ket{\Psi\{a_t\}}$.
  We know that the state $V^{GI}\ket{\Psi\{a_t\}}$ has the same energy.
  Since the minimum eigenvalue is given by
  \begin{align}
  \lambda_0(H_q(L)) &= \min_{x\in \HS_q^{\ox L} } \frac{\bra{x}H_q(L)\ket{x}}{\braket{x}{x}} \\
   &= \min_{x\in \HS_q^{\ox L} } \frac{\bra{x}U U^\dagger  H_q(L) U^\dagger U \ket{x}}{\bra{x}U^\dagger U\ket{x}} \label{Eq:GS_Line_2} \\
   &\leq \min_{\substack{x\in \HS_q^{\ox L}\\ V^{GI} \ket{x}\neq 0  } } \frac{\bra{x}V^{GI} V^{GI\dagger}  H_q(L) V^{GI\dagger} V^{GI} \ket{x}}{\bra{x}V^{GI\dagger} V^{GI}\ket{x}} \label{Eq:GS_Line_3} \\
   &= \lambda_0(\r(H_q)(L/2)),
  \end{align}
  where going from \cref{Eq:GS_Line_2} to \cref{Eq:GS_Line_3} we have used the fact that we have restricted the subspace to remove the states that are integrated out by $V^{GI}$.
  Since $\lambda_0(\r(H_q)(L/2))=\lambda_0(H_q(L/2))$, then we can confirm $V^{GI}\ket{\psi_{halt}}$ and $V^{GI}\ket{\psi_{hist}}$ are the appropriate ground states after the renormalisation procedure.

%  To see that the ground state is still the ground state (and still have the same energy), realise that restricting the local Hilbert space to $\r(\HS)_q^{\ox L/2}$ can only increase the energy of the eigenvalues, then this state must be the ground state if they were the ground state beforehand.
%  Since we know from \cref{Lemma:Ground_State_GI_Form} that the ground states are of this form, then the ground state energy is preserved, thus satisfying Claim~\ref{RG-GI:GS_Energy}.
  % then all states have the same energy, we are assured that the ground states in particular also have the same energy

  \paragraph{Claim~\ref{RG-GI:Ren_H_6_structure}:}
  The preservation of the structure in~\cref{6_tracks_TP_Hamiltonian} follows directly from the tensor product form of the isometry given in~\cref{eq:Isometry} applied according to the renormalisation method described by~\cref{eq:renormalized_GI_Hamiltonian}.

\end{proof}

% =====================================================================
\subsection{Multiple Iterations}
% =====================================================================
Consecutive steps of the RG procedure can be derived straightforwardly.
The Hilbert space obtained after $k$-th RG steps of can be constructed by induction
\begin{equation}
  (\T_q\circ \B_q)^{\circ (k)} =\T_q\circ \B_q  \circ (\T_q \circ \B_q)^{\circ (k-1)}
\end{equation}
We first combine two basis elements in the space $\frk{B}^{(k-1)}$ into a new state, i.e.,
\begin{equation*}
  \ket{
    \twocellsvert{a_1}{\alpha_1} \twocellsvert{a_2}{\alpha_2}
    \twocellsvert{\cdots}{\cdots} \twocellsvert{a_{2^{(k-1)}}}{\alpha_{2^{(k-1)}}}
  }
  \otimes
  \ket{
    \twocellsvert{b_1}{\beta_1} \twocellsvert{b_2}{\beta_2}
    \twocellsvert{\cdots}{\cdots} \twocellsvert{b_{2^{(k-1)}}}{\beta_{2^{(k-1)}}}
  }
  =
  \ket{
    \twocellsvert{a_1}{\alpha_1} \twocellsvert{a_2}{\alpha_2}
    \twocellsvert{\cdots}{\cdots} \twocellsvert{a_{2^{(k-1)}}}{\alpha_{2^{(k-1)}}}
    \twocellsvert{b_1}{\beta_1} \twocellsvert{b_2}{\beta_2}
    \twocellsvert{\cdots}{\cdots} \twocellsvert{b_{2^{(k-1)}}}{\beta_{2^{(k-1)}}}
  }
\end{equation*}
We then truncate the basis set according to the criteria described in the previous section.
This will generate the set of renormalised local basis states $\frk{B}^{(k)}$.
The local Hilbert space after $k$ RG iterations is denoted by $\Rk(\HS)$.
We note that this can still be decomposed it as $\Rk(\HS)=\bigotimes_{i=1}^6\Rk(\HS)_i$ corresponding to the 6 tracks of the original construction.\\

We can thus concatenate multiple renormalisations of the Gottesman-Irani Hamiltonian in one isometry, $V^{GI}(k):\r^{(k-1)}(\HS_q)^{\ox 2L}\rightarrow\Rk(\HS_q)^{\ox L}$, given by
\begin{equation}
  V^{GI}[k] = \Pi_{j=1}^k V^{GI}_{L/2^j}
\end{equation}
where $V^{GI}_{L/2^j}$ is the isometry outlined in~\cref{Lemma:Renormalisation_Isometry}, but now acting on the appropriate local Hilbert space, and the subscript $L/2^j$ indicates that the operator is acting on a 1D chain of $L/2^j$ sites.
We note the use of square brackets $[]$ is to distinguish the isometry from $V^{GI}(j)$ which will denote the isometry $V^{GI}$ acting on the $j^{th}$ row of a $2\times 2$ lattice.

Accordingly, the renormalised Hamiltonian is then
\begin{equation}
  \Rk(H_q(L))
  =
  V^{GI}[k] H_q(L) V^{GI\dagger}[k].
\end{equation}
% where $\Gamma^{(k)}$ is the subspace defined as in~\cref{Lemma:RG_GI_Properties} when taking the span over elements of $\frk{B}^{(k)}$.\\

It follows immediately from \cref{Lemma:RG_GI_Properties} is that this RG mapping takes standard form Hamiltonians to standard form Hamiltonians while preserving the energy of the ground state.
Thus:

\begin{corollary}
  Multiple iterations of the RG map applied to $H_q(L)$ preserve the properties (1-5) in~\cref{Lemma:RG_GI_Properties}.
\end{corollary}

% ===================================================================
% ===================================================================
\section{Putting it all Together}\label{sec:all_together}
% ===================================================================
% ===================================================================
In this section we combine the renormalisation group schemes for the separate parts of the Hamiltonian.
First recall Lemma 51 of \cite{Cubitt_PG_Wolf_Undecidability} which characterises the ground state of the Hamiltonian defined by the local terms $h_u$:

\begin{lemma}[Tiling + quantum layers, Lemma 51 of \cite{Cubitt_PG_Wolf_Undecidability}]\label{Lemma:Tiling+Quantum}
  Let $h_c^{\mathrm{row}},h_c^{\mathrm{col}}\in\B(\C^C\ox\C^C)$ be the local interactions of a 2D tiling Hamiltonian $H_c$, with two distinguished states (tiles) $\ket{L},\ket{R}\in\C^C$.
  Let $h_q\in\B(\C^Q\ox\C^Q)$ be the local interaction of a Gottesman-Irani Hamiltonian $H_q(r)$, as in \cref{Sec:Quantum_RG}.
  Then there is a Hamiltonian on a 2D square lattice with nearest-neighbour interactions $h_u^\mathrm{row},h_u^\mathrm{col}\in\B(\C^{C+Q+1}\ox\C^{C+Q+1})$ with the following properties: For any region of the lattice, the restriction of the Hamiltonian to that region has an eigenbasis of the form $\ket{T}_c\ox\ket{\psi}_q$, where $\ket{T}_c$ is a \emph{product} state representing a classical configuration of tiles. Furthermore, for any given $\ket{T}_c$, the lowest energy choice for $\ket{\psi}_q$ consists of ground states of $H_q(r)$ on segments between sites in which $\ket{T}_q$ contains an $\ket{L}$ and an $\ket{R}$, a 0-energy eigenstate on segments between an $\ket{L}$ or $\ket{R}$ and the boundary of the region, and $\ket{e}$'s everywhere else.
\end{lemma}

The $\ket{L}$ and $\ket{R}$ tiles are identified in \cite{Cubitt_PG_Wolf_Undecidability} with the right-down and left-down red cross in the Robinson tiles respectively (see \cref{Sec:Robinson_RG}).
The ground state can then be shown to be the ground state of the Robinson tiling Hamiltonian plus a ``quantum layer'' in which the Gottesman-Irani ground states appear only over the tops of the Robinson squares.
Everywhere else in the quantum layer is a filler state $\ket{e}$.

A key point is that the eigenstates are all product states across $\HS_c$ and $\HS_{eq}$.
We wish for the RG mapping to preserve this property.
This restricts the type of renormalisation isometries we use, as detailed in the following lemma.
\begin{lemma}[Separable Eigenstates]\label{Lemma:Separable_Eigenstates}
  Let $H_u^{\Lambda(2L)}$ denote the Hamiltonian in \cref{Lemma:Tiling+Quantum}.
  Then for an isometry $Z=Z_c\ox Z_{eq}$ where $Z_c: \HS_c^{\ox 2\times 2} \rightarrow \r(\HS_{c})$ and $Z_{eq}: \HS_{eq}^{\ox 2\times 2} \rightarrow \r(\HS_{eq})$, the operator $ZH_u^{\Lambda(2L)}Z^\dagger$ also has eigenstates of the form $\ket{T'}_{c}\ox \ket{\psi}_{eq}$ for $\ket{T'}_{c}\in \r(\HS_c)^{\ox \Lambda(L)}$ and $\ket{\psi}_{eq}\in \r(\HS_{eq})^{\ox \Lambda(L)}$.
\end{lemma}
\begin{proof}
  As per \cref{Lemma:Tiling+Quantum}, the eigenstates of $H_u^{\Lambda(2L)}$ decompose as product states $\ket{T_c}\ox\ket{\psi_i}_{eq}$, hence we can write
  \begin{align}
    H_u^{\Lambda(2L)} = \sum_i \lambda_i \ketbra{T_i}\ox\ketbra{\psi_i}_{eq}.
  \end{align}
  Applying the renormalisation isometry $Z$ gives
  \begin{align}
    ZH_u^{\Lambda(2L)}Z^{\dagger} &= \sum_i \lambda_i Z_C\ketbra{T_i}_cZ_C^{\dagger}\ox Z_{eq}  \ketbra{\psi_i}_{eq}Z_{eq}^{\dagger}\\
                                 &=: \sum_i \lambda_i \ketbra{T'_i}_{c'}\ox \ketbra{\psi'_i}_{eq'}.
  \end{align}
  Thus the product structure across the two subspaces is preserved.
\end{proof}

In \cref{Sec:Robinson_RG} we showed that the ground state of the renormalised tiling Hamiltonian preserves the tiling pattern of the unrenormalised Hamiltonian.
Here we show that renormalising the full Hamiltonian preserves this Robinson tiling plus Gottesman-Irani ground state structure.

We start by considering how to renormalise the Gottesman-Irani Hamiltonian in the presence of filler states on a 2D lattice (as opposed to the 1D chain considered previously).
After this we show the ground state energy Hamiltonian is preserved under the RG map.

%======================================================================
\subsection{Renormalising $\HS_T\otimes (\HS_e \oplus \HS_q)$}\label{sec:renormlaising_3_H}
%======================================================================
%In this section we will implement an RG scheme on the classical and quantum parts of the Hamiltonian together by applying blocking an truncation operations to the Hamiltonian.
%We will then go on to show that this preserves certain key properties about the ground state and the spectral gap.

From \cref{Lemma:Separable_Eigenstates}, we know the eigenstates of the Hamiltonian defined by $h_u$ are product states across the classical-quantum Hilbert space partition and this structure is preserved under a tensor product of isometries on the two subspace separately.
Thus we can consider the basis states of $\HS_T$ and $\HS_{eq}$ separately and then later show this preserves the desired properties.

\paragraph{Blocking Operation $\B_u$}

% We now restrict to $\kappa_{ij}$ and
%We now choose the action of the renormalising isometries and later show they preserve the necessary properties.
We know that $V^C$ from \cref{Lemma:Classical_RG_Hamiltonian} will renormalise the classical state space by mapping sets of $2\times 2$ tiles to new tiles which recreate the tiling pattern at all but the lowest level.
We use this isometry unchanged, acting on the classical part of the Hilbert space.

Consider the quantum Hilbert space $\HS_{eq}$.
First note that the Gottesman-Irani Hamiltonian to be renormalised is a standard form Hamiltonian, and so can be renormalised as per \cref{Sec:GI_Blocking}.
However, the blocking procedure from \cref{Sec:GI_Blocking} is not sufficient for our purposes as it (a) takes a set of $2\times 1 $ lattice sites to a single lattice site and so is not appropriate for a 2D lattice, and (b) does not include the filler state $\ket{e}_e$.
To remedy this we need an isometry which acts as:
\begin{align}
  V^{eq}_{(i,i+1)(j,j+1)}: \HS_{eq}^{(i,j)} \ox \HS_{eq}^{(i+1,j)} \ox \HS_{eq}^{(i,j+1)} \ox \HS_{eq}^{(i+1,j+1)} \rightarrow  (\HS_{eq}'\ox\HS_{eq}')^{(i/2,j/2)}.
\end{align}
We will find it useful to define the following notation:
 \begin{definition}[$k$-times Blocked Basis States]
	Let $\ket{x_1},\ket{x_2},\dots,\ket{x_{2^{k}}}\in \frk{B}\cup \ket{e}_e$, then we denote the corresponding renormalised basis state after $k$ applications of the RG mapping as $\ket{x_1x_2\dots x_{2^k}}$.
\end{definition}
Now define $V^q_{(i,i+1)}(j)$ as follows, where $V^{GI}_{i,i+1}$ is the isometry used in \cref{Lemma:RG_GI_Properties}:
\begin{align}
  V^q_{(i,i+1)}(j) = V^{GI}_{i,i+1} &+ \ket{ee}_{i/2,j/2}\bra{e}_{i,j}\bra{e}_{i+1,j} \\
                                    &+ \ket{xe}_{i/2,j/2}\bra{x}_{i,j}\bra{e}_{i+1,j} + \ket{ex}_{i/2,j/2}\bra{e}_{i,j}\bra{x}_{i+1,j}.
\end{align}
This defines a new set of quantum basis states which now reflect the fact $\ket{e}_e$ is part of the Hilbert space.
Denote this
\begin{align}
  \frk{C}^{(1)}:=\frk{B}^{(1)}\cup \ket{ee}\bigcup_{x\in \frk{B}} \ket{ex} \bigcup_{x\in \frk{B}} \ket{xe}.
\end{align}
These isometries essentially apply the same mapping as $V^{GI}$, but now account for the additional $\ket{e}_e$ state we have present.
However, $V^{q}$ only maps $2\times 1$ spins to a single spin.
We need an operator which maps a $2\times 2$ spin to a single spin.
Define $W:\HS^{\prime (i/2,j)}_{eq}\ox \HS^{\prime (i/2,j+1)}_{eq} \rightarrow (\HS_{eq}^\prime\ox\HS_{eq}^\prime)^{(i/2,j/2)}$, as simply
\begin{align}
  W_{(i,i+1)(j,j+1)} = \sum (\ket{x}_{q_1}\otimes\ket{y}_{q_2})_{i/2,j/2} \bra{x}_{i/2,j}\ox\bra{y}_{i/2,j+1}.
\end{align}
This unitary acts to map the $1\times 2$ set of sites to a single lattice site in the renormalised lattice.

The isometry:
\begin{align} \label{Eq:V^eq_Definition}
  V^{eq}_{(i,i+1)(j,j+1)} := W_{(i,i+1)(j,j+1)}\left(V^q_{(i,i+1)}(j)\ox V^q_{(i,i+1)}(j+1)\right),
\end{align}
then maps $2\times 2$ spins to a single spin.

The overall blocking map $\B_u$ is then given by:
\begin{definition}[Blocking Isometry, $V^b$, $\B_u$] \label{Def:V^b_Isometry}
  Let $V^C$ and $V^{eq}$ be the isometries from \cref{Def:Tiling_Hamiltonian_RG}  and \cref{Eq:V^eq_Definition} respectively.
  Then the blocking isometry for $H_u$ is given by
  \begin{align}
    V^b_{(i,i+1)(j,j+1)} = V_{(i,i+1)(j,j+1)}^C\otimes V^{eq}_{(i,i+1)(j,j+1)}.
  \end{align}

\end{definition}

We now need to consider the full renormalisation process: the isometry defined above will map a certain subset of states to states on the renormalised lattice.
However, some parts of the Hilbert space will be ``integrated out''.
For convenience we will sometimes use indices $I,J$ to indicate row and column indices on the new lattice after the RG transformation.

Let $h_q^{(i,i+1)}(j), h_q^{(i,i+1)}(j+1)$ be the local terms of the quantum Hamiltonian before renormalisation, then we see that
\begin{align}
  &V^{eq}_{(i,i+1)(j,j+1)} \left(h_q^{(i,i+1)}(j+1) +  h_q^{(i,i+1)}(j)\right) V^{eq\dagger}_{(i,i+1)(j,j+1)} \nonumber \\
  &= h_q^{(1)\prime(I,J)}\ox \1_{q_2}
   + \1_{q_1} \ox h_q^{(1)\prime(I,J)}
\end{align}
and
\begin{align}
  &V^{eq}_{(i+2,i+3)(j,j+1)}V^{eq}_{(i,i+1)(j,j+1)} \left(h_q^{(i+1,i+2)}(j) +  h_q^{(i+1,i+2)}(j)\right)V^{eq\dagger}_{(i+2,i+3)(j,j+1)} \nonumber  \\  &\times  V^{eq\dagger}_{(i,i+1)(j,j+1)}
                                                                                             = h_{q_1}^{\prime(I,I+1)}\ox \1^{(I,J)}_{q_2} \ox \1^{(I+1,J)}_{q_2} + \1^{(I,J)}_{q_1} \ox \1^{(I+1,J)}_{q_1} \ox h_{q_2}^{\prime(I,I+1)}.
\end{align}

\paragraph{Truncation Operation $\T_u$ }
The operator $W$ has essentially merged two sites into a single site.
We now wish to integrate out one of these sites and restrict to the set of ``allowed states'' in the other.
We will implement this using the $1$-local projector $\Pi_{gs}(k)$
\begin{definition}[Truncation Operation $\T_u$] \label{Def:T_u_Mapping}
  Let $\ket{\psi}\in \HS_c\ox\HS_{eq}$, then
  \begin{align}
    \T_u: \ket{\psi} \mapsto (\1_c \ox \1_{q_1} \ox \Pi_{gs}(k))\ket{\psi},
  \end{align}
  where
  \begin{adjustwidth}{-2cm}{-2cm}
  	\begin{align} \label{Eq:Pi_gs_Definition}
  	\Pi_{gs}(k)= \begin{cases}
  	\ketbra{e^{\times 2^k}} 	\quad  &  \ \ k \ \text{even}  \\
  	\ketbra{\psi_{hist}(4^n+1)e^{\times 2^k-4^n-1}}  \quad & \ \ \text{if $k$ odd,  $2^{k-1} <4^n+1 <2^k$,} \\
  	& \ \ \text{  and non-halting}\\
  	\ketbra{\psi_{halt}(4^n+1)e^{\times 2^k-4^n-1}} \quad &  \ \  \text{if $k$ odd, $2^{k-1} <4^n+1 <2^k$,} \\
  	& \ \  \text{ and halting},
  	\end{cases} 
  	\end{align}
  \end{adjustwidth}
   and where $\ket{\psi_{hist}(L)}$ and $\ket{\psi_{halt}(L)}$ are defined in \cref{Lemma:Ground_State_GI_Form}.
  This extends to states $\ket{\chi}\in (\HS_c\ox\HS_{eq})^{\ox \Lambda(L)}$, as
  \begin{align}
    \T_u: \ket{\chi} \mapsto \bigotimes_{(I,J)\in \Lambda(L)}(\1_c^{(I,J)}\ox \1_{q_1}^{(I,J)} \ox \Pi_{gs}^{(I,J)}(k))\ket{\chi}.
  \end{align}
\end{definition}

\begin{definition}[Renormalisation Isometry, $V^u$] \label{Def:V^u_Isometry}
  Let $V^b_{(i,i+1)(j,j+1)}$ and $\Pi_{gs}$ be as defined in \cref{Def:V^b_Isometry} and \cref{Eq:Pi_gs_Definition} respectively.
  We define the isometry implementing the entire renormalisation scheme as
  \begin{align}
    V^u_{(i,i+1)(j,j+1)} :=(\1_c \ox \Pi_{gs}) V^b_{(i,i+1)(j,j+1)}.
  \end{align}
\end{definition}

To see why this is appropriate note that the Hamiltonian after the application of the blocking isometries has two sets of local terms: a 1-local term and a 2-local term (see \cref{Def:h_u_RG_Mapping} and the discussion following).
First consider the 1-local term $h_q^{(1)\prime(I,J)}\ox \1_{q_2} + \1_{q_1} \ox h_q^{(1)\prime(I,J)}$ and examine how it transforms under $\T_u$ and $\Pi_{gs}$.
The idea is that $\Pi_{gs}$ will ``integrate out'' the $q_2$ subspace by removing all states which are not the ground state while maintaining the energy contribution from this subspace.
% We now
If the site is large enough to contain a full history state of length $4^n+1$, for some $n\in \N$, then we keep only that state and the relevant renormalised $\ket{e}$ states.
Otherwise we keep only the renormalised $\ket{e}$ states.
Hence
\begin{align}
  \Pi_{gs}^{(I,J)}(k)(&h_{q_1}^{(1)\prime(I,J)}\ox \1_{q_2}^{(I,J)} + \1_{q_1}^{(I,J)} \ox h_{q_2}^{(1)\prime(I,J)} )\Pi_{gs}^{(I,J)}(k) \\
  = &h_{q_1}^{(1)\prime(I,J)} \ox \Pi_{gs}^{(I,J)}(k) + \tr\left(\Pi_{gs}^{(I,J)}(k) h_{q_2}^{\prime(I,J)} \right) \1_{q_1}^{(I,J)} \ox \Pi_{gs}^{(I,J)}(k).
\end{align}
Since $\Pi_{gs}$ is a projector onto a 1-dimensional subspace, we will often omit it when writing the Hamiltonian.
Thus obtain the term
\begin{equation}
	h_q^{(1)\prime(I,J)} + \Tr\left(\Pi_{gs}(k)h_{q_2}^{\prime(I,J)}\right) \1_{q}.
\end{equation}

Now examine how the 2-local terms transform:
{\small
  \begin{align}
    &\Pi_{gs}(k)^{(I,J)}\ox \Pi_{gs}(k)^{(I+1,J)}\big(h_q^{\prime(I,I+1)}\ox \1^{(I,J)}_{q_2} \ox \1^{(I+1,J)}_{q_2} \\
    &+ \1^{(I,J)}_{q_1} \ox \1^{(I+1,J)}_{q_1} \ox h_q^{\prime(I,I+1)} \big) \Pi_{gs}(k)^{(I,J)}\ox \Pi_{gs}(k)^{(I+1,J)}\\
    &= h_q^{\prime(I,I+1)}\ox \Pi_{gs}(k)^{(I)}\ox \Pi_{gs}(k)^{(I+1)} \\
    &+ \tr\left(h_q^{\prime(I,I+1)} \Pi_{gs}(k)^{(I)}\ox \Pi_{gs}(k)^{(I+1)}\right) \1^{(I,J)}_{q_1}\ox  \Pi_{gs}(k)^{(I)}\ox \Pi_{gs}(k)^{(I+1)}.
  \end{align}
}
Importantly $\tr\left(h_q^{\prime(I,I+1)} \Pi_{gs}(k)^{(I)}\ox \Pi_{gs}(k)^{(I+1)}\right)$ only picks up a non-zero contribution from the terms proportional to $\1^{(I)} \ox \1^{(I+1)}$ (we also note that this latter term is zero for interactions going along columns).
Again the subspace spanned by $\Pi_{gs}(k)^{(I)}\ox \Pi_{gs}(k)^{(I+1)}$ is a 1-dimensional subspace and hence we will often omit writing it explicitly.
Thus the 2-local terms effectively become
$h_q^{\prime(I,I+1)} + \tr\left(h_q^{\prime(I,I+1)} \Pi_{gs}(k)^{(I)}\ox \Pi_{gs}(k)^{(I+1)}\right)\1^{(I,J)}_{q}\ox \1^{(I+1,J)}_{q}$.

\paragraph{Multiple Iterations}
The above is the RG transformation for a single iteration; in the following we construct the further iterations of the RG mapping analogously to the above.

\medskip

First define the set of local basis states in the quantum part of the Hilbert space,
\begin{align}
  \frk{C}^{\prime(k)}:=\frk{B}^{(k)}\bigcup_{j=0^{2^K}}\bigcup_{\ket{x_i}\in \frk{B}\cup \ket{e}} \ket{x_1...x_{2^k}},
\end{align}
such that $\ket{x_1\dots x_{2^{k-1}}}\in \frk{C}^{(k-1)}$.
From this we can define $\HS_{eq}^{\prime(k)}=\spann \left\{\ket{x} \big| \ket{x}\in \frk{C}^{\prime(k)} \right\}$.

Then  $\B_u:(\r^{(k-1)}\HS_{eq})^{\ox 2\times 2} \rightarrow \HS_{eq}^{\prime(k)}$.
Finally we truncate the basis states which are either bracketed or can immediately be identified as being illegal or evolving to an illegal state using $\T_u$.
This leaves us with the basis $\frk{C}^{(k)}$ as the set of basis states and the renormalised local quantum Hilbert space as $\Rk( \HS_{eq})=\spann\{\ket{x} | \ket{x}\in \frk{C}^{(k)} \}$.

The $\T_u \circ\B_u$ operation can be implemented analogously to the previously described transformation: we apply $V^b$ --- now defined on $\r^{(k-1)}(\HS_u)$ --- across the lattice which blocks and truncates part of the Hilbert space.
We then apply $\Pi_{gs}(k)$, as defined in \cref{Eq:Pi_gs_Definition}, to project out the local ground state (which may pick up energy).

We formalise the overall RG mapping in the following definition:
\begin{definition}[$h_u$ Renormalisation Mapping] \label{Def:h_u_RG_Mapping}
  Let $h_u^{col(i,i+1)}, h_u^{row(j,j+1)}\in \mathcal{B}(\C^d\otimes \C^d)$ and $V^u_{(i,i+1)(j,j+1)}$ be as in \cref{Def:V^u_Isometry}.
  Then the renormalised local terms are given by
  \begin{adjustwidth}{-2cm}{-2cm}
  \begin{align}
    \sR: & \ h_u^{row(i+1,i+2)}(j)+h_u^{row(i+1,i+2)}(j+1)
          \rightarrow
          V^u_{(i+2,i+3)(j,j+1)} V^u_{(i,i+1)(j,j+1)} \times \\
        &\left( h_u^{row(i+1,i+2)}(j)+h_u^{row(i+1,i+2)}(j+1) \right)
          \, V_{(i,i+1)(j,j+1)}^{u\dagger}
          V_{(i+2,i+3)(j,j+1)}^{u\dagger} \\
          %	&+ \gamma^{(k)}\1^{i}\ox \1^{(i+1)}, \\
        &=: \r(h_u^{row})^{(i,i+1)} \\
    \sR:& \ h_u^{col(j+1,j+2)}(i) +h_u^{col(j+1,j+2)}(i+1)
         \rightarrow
         V^u_{(i+2,i+3)(j,j+1)}V^u_{(i,i+1)(j,j+1)} \times \\ &\bigg(h_u^{col(j+1,j+2)}(i) \ +h_u^{col(j+1,j+2)}(i+1) \bigg) V_{(i,i+1)(j,j+1)}^{u\dagger}V_{(i+2,i+3)(j,j+1)}^{u\dagger}, \\
        &=:\r(h_u^{col})^{(i,i+1)}
  \end{align}
  \begin{align}
    \sR:& \ h_u^{row(i,i+1)}(j)+h_u^{row(i+1,i+2)}(j+1) + \sum_{\substack{k=0,1 \\ \ell=1,2}}\left(h_u^{(1)(i+k,j+\ell )}\right)  \rightarrow \\
       V^u_{(i,i+1)(j,j+1)} &\bigg( h_u^{row(i,i+1)}(j) +h_u^{row(i+1,i+2)}(j+1) + \sum_{\substack{k=0,1 \\ \ell=1,2}} \left(h_u^{(1)(i+k,j+\ell )} \right) \bigg)  V_{(i,i+1)(j,j+1)}^{u\dagger} \\
       %	& +  \kappa^{(k)}\1^{(i)} \\
       &=:\r(h_u^{(1)})^{(i)}.
  \end{align}
  \end{adjustwidth}
  $\Rk(h_u^{row}), \Rk(h_u^{col})^{(i,i+1)}, \Rk(h_u^{(1)})^{(i)}$ are defined in the same way but with the appropriate isometries for the $k^{th}$ iteration of the RG mapping.
\end{definition}
\begin{remark} \label{Remark:Local_Projectors}
  $\Rk(h_u^{(1)})^{(i)}$ and $\Rk(h_u^{row})^{(i,i+1)}$ have local projector terms of the form $\sum_{m=1}^{k}4^m\kappa^{(m)}\1^{(i)}$ and $\sum_{m=1}^{k}2^m\gamma^{(m)}\1^{(i)}\ox \1^{(i+1)}$, where $\gamma^{(k)}$  and  $\kappa^{(k)}$ are given by
  \begin{align}
    \kappa^{(k)} &:= \Tr\left(\Pi_{gs}(k)h_{q_2}^{\prime(I,J)}\right) % \tr\left(h_q^{(1)\prime}\Pi_h(4^n+1)\right) \label{Eq:kappa_defintion}
    \\
    \gamma^{(k)} &:= \tr\left( h_q^{\prime(I,I+1)} \Pi_{gs}(k)^{(I)} \ox \Pi_{gs}(k)^{(I+1)} \label{Eq:gamma_defintion} \right).
  \end{align}
\end{remark}

We now examine the properties of the full Hamiltonian under this mapping, and show that its ground state energy and other properties are preserved.

\begin{lemma}[$H_u$ Renormalisation] \label{Lemma:Full_Renormalization}
  Let $H_u(L)=\sum h_u^{row(j,j+1)} + \sum h_u^{col(i,i+1)}$, where
  \begin{subequations}\label{TQ:overall_H}
    \begin{align}
      h^{\mathrm{col}}_{j,j+1} =
      \label[term]{TQ:cols} &h_c^{\mathrm{col}}\ox \1_{eq}^{(j)} \ox \1_{eq}^{(j+1)}\\
      h^{\mathrm{row}}_{i,i+1} =
      \label[term]{TQ:Hc}   &h_c^{\mathrm{row}}\ox\1_{eq}^{(i)}\ox\1_{eq}^{(i+1)}\\
      \label[term]{TQ:Hq}   &+\1_c^{(i)}\ox\1_c^{(i+1)}\ox h_q\\
      % force < above L
      \label[term]{TQ:<L}   &+\ketbra{L}^{(i)}_c \ox (\1_{eq}-\ketbra{\leftend})^{(i)}
                              \ox \1_{ceq}^{(i+1)}\\
                              % force L below <
      \label[term]{TQ:L<}   &+(\1_c- \ketbra{L}_c)^{(i)} \ox \ketbra{\leftend}^{(i)}
                              \ox \1_{ceq}^{(i+1)}\\
                              % force > above R
      \label[term]{TQ:>R}   &+\1_{ceq}^{(i)} \ox \ketbra{R}^{(i+1)}_c
                              \ox (\1_{eq} - \ketbra{\rightend})^{(i+1)}\\
                              % force R below >
      \label[term]{TQ:R>}   &+\1_{ceq}^{(i)}
                              \ox (\1_c-\ketbra{R})^{(i+1)}_c
                              \ox\ketbra{\rightend}^{(i+1)}\\
                              % force non-blank in q-layer to left of R in c-layer
      \label[term]{TQ:1R}   &+\1_c^{(i)} \ox \ketbra{0}^{(i)}_e
                              \ox \ketbra{R}^{(i+1)}_c \ox \1_{eq}^{(i+1)}\\
                              % force non-blank in q-layer to right of L in c-layer
      \label[term]{TQ:L1}   &+\ketbra{L}^{(i)}_c \ox \1_{eq}^{(i)}
                              \ox\1_c^{(i+1)} \ox \ketbra{0}^{(i+1)}_e\\
                              % forbid non-blank to right of blank in q-layer except over L
      \label[term]{TQ:10-a} &+\1_c^{(i)} \ox \ketbra{0}^{(i)}_e
                              \ox (\1_c-\ketbra{L})^{(i+1)}_c
                              \ox (\1_{eq}-\ketbra{0})^{(i+1)}_e\\
                              % forbid non-blank to left of blank in q-layer except over R
      \label{TQ:10-b} &+(\1_c-\ketbra{R})^{(i)}_c
                        \ox (\1_{eq}-\ketbra{0})^{(i)}_e
                        \ox\1_c^{(i+1)} \ox \ketbra{0}^{(i+1)}_e, \\
                            &+ \1_{ceq}^{(i)} \ox \1_{ceq}^{(i+1)} \label{TQ:2-Local-Constant}\\
      h^{(1)}_i = &-(1 + \alpha_2(\varphi))\1_{ceq}^{(i)}, \label{TQ:1-Local-Constant}
    \end{align}
  \end{subequations}
  where
  \begin{align}
    \alpha_2(\varphi) := \sum_{4^n+7>|\varphi| }4^{-2n-1}\lambda_0(H_q(4^n)),
  \end{align}
  as defined in Proposition 53 of \cite{Cubitt_PG_Wolf_Undecidability}.
  Then the $k$ times renormalised Hamiltonian $\Rk(H_u)^{\Lambda(L\times H)}$ has the following properties:
  \begin{enumerate}
  \item For any finite region of the lattice, the restriction of the
    Hamiltonian to that region has an eigenbasis of the form  $\ket{T}_c\otimes\ket{\psi_i}$ where $\ket{T}_c$ is a classical tiling state (cf. Lemma 51 of \cite{Cubitt_PG_Wolf_Undecidability}). \label{H_e_Renorm:Separable_Eigenstates}
  \item Furthermore, for any given $\ket{T}_c$, the lowest energy choice for $\ket{\psi}_q$ consists of ground states of $\Rk(H_q)(r)$ on segments between sites in which $\ket{T}_c$ contains an $\ket{\Rk(L)}$ and an $\ket{\Rk(R)}$, a 0-energy eigenstate on segments between an $\ket{\Rk(L)}$ or $\ket{\Rk(R)}$ and the boundary of the region, and $\ket{e}$’s everywhere else.
    Any eigenstate which is not an eigenstate of $\Rk(H_q)(r)$ on segments between sites in which $\ket{T}_c$ contains an $\ket{\Rk(L)}$ and an $\ket{\Rk(R)}$ has an energy $>1$ (cf. Lemma 51 of \cite{Cubitt_PG_Wolf_Undecidability}). \label{H_e_Renorm:GS_Structure}
  \item The ground state energy is contained in the interval
    \begin{align}
      \bigg[ &(g(k) - 4^k\alpha_2(\varphi)) LH - 2^{-k}H
               + \sum_{n=1}^{\lfloor\log_4(L/2)\rfloor}
               \bigg(
               \left\lfloor\frac{H}{2^{2n+1 (k \ mod 2)}}\right\rfloor \\
             &\times\left(\left\lfloor\frac{L}{2^{2n+1 - (k \ mod 2)}}\right\rfloor -1\right)
               \bigg)
               \lambda_0(\Rk(H_{q})(4^{n- \lfloor (k \ mod 2)/2\rfloor})) , \\
             &(g(k) -  4^k\alpha_2(\varphi)) LH - 2^{-k}H
               +  \sum_{n=1}^{\lfloor\log_4(L/2)\rfloor}
               \bigg(\left(
               \left\lfloor\frac{H}{2^{2n+1- (k \ mod 2)}}\right\rfloor +1\right) \\
             &\times \left\lfloor\frac{L}{2^{2n+1 - (k \ mod 2)}}\right\rfloor
               \bigg)
               \lambda_0(\Rk(H_{q})(4^{n - \lfloor (k \ mod 2)/2\rfloor})) \bigg]
    \end{align}
    where
    \begin{align}
      g(k) = 4^k\sum_{4^n+1<2^k}4^{-2n-1}\lambda_0(H_q(4^n)),
    \end{align} \label{H_e_Renorm:GS_Energy}
    (cf. Lemma 52 of \cite{Cubitt_PG_Wolf_Undecidability}).
  \end{enumerate}
\end{lemma}
\begin{proof}
  We prove this in \cref{Appendix:Total_RG_Proof}.
\end{proof}

\begin{lemma}\label{Lemma:Renormalised_GI_Energies}
  Let $S_{br}(k)$ be the subspace spanned by states for which the left-most site is of the form $\ket{e^{\times p}\leftend \{x\}^{\times 2^k-p-1}}$ for a fixed integer $1 \leq p\leq 2^k$ and the right-most site is of the form $\ket{\{y\}^{\times 2^k-q-1 }\rightend e^{\times q}}$ for fixed integer $1 \leq q\leq 2^k$.
  Then
  \begin{align}
    \lambda_0(\Rk(H_q)(L)|_{S_{br}(k)}) = \min_{2^{k-1}L+1 \leq x\leq 2^kL} \lambda_0(H_q(x))
  \end{align}
\end{lemma}
\begin{proof}

  $\Rk(h_q)$ is block-diagonal with respect to the subspaces of $\Rk (\HS_{eq})^{\ox 2}$ spanned by products of $\ket{e^{\times p}\leftend \{x\}^{\times 2^k-p-1}}$ and~\newline
   $\ket{\{y\}^{\times 2^k-q-1 }\rightend e^{\times q}}$ for fixed $p,q$, together with the orthogonal complement thereof, while acting as identity on $\Rk(\HS_c)^{\ox 2}$.

\iffalse
  We then see that the ground state is given by either a history state or a modified history state as per \cref{Lemma:Ground_State_GI_Form}.
  The energy of such a state is now dependent on the runtime of the encoded QTM, which itself is determined by the effective length of the state (i.e. the distance between $\leftend$ and $\rightend$ markers), which itself is determined by the values of $p,q$.

  For a fixed $L$:
  \begin{itemize}
  \item if the encoded computation is non-halting for all $p,q$, the ground state is a zero energy history state.
  \item if the encoded computation is halting for all $p,q$, the ground state is the halting state and has energy $1-\cos(\frac{\pi}{2T})$ where $T$ here will depend on $p,q$.
    Thus to minimise the energy, we minimise $p$ and $q$.
  \item if the encoded computation is halting for some $p,q$ but non-halting for others, then the ground state must be non-halting history state with zero energy.
  \end{itemize}
\fi

  Thus the ground state energy is equal to $\min_{2^{k-1}L+1 \leq x\leq 2^kL}\lambda_0(H_q(x))$.
\end{proof}

\begin{corollary} \label{Corollary:GS_in_Limit}
  If $\lim\limits_{L\rightarrow\infty}\lambda_0(H_u^{\Lambda(L)})=+\infty$, then $\lim\limits_{L\rightarrow\infty}\lambda_0(\Rk(H_u)^{\Lambda(L)})=+\infty$ for all $k\geq k_0(|\varphi|)$, and $k_0(|\varphi|)$ is the smallest integer such that $2^{k_0}>|\varphi|+7$.
  If $\lim\limits_{L\rightarrow\infty}\lambda_0(H_u^{\Lambda(L)})=-\infty$, then $\lim\limits_{L\rightarrow\infty}\lambda_0(\Rk(H_u)^{\Lambda(L)})=-\infty$ for all $k\geq k_0(\varphi)$.
\end{corollary}
\begin{proof}
  Consider applying the RG mapping $k>k_0(\varphi)$ times, then we see that
  \begin{align}
    g(k) &= 4^k\sum_{4^n+1<2^k}4^{-2n-1}\lambda_0(H_q(4^n)) \\
         &= 4^k\sum_{4^n+1<2^{k_0}}4^{-2n-1}\lambda_0(H_q(4^n)) + 4^k\sum_{2^{k_0}<4^n+1<2^k}4^{-2n-1}\lambda_0(H_q(4^n)) \\
         &= 4^k\alpha_2(\varphi) +   4^k\sum_{2^{k_0}<4^n+1<2^k}4^{-2n-1}\lambda_0(H_q(4^n)).
  \end{align}
  From \cref{Lemma:Full_Renormalization}, the interval the ground state energy is contained in is
  \begin{adjustwidth}{-2cm}{-2cm}
  \begin{align}
    \bigg[ & LH\sum_{2^{k_0}<4^n+1<2^k}4^{-2n-1}4^k\lambda_0(H_q(4^n)) - 2^{-k}H  \nonumber  \\ &+\sum_{n=1}^{\lfloor\log_4(L/2)\rfloor} \left(
    \left\lfloor\frac{H}{2^{2n+1-(k \ mod  2)}}\right\rfloor
 \left(\left\lfloor\frac{L}{2^{2n+1-(k \ mod  2)}}\right\rfloor -1\right)
      \right)
      \lambda_0(\Rk(H_{q})(4^{n- \lfloor (k \ mod  2)/2\rfloor})), \nonumber \\
      & LH\sum_{2^{k_0}<4^n+1<2^k}4^{-2n-1}4^k\lambda_0(H_q(4^{n})) - 2^{-k}H \nonumber \\
      &+  \sum_{n=1}^{\lfloor\log_4(L/2)\rfloor}
             \left(\left(
             \left\lfloor\frac{H}{2^{2n+1-(k \ mod  2)}}\right\rfloor +1\right)
             \left\lfloor\frac{L}{2^{2n+1-(k \ mod  2)}}\right\rfloor
             \right)
             \lambda_0(\Rk(H_{q})(4^{n-\lfloor (k \ mod  2)/2\rfloor})) \bigg].
  \end{align}
  \end{adjustwidth}
  From \cref{Lemma:Renormalised_GI_Energies}, if $\lambda_0(H_q(4^n+1))=0$ for all $n$, then $\lambda_0(\Rk(H_q)(4^n+1))=0$ for all $n$.
  In this case the ground state energy becomes $\lambda_0(\Rk(H)^{\Lambda(L)})=-2^{-k}L \xrightarrow{L\rightarrow\infty}-\infty$.

  We see that if for any $n_0$, $\lambda_0(H_q(4^{n_0}+1))>0$, then $\lambda_0(\Rk(H_q)(4^n+1))>0$ $\forall n\geq n'_0$ ($n'_0$ not necessarily equal to $n_0$).
  Define $g(k)=\eta(k)+4^k\alpha_2(\varphi)$ then $\eta(k)\geq0$, and we see that the lower bound of the ground state is
  \begin{align}
    L^2\eta(k) - 2^{-k}L  + &\sum_{n=1}^{\lfloor\log_4(L/2)\rfloor}
    \left(
    \left\lfloor\frac{L}{2^{2n+1-(k \ mod  2)}}\right\rfloor
    \left(\left\lfloor\frac{L}{2^{2n+1-(k \ mod  2)}}\right\rfloor -1\right)
    \right) \times \nonumber \\
   &\lambda_0(\Rk(H_{q})(4^{n-\lfloor (k \ mod  2)/2\rfloor})) \xrightarrow{L\rightarrow \infty} +\infty.
  \end{align}

\end{proof}
For $2^{k_0}\leq |\varphi|+7$ the above relationship is not necessarily preserved.
To see why, note that for lengths $\ell \leq |\varphi|+7$ the Gottesman-Irani Hamilonian will not encode the correct computation and hence will pick up some energy.
Since $\lambda_0(\Rk(H_q)(L)|_{S_{br}}) = \min_{2^{k-1}L+1 \leq x\leq 2^kL} \lambda_0(H_q(x))$ rather than \newline
$\lambda_0(\Rk(H_q)(L)|_{S_{br}}) = \lambda_0(H_q(x))$, the energies in the summation term and the $\alpha_2$ term will not exactly cancel out until we reach higher order steps of the RG flow.
This is only rectified once we reach $2^{k_0}> |\varphi|+7$ as the energy integrated out by the projector $\Pi_{gs}$, as given in \cref{Def:T_u_Mapping}, is exactly $\lambda_0(H_q(x))$, not $\lambda_0(\Rk(H_q)(L)|_{S_{br}}) $.

\subsection{Renormalising $H_d$}
The only part of the Hamiltonian acting on $\HS_d$ is $H_d$; there is no coupling to other parts of the Hilbert space and so we can renormalise this part independently.
For concreteness, following \cite{Cubitt_PG_Wolf_Undecidability}, we will let $H_d$ be the critical XY-model with local terms $X_i\ox X_{i+1} + Y_i\ox Y_{i+1} + Z_i \ox \1^{(i+1)} + \1^{(i)}\ox Z_{i+1}$, which can be written as:
\begin{align}
  h_d^{row(i,i+1)} &= X_i\ox X_{i+1} + Y_i\ox Y_{i+1}, \\
  h_d^{col(i,i+1)} &= 0, \\
  h_d^{(1)(i)}     &= 2Z_i.
\end{align}
However any Hamiltonian with a dense spectrum in the thermodynamic limit could be substituted.
Since the critical XY model is critical, it forms a fixed point in any reasonable RG scheme.
Thus we expect any reasonable RG procedure to map the model to itself.
\begin{definition}[Renormalisation Unitary for $h_d$, $V^d$]
We define the isometry implementing the renormalisation operation as
	\begin{adjustwidth}{-2cm}{-2cm}
	\begin{align}\label{Eq:V^d_Definition}
		&V^d_{(i,i+1),(j,j+1)}\left(h_d^{row(i+1,i+2)}(j)+h_d^{row(i+1,i+2)}(j+1)\right)V^{d\dagger}_{(i,i+1),(j,j+1)}  = h_d^{row(i/2,i/2+1)}(j/2), \\
		&V^d_{(i,i+1),(j,j+1)}\left(h_d^{row(i,i+1)}(j)+h_d^{row(i,i+1)}(j+1)\right)V^{d\dagger}_{(i,i+1),(j,j+1)}  = 2Z_{i/2}.
	\end{align}
	\end{adjustwidth}
\end{definition}

%\begin{definition}[Renormalisation Unitary for $h_d$, $V^d$]
%  We define the unitary implementing the renormalisation operation as:
%  \begin{align}\label{Eq:V^d_Definition}
%    V^d_{(i,i+1),(j,j+1)}:=\sum_{x,y,z,w\in \{0,1\} } &\ket{x}_{(i/2,j/2)}\bra{x}_{(i,j)}\ox \ket{y}_{(i/2,j/2)}\bra{y}_{(i+1,j)}\\ &\ox \ket{z}_{(i/2,j/2)}\bra{z}_{(i,j+1)} \ox \ket{w}_{(i/2,j/2)}\bra{w}_{(i+1,j+1)}%\biggl(\bra{x}_{i,j}\ox \sum_{y,w,z\in \{0,1\}}\bra{y}_{(i,j+1)}\bra{w}_{(i+1,j)}\bra{z}_{(i+1,j+1)}\biggr).
%  \end{align}
%\end{definition}
%
%\begin{lemma}
%  Let $H_d=\sum_{i=-\infty}^{\infty}h_d^{(i,i+1)}$ be the dense Hamiltonian, and let $h_d^{(i,i+1)}$ be the local interactions of the 1D critical XY model.
%  Then the RG map
%  \begin{align}
%    \r_d(h_d^{(i,i+1)}) = V^{d\dagger}_{(i,i+1),(j,j+1)}h_d^{(i,i+1)}V^d_{(i,i+1),(j,j+1)}
%  \end{align}
%  preserves the ground state and excited states of $H_d=\sum_{i\in \Lambda} h_{i,i+1} $ in the thermodynamic limit.
%\end{lemma}
%\begin{proof}
%  The proof is trivial.
%  % After the RG map, the local terms are the same and so the ground state is still zero, and all excited states are still present with the same energy.
%\end{proof}

\subsection{Renormalising $\ket{0}$}
If we wish to preserve the form of the possible ground states depending, it is straightforward to see that this can be done if the states $\ket{0}$ simply get mapped to themselves $\ket{0}^{\ox (2\times 2)} \rightarrow \ket{0}$ under the RG operation.
This can be implemented using the isometry
\begin{align} \label{Eq:V^0_Definition}
  V^0_{(i,i+1),(j,j+1)} := \ket{0}_{(i/2,j/2)}\bra{0}_{(i,j)} \bra{0}_{(i+1,j)}\bra{0}_{(i,j+1)}\bra{0}_{(i+1,j+1)}.
\end{align}

\subsection{The Overall Renormalised Hamiltonian}
Accounting for the renormalisation of all the different parts of the Hamiltonian, we can now define renormalisation group mapping for the entire Hamiltonian.
Recall that the original local terms are
\begin{align} \label{Eq:Original_Full_Hamiltonian}
  h(\varphi)^{(i,j)} =&\ket{0}\bra{0}^{(i)}\otimes (\1 -\ket{0}\bra{0} )^{(j)} + (\1 -\ket{0}\bra{0} )^{(i)}\otimes \ket{0}\bra{0}^{(j)} \\
  &+h_u^{(i,j)}(\varphi) \otimes \1_d^{(i,j)} + \1_u^{(i,j)} \otimes h_d^{(i,j)} \\
  h(\varphi)^{(1)} =& -(1+\alpha_2(\varphi) )\Pi_{ud},
\end{align}
where $\alpha_2(\varphi)$ is defined in \cref{Lemma:Full_Renormalization}.

\begin{definition}[Full Renormalisation Group Mapping] \label{Def:Full_RG_Mapping}
  Let $V^u$, $V^0$, $V^d$ be the isometries defined in \cref{Def:V^u_Isometry}, \cref{Eq:V^0_Definition}, and \cref{Eq:V^d_Definition} respectively.
  Define
  \begin{align}
  V^{r}_{(i,i+1),(j,j+1)} \coloneqq  V^{0}_{(i,i+1),(j,j+1)}\oplus \Big(V^u_{(i,i+1),(j,j+1)}\ox V^d_{(i,i+1),(j,j+1)}\Big).
  \end{align}
  Then the overall RG mapping of local Hamiltonian terms is given by
  \begin{align}
    \sR:& h(\varphi)^{(i,i+1)} \mapsto V^{r\dagger}_{(i,i+1),(j,j+1)} h(\varphi)^{(i,i+1)}  V^{r}_{(i,i+1),(j,j+1)}\\
    \sR:& h(\varphi)^{(i+1,i+2)} \mapsto V^{r\dagger}_{(i+2,i+3),(j,j+1)}  h(\varphi)^{(i+1,i+2)} V^{r}_{(i,i+1),(j,j+1)} V^{r}_{(i+2,i+3),(j,j+1)}
  \end{align}
\end{definition}

\begin{lemma}\label{Lemma:Fully_Renormalised_Hamiltonian}
  Applying the RG mapping from \cref{Def:Full_RG_Mapping} to the terms in \cref{Eq:Original_Full_Hamiltonian} we see that the renormalised 1- and 2-local terms become
  \begin{align}
    \Rk (h(\varphi))^{(i,j)} =&2^k(\ketbra{0}^{(i)}\otimes \Pi_{ud}^{(j)}+ \Pi_{ud}^{(i)}\otimes \ketbra{0}^{(j)})\\
                                   &+ \Rk(h_u(\varphi))^{(i,j)} \otimes \1_d^{(i,j)} + \1_u^{(i,j)} \otimes h_d^{(i,j)} \label{Eq:Renormalised_Full_Hamiltonian} \\
    \Rk(h(\varphi))^{(1)} =&(g(k)  -4^k\alpha_2(\varphi)  - 2^k)\Pi_{ud}^{(i)} + \Rk(h_u^{(1)})^{(i)}
  \end{align}
  where $g(k)$ is defined in \cref{Lemma:Full_Renormalization}.
  All the terms are computable.
\end{lemma}
\begin{proof}
  Note that the RG isometry acts block-diagonally with respect to the subspaces spanned by $\ket{0}^{\ox (2\times 2)}$ and those spanned by states in $(\Rk(\HS_u)\otimes \HS_d)^{\ox(2\times 2)}$.
  Furthermore, any state are not in one of the two subspaces is projected out.
  The $h_u(\varphi)$, $h_d$ and 1-local terms transform as they would in the absence of the $\ket{0}$ state, thus giving the terms seen above.
  The explicit coefficients are calculated in \cref{Lemma:put-promise-together} in the appendix.
  The term $g(k)$ is computable for any $k$ by calculating the $\lambda_0(H_q)(4^n+1)$ for all $n\leq 2k+1$.
  Since this is a finite dimensional matrix for any finite $n$, this is a computable quantity.

  The form of the overall renormalisation isometry means the $\ketbra{0}^{(i)}\ox \Pi_{ud}^{(j)}$ term must be preserved in form, however, we note that because all states of $2\times 2$ blocks in different subspaces in the previous RG step must be in $\ket{0}^{\ox (2\times 2)}$ or  $(\Rk(\HS_u)\otimes \Rk(\HS)_d)^{\ox(2\times 2)}$, then two neighbouring blocks must pick up an energy penalty of $\times 2$ of the previous local terms.
\end{proof}

\begin{corollary}\label{Corollary:Uncomputable_Parameters} \label{Corollary:Family}
  The local terms of the initial Hamiltonian $h(\varphi)$ and all further renormalised local terms belong to a family of Hamiltonians\linebreak $\mathcal{F}(\varphi, \tau_1,\tau_2,\{\alpha_i\}_i,\{\beta_i\}_i)$, which all take the form
  \begin{align}
    \Rk (h(\varphi))^{(i,j)} =&\tau_1(\ketbra{0}^{(i)}\otimes \Pi_{ud}^{(j)}+ \Pi_{ud}^{(i)} \otimes \ketbra{0}^{(j)})\\
                                   &+ \Rk(h_u(\varphi,\{\beta_t\}_t))^{(i,j)} \otimes \1_d^{(i,j)} + \1_u^{(i,j)} \otimes \Rk(h_d)^{(i,j)} \\
    \Rk(h(\varphi))^{(1)} =& \tau_2\Pi_{ud} + \Rk(h_u(\varphi, \{\alpha_t\}_t))^{(1)}, \label{Eq:tau_2_Parameter}
  \end{align}
  where the sets $\{\alpha_t\}_t$, $\{\beta_i\}$ characterises the parameters of the renormalised Gottesman-Irani Hamiltonian.
  Furthermore, for any $k\in \mathbb{N}$, the coefficients $\tau_1(k)$, $\tau_2(k), \{\alpha_t(k)\}_t$ and $\{\beta_t(k)\}_t$ are computable.
\end{corollary}
\begin{proof}
  Follows immediately from \cref{Lemma:Fully_Renormalised_Hamiltonian}.
\end{proof}

\begin{lemma}\label{Lemma:Thermodynamic_Limit_Properties}
  Let $\Rk(h(\varphi))^{(i,j)},\Rk(h(\varphi))^{(1)}$ be the local terms defined by the RG mapping in \cref{Def:Full_RG_Mapping} for any $k>k_0(|\varphi|)$.
  The Hamiltonian $\Rk(H)$ defined by these terms then has the following properties:
  \begin{enumerate}
  \item If the unrenormalised Hamiltonian $H(\varphi)$ has a zero energy ground state with a spectral gap of 1/2, then $\Rk(H)$ also has a zero energy ground state with zero correlations functions, and has a spectral gap of $\geq2^{k}$.
  \item If the unrenormalised Hamiltonian $H(\varphi)$ has a ground state energy $-\infty$ with a dense spectrum above this, then $\Rk(H)$ also a ground state energy of $-\infty$ with a dense spectrum, and has algebraically decaying correlation functions.
  \end{enumerate}
\end{lemma}
\begin{proof}
	First examine the spectrum of the renormalised Hamiltonian from \cref{Lemma:Fully_Renormalised_Hamiltonian}: for convenience let
	\begin{equation}
		\Rk(h_0)^{(i,j)}:=2^k(\ketbra{0}^{(i)}\otimes \Pi_{ud}^{(j)}+ \ketbra{0}^{(i)}\otimes \Pi_{ud}^{(j)}).
	\end{equation}
	  Further let
	  \begin{align}
	  \Rk(H_0^{\Lambda(L)})&:=\sum_{\langle i,j\rangle}\Rk(h_0)^{(i,j)}, \\ \Rk(\tilde{H_u})^{\Lambda(L)}&:=\sum_{\langle i,j\rangle} \1^{(i,j)}_d \ox \Rk(h_u)^{(i,j)} \\ \Rk(\tilde{H_d})^{\Lambda(L)}&:=\sum_{\langle i,j\rangle} \1^{(i,j)}_u \ox \Rk(h_d)^{(i,j)}
	  \end{align}
	We note $\Rk(H_0)^{\Lambda}$,$ \Rk(\tilde{H_d})^{\Lambda}$, $ \Rk(\tilde{H_u})^{\Lambda}$ all commute.
	Further note that
	\begin{equation}
		\spec\Rk(H_0)^{\Lambda}\subset 2^k\mathbb{Z}_{\geq 0}.
	\end{equation}

	If $\lambda_0(H(\varphi))=0$, then it implies $\lambda_0(H_{u}(\varphi))\rightarrow +\Omega(L^2)$ (see \cref{Section:Properties_Of_Spec_Gap}).
	By \cref{Corollary:GS_in_Limit}, this implies $\lambda_0(\Rk(H_u(\varphi)))\rightarrow +\Omega(L^2)$ too.
	Hence the ground state is the zero-energy $\ket{0}^{\Lambda(L)}$ state.
	Since $\spec\Rk(H_0)^{\Lambda}\subset 2^k\mathbb{Z}_{\geq 0}$, then
	the first excited state (provided $L$ is sufficiently larger) has energy at least $2^k$.
	Finally, the state $\ket{0}^{\Lambda(L)}$ has zero correlations.

	If  $\lambda_0(H(\varphi))=-\Omega(L)$, then $\lambda_0(H_{u}(\varphi))\rightarrow-\Omega(L)$ (see \cref{Section:Properties_Of_Spec_Gap}).
	By \cref{Corollary:GS_in_Limit}, this implies  $\lambda_0(\Rk(H))\rightarrow-\Omega(L)$.
	Since $\spec(\Rk(H_0))\subset 2^k\mathbb{Z}_{\geq 0}$, then the ground state is the ground state of $\Rk(\tilde{H_d})^{\Lambda(L)}+\Rk(\tilde{H_u})^{\Lambda(L)}$.
	Since $\spec(\Rk(\tilde{H_d})^{\Lambda(L)})$ becomes dense in the thermodynamic limit, we see that the Hamiltonian has a dense spectrum in the thermodynamic limit.
	Let $\ket{\psi}_u$ and $\ket{\phi}_d$ be the ground states of $\Rk(H_u)^{\Lambda(L)}$ and $\Rk(H_d)^{\Lambda(L)}$ respectively, then the ground state of $\Rk(\tilde{H_d})^{\Lambda(L)}+\Rk(\tilde{H_u})^{\Lambda(L)}$ is $\ket{\psi}_u\ket{\phi}_d$
	Since $\Rk(H_d)^{\Lambda(L)}$ is just the critical XY-model and its ground state has algebraically decaying correlations \cite{Lieb_Schultz_Mattis_1961}, hence the overall ground state has algebraically decaying correlations.
\end{proof}

\subsection{Order Parameter Renormalisation}

In \cref{Sec:Order_Parameter} we saw that the observable $O_{A/B}(r)$ functioned as an order parameter which distinguished the two phases.
Defining $V_r\coloneqq V^{0}_{(i,i+1),(j,j+1)}\oplus \Big(V^u_{(i,i+1),(j,j+1)}\ox V^d_{(i,i+1),(j,j+1)}\Big)$, and $V_r[k]$ as the corresponding isometry for the $k^{th}$ step of the RG process, then define
\begin{align}
\Rk(O_{A/B})(r) := V^r[k] O_{A/B}(2^kr)V^{r\dagger}[k].
\end{align}
The following lemma then holds:
\begin{lemma} \label{Lemma:Order_Parameter_Renormalized}
Let $\ket{\psi_{gs}}$ be the ground state of $H_u$.
The expectation value of the order parameter satisfies:
\begin{equation}
	\bra{\psi_{gs}} \Rk(O_{A/B})(r)\ket{\psi_{gs}}
	=
	\begin{cases}
 		1 & \text{if} \ \lambda_0(\Rk(H))=0 \\
 		0 & \text{if} \ \lambda_0(\Rk(H))=\Omega(L).
	\end{cases}
\end{equation}
\end{lemma}
\begin{proof}
	If $\lambda_0(\Rk(H)) \rightarrow -\Omega(L)$, then the ground state is that of $H_u^{(\Lambda(L))}$, and hence the state $\ket{0}$ does not appear anywhere in the ground state.~\newline
	If $\lambda_0(\Rk(H))=0$, the ground state is $\ket{0}^{\Lambda(L)}$.
	Since, under $V_r[k]$, $\ket{0}^{\ox 2^k\times 2^k}\mapsto \ket{0}$, the lemma follows.
\end{proof}
Thus the renormalised order parameter still acts as an order parameter for the renormalised Hamiltonian.
In particular, it still undergoes a non-analytic change when moving between phases.

\subsection{Uncomputability of RG flows}
We finally have all the ingredients for the proof of our two main results.
\begin{theorem}[Exact RG flow for undecidable Hamiltonian]\label{Theorem:Undecidability_of_RG_Flows_Formal}
  Let $H$ be the Hamiltonian defined in \cite{Cubitt_PG_Wolf_Undecidability}.
  % $\Rk(H)$ be the Hamiltonians defined in the thermodynamic limit by the original local terms and the $k$-times renormalised local terms respectively.
  The renormalisation group procedure, defined in \cref{Def:Full_RG_Mapping}, has the following properties:
  \begin{enumerate}
  \item \label{RG_Condition_1_2} $\sR(h)$ is computable.
  \item \label{RG_Condition_2_2} If $H(\varphi)$ is gapless, then $\Rk(H(\varphi))$ is gapless, and if $H(\varphi)$ is gapped, then $\Rk(H(\varphi))$ is gapped.
  \item \label{RG_Condition_3_2} For the order parameter of the form $O_{A/B}(r)$ which distinguished the phases of $H^{\Lambda(L)}$, there exists a renormalised observable $\Rk(O_{A/B})(r)$ which distinguishes the phases of $\Rk(H)^{\Lambda(L)}$ and is non-analytic at phase transitions. % For any order parameter of the form $O(\ket{\nu})$ which distinguished the phases of $H^{\Lambda(L)}$, there exists a renormalised observable $\Rk(O(\ket{\nu}))$ which distinguishes the phases.
  \item \label{RG_Condition_4_2} For $k$ iterations, the renormalised local interactions of $\Rk(H)$ are computable and belong to the family $\mathcal{F}(\varphi, \tau_1, \tau_2, \{\beta_i\})$, as defined  in \cref{Corollary:Family}.
  \item \label{RG_Condition_5_2} If $H(\varphi)$ initially has algebraically decaying correlations, then $\Rk(H(\varphi))$ also has algebraically decaying correlations.
    If $H(\varphi)$ initially has zero correlations, then $\Rk(H(\varphi))$ also has zero correlations.
  \end{enumerate}
\end{theorem}
\begin{proof}
  Claim \ref{RG_Condition_1_2} follows from \cref{Def:Full_RG_Mapping}, where the renormalisation isometries and subspace restrictions are explicitly written down and are manifestly computable, and hence for any $k$ the coefficients in \cref{Lemma:Fully_Renormalised_Hamiltonian} are computable.
  Claim \ref{RG_Condition_2_2} follows from \cref{Lemma:Thermodynamic_Limit_Properties}: we see that, for all $k>k_0$ the spectrum below energy $2^{k-1}$ is either dense with a ground state with energy at $-\infty$, or is empty except for a single zero energy state, corresponding to the gapped and gapless cases of $H(\varphi)$.
  Claim \ref{RG_Condition_3_2} follows from \cref{Lemma:Order_Parameter_Renormalized}.
  Claim \ref{RG_Condition_4_2} follows from \cref{Corollary:Uncomputable_Parameters}.
  Claim \ref{RG_Condition_5_2} follows from the properties of the ground states in the cases $\lambda_0(H_u^{\Lambda(L)})\rightarrow \pm \infty$ and by \cref{Lemma:Thermodynamic_Limit_Properties}.

\end{proof}

\begin{theorem}[Uncomputability of RG flow]\label{Theorem:tau_2_Divergence}
  Let $h(\varphi)$, $\varphi\in \mathbb{Q}$, be the full local interaction of the Hamiltonian from \cite{Cubitt_PG_Wolf_Undecidability}.
  % $H(\varphi):=\sum h(\varphi)^{(i,j)}$ is gapped if the UTM corresponding to $h{(\varphi)}$ halts on input $\varphi$, and gapless if the UTM never halts, where gapped and gapless are defined in \cref{Def:gapped} and \cref{Def:gapless}.
  Consider $k$ iterations of the RG map from \cref{Def:Full_RG_Mapping} acting on $H(\varphi)$, such that the renormalised local terms are given by  $\Rk(h(\varphi))$, which can be parameterised as per \cref{Corollary:Uncomputable_Parameters}.

  If the UTM is non-halting on input $\varphi$, then for all $k>k_0(\varphi)$ we have that $\tau_2(k)=-2^{k}$, for some computable $k_0(\varphi)$.
  If the UTM halts on input $\varphi$, then there exists an uncomputable $k_h(\varphi)$ such that for $k_0(\varphi)<k< k_h(\varphi)$ we have $\tau_2(k)=-2^{k}$, and for all $k>k_h(\varphi)$ then $\tau_2(k)=-2^{k}+\Omega(4^{k-k_h(\varphi)})$.
\end{theorem}
\begin{proof}
  Consider the expression for $\tau_2$ from \cref{Lemma:Fully_Renormalised_Hamiltonian}:
  \begin{align}
    \tau_2(k) &= 4^k\sum_{4^n+1<2^k}4^{-2n-1}\lambda_0(H_q(4^n))+  4^k\alpha_2(\varphi) -  2^k \\
    \tau_2(k) &= 4^k\sum_{4^n+1<2^k}4^{-2n-1}\lambda_0(H_q(4^n))+  4^k\alpha_2(\varphi) -  2^k.
  \end{align}
  From the definition of $\alpha_2(\varphi)$, we see that there is a $k_0(\varphi)$ such that $g(k_0(\varphi))=\alpha_2(\varphi)$, and hence we get
  \begin{align}
    \tau_2(k) &= -  2^k + 4^k\sum_{2^{k_0(\varphi)}<4^n+1<2^k}4^{-2n-1}\lambda_0(H_q(4^n)).
  \end{align}
  If the encoded QTM never halts, then by \cref{Lemma:Ground_State_GI_Form} $\lambda_0(H_q(4^n))=0$ for all $n$ such that $4^n+1>2^{k_0(\varphi)}$.
  If the encoded UTM halts then by \cref{Lemma:Ground_State_GI_Form} there exists an $n_0$ such that $\lambda_0(H_q(4^n))>0$ for all $n>n_0$.
  Then $k_h(\varphi)$ is defined as the minimum $k$ such that  $4^{n_0}+1<2^{k_h(\varphi)}$.
  Thus determining $k_h(\varphi)$ is at least as hard as computing the halting time and thus is an uncomputable number.

\end{proof}

\begin{figure}
  \centering
  \includegraphics[width=0.5\textwidth]{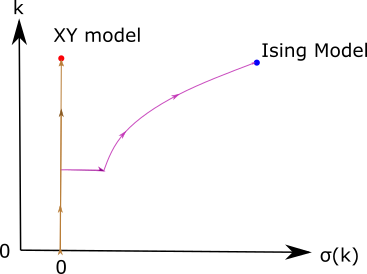}
  \caption{A schematic picture of the flow of Hamiltonians in parameter space.
    $\sigma(k)$ is defined in
    Orange represents some value of $\varphi=\varphi_0$ for which the QTM does not halt on input, while purple represents $\varphi=\varphi_0+\epsilon$ for any algebraic number $\epsilon$ for which the QTM halts.
    For small $k$, the orange and purple lines coincide.
    Then at a particular value of $k$, $\sigma(k)$ becomes non-zero and then increases exponentially.}
\end{figure}

%=============================================================
\section{Fixed points of the RG flow} \label{Sec:Fixed_Points}
%=============================================================
 \Cref{Theorem:Undecidability_of_RG_Flows_Formal} shows that our RG scheme satisfies the expected properties.
 We now qualitatively examine the Hamiltonian for large values of $k$.

 \subsection{Fixed Point for Gapped Instances}
	Here we show that for gapped instances the Hamiltonian becomes ``Ising-like'', for appropriately small energy scales.
	From \cref{Corollary:Family} the renormalised Hamiltonian is
	\begin{align}
	\Rk (h^{row}(\varphi))^{(i,j)} =&2^k(\ketbra{0}^{(i)}\otimes \Pi_{ud}^{(j)}+ \ketbra{0}^{(i)}\otimes \Pi_{ud}^{(j)}) \label{Eq:Ising_Terms}\\
	+& \Rk(h_u^{row}(\varphi)')^{(i,j)} \otimes \1_d^{(i,j)} + \1_u^{(i,j)} \otimes \Rk(h_d)^{(i,j)}  \label{Eq:Non-Ising_Terms}\\
	+&2^k \Pi_{ud}^{(i)}\ox \Pi_{ud}^{(j)} \\
	\Rk (h^{col}(\varphi))^{(i,j)} =&2^k(\ketbra{0}^{(i)}\otimes \Pi_{ud}^{(j)}+ \ketbra{0}^{(i)}\otimes \Pi_{ud}^{(j)}) \label{Eq:Ising_Terms_2}\\
	+& \Rk(h_u^{col}(\varphi)')^{(i,j)} \otimes \1_d^{(i,j)}   \\
	\Rk(h(\varphi))^{(1)} =&(g(k)  -4^k\alpha_2(\varphi)  - 2^k)\Pi_{ud} + \Rk(h_u^{(1)}(\varphi)),
	\end{align}
	where here we have explicitly separated out $\Pi_{ud}^{(i)}\ox \Pi_{ud}^{(j)}$ from the term\linebreak $\Rk(h^{row}_u(\varphi))^{(i,j)}=\Rk(h^{row}_u(\varphi)')^{(i,j)}+\Pi_{ud}^{(i)} \ox \Pi_{ud}^{(j)}$.

	Define the Ising-like Hamiltonian with local terms:
	\begin{align}
%	=& \Rk (h(\varphi))^{(i,j)} -  \Rk(h_u(\varphi)')^{(i,j)} \otimes \1_d^{(i,j)} - \1_u^{(i,j)} \otimes \Rk(h_d)^{(i,j)} - \\
%	&- \1_{ud}^{(i)}\ox \Rk(h_u^{(1)}(\varphi))^{(j)} - \Rk(h_u^{(1)}(\varphi))^{(j)}\ox \1_{ud}^{(j)} , \\
	h^{'row}_{Ising}(k)^{(i,j)} &:= 2^k\left(\ketbra{0}^{(i)}\otimes \Pi_{ud}^{(j)}+  \Pi_{ud}^{(j)}\otimes\ketbra{0}^{(i)} + \Pi_{ud}^{(i)}\ox\Pi_{ud}^{(j)}\right)\\
	h^{'col}_{Ising}(k)^{(i,j)} &:= 2^k\left(\ketbra{0}^{(i)}\otimes \Pi_{ud}^{(j)}+  \Pi_{ud}^{(j)}\otimes\ketbra{0}^{(i)}  \right) \\
	h'_{Ising}(k)^{(1)} &:= B(k)\Pi_{ud}.
	\end{align}
	This is reminiscent of the Ising interaction with both an ferromagnetic \linebreak $\dyad{0}^{(i)} \ox \dyad{1}^{(j)} + \dyad{1}^{(i)} \dyad{0}^{(j)}$ along the rows and columns and an anti-ferromagnetic $\dyad{1}^{(i)} \ox \dyad{1}^{(j)}$ term along just the rows, with local field $B(k)=(g(k) -4^k\alpha_2(\varphi) - 2^k)\dyad{1}$,
    but with the orthogonal projector $\Pi_{ud}$ playing the role of the projector onto the $\dyad{1}$ state.
    However, note that $\Pi_{ud}$ projects onto a larger dimensional subspace than $\dyad{1}$, so e.g.\ the partition function of this Ising-like Hamiltonian is not identical to that of an Ising model.

        We now show the following:
	\begin{proposition}
        Let $E$ be a fixed energy cut-off and $H'_{Ising}(k)=\sum_{\langle i,j\rangle }h'_{Ising}(k)^{(i,j)}$.
        Then
	\begin{align}
	\norm{\Rk (H(\varphi))|_{\leq E} - H'_{Ising}(k)|_{\leq E}}_{op} &\leq \left(\frac{E}{2^k}\right)^2.
	\end{align}
	\end{proposition}
\begin{proof}
	Consider the local interaction term $h_0 = \ketbra{0}\ox \Pi_{ud}+ \Pi_{ud}\ox \ketbra{0}$.
	This commutes with all other terms in both the $\Rk (H(\varphi))$ Hamiltonian and the Ising-like Hamiltonian, and hence the eigenstates of both of the overall Hamiltonians are also eigenstates of $\ketbra{0}\ox \Pi_{ud}+ \Pi_{ud}\ox \ketbra{0}$.
	As a result, for each eigenstate, a given site $p\in \Lambda$ either has support only on $\ket{0}_p$ or only on $\Rk(\HS_{ud})$.
	Therefore, an eigenstate defines regions (domains) of the lattice where all points in the domain are in $\HS_{ud}$.

	For a given eigenstate $\ket{\psi}$, let $D\coloneqq\left\{i\in\Z^2|\tr(\ket{0}\bra{0}^{(i)} \ket{\psi}\bra{\psi}) = 0 \right\}$ denote the region of the lattice where the state is supported on $\Rk(\HS_{ud})$, and $\partial D$ be the set of sites on the boundary of $D$.
	Then we see that the terms in \cref{Eq:Non-Ising_Terms} act non-trivially only within $D$, and that the boundaries of $D$ receive an energy penalty of $2^k|\partial D|$ from terms in \cref{Eq:Ising_Terms} and \cref{Eq:Ising_Terms_2}.

	Note that $\norm{\Rk(h_d)^{(i,j)}}_{op}, \norm{\Rk(h_u(\varphi)')^{(i,j)}}_{op}, \norm{\Rk(h_u^{(1)}(\varphi))}_{op}\leq 2$.
	For $\norm{\Rk(h_d)^{(i,j)}}_{op}$ this is straightforward to see.
        For $\norm{\Rk(h_u(\varphi)')^{(i,j)}}_{op}$, any states which pick up non-zero energy, other than those which receive a penalty due to halting, are removed from the local Hilbert space (as per \cref{Sec:Quantum_RG}).
    
    Let $m\in \N$ be a cut-off such that $|\partial D| \leq m$, hence $|D|\leq m^2/16$.
	Since for each boundary term we get an energy penalty of at least $2^k$ from $h_0$, we can relate $m$ to the energy cut-off $E$ to $m$ as $E \coloneqq 2^km$.    
	If we consider the Hamiltonians restricted to a subspace with energy $\leq E \coloneqq 2^km$, then
	\begin{align}
          &\norm{\Rk (H(\varphi))|_{\leq E} - H'_{Ising}(k)|_{\leq E}}_{op}\\
          &=  \norm{\sum_{\langle i,j\rangle}\left(\Rk(h_u(\varphi)')^{(i,j)} \otimes \1_d^{(i,j)} + \1_u^{(i,j)} \otimes \Rk(h_d)^{(i,j)}\right)\bigg|_{\leq E}}_{op} \label{Eq:Norm_Diff_1} \\
          \begin{split}
            &\leq \frac{m^2}{16}\bigg(\norm{\Rk(h_u(\varphi)')^{(i,j)} }_{op}+\norm{\Rk(h_d)^{(i,j)}}_{op}\\
            &\mspace{100mu}+\norm{\Rk(h_u^{(1)}(\varphi))}_{op}\bigg)\label{Eq:Norm_Diff_2}
          \end{split}\\
          &\leq \frac{m^2}{2} \\ % = \frac{E^2}{2^{2k+1}} 
            &< \left(\frac{E}{2^k}\right)^2.\label{Eq:Norm_Diff_3}
        \end{align}
        Going from \cref{Eq:Norm_Diff_1} to \cref{Eq:Norm_Diff_2} we have used the fact that the terms in the sum are only non-zero within domains, and $|D|\leq m^2/16$.
        Going from \cref{Eq:Norm_Diff_2} to \cref{Eq:Norm_Diff_3} we have used the bound on the individual norms of the local terms.
      \end{proof}

	Thus, for appropriately small energies, we expect only small deviations from the "Ising-like" Hamiltonian.
	And these deviations vanish as the RG process is iterated.
	In particular, the spectrum will look like \cref{Fig:Ising-Like_Energy_Diagram}.

	 \begin{figure}[h!]
		\begin{center}
			\includegraphics[width=0.6\textwidth]{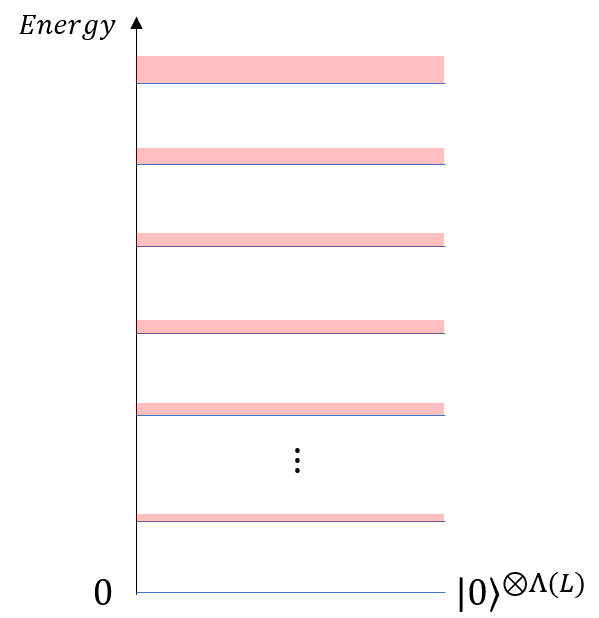}
		\end{center}
		 \caption{The energy level diagram of $\Rk(H)$.
		 The blue levels represent excitations of the $2^k(\ketbra{0}^{(i)}\otimes \Pi_{ud}^{(j)} + \Pi_{ud}^{(i)} \otimes \ketbra{0}^{(j)})$ term, while the red area represents the excited states of $\Rk(h_u(\varphi)')^{(i,j)}$, $\Rk(h_d)^{(i,j)}$,  and $\Rk(h_u^{(1)}(\varphi))$.
	 	The size of the red region increases as the domains get larger, and hence there are more high energy states.
	 	The ground state has no associated red region due to the presence of the spectral gap.
	 	The blue lines have an energy spacing of integer multiples of $2^k$.}\label{Fig:Ising-Like_Energy_Diagram}
	\end{figure}

 \subsection{Fixed Point for Gapless Instances}
 For a $\varphi$ for which $H(\varphi)$ is gapless, $\Rk(H(\varphi))$ is also gapless and we see that the ground state is that of $\Rk(H_u(\varphi))$.
 If we restrict to a low energy subspace, one can see that excited states are either the excited states of the Gottesman-Irani Hamiltonians or the excited states of the critical XY-model.
 Indeed, let $E(k)$ be the subspace of states with energy less than $2^k$, then for sufficiently large $k$ we see that
 \begin{align}
 \Rk(H)^{\Lambda}|_{E(k)} = \Rk(H_u(\varphi))^{\Lambda}|_{E(k)} \ox \1^{\Lambda}_d + \1^{\Lambda} \ox  \Rk(H_d)^{\Lambda}|_{E(k)}.
 \end{align}
 Since $\Rk(H_d)^{\Lambda}|_{E(k)}$ has the same spectrum as $H_d$, the spectrum of $\Rk(H)^{\Lambda}|_{E(k)}$ is also dense in the thermodynamic limit.
 Furthermore, $\Rk(H)^{\Lambda}|_{E(k)}$ has algebraically decaying correlations since $\Rk(H_d)^{\Lambda}|_{E(k)}$ also has algebraically decaying correlations \cite{Lieb_Schultz_Mattis_1961}.

%==============================================================
\section{Conclusions} \label{Sec:Conclusions}
%==============================================================

We have seen under the renormalisation group procedure constructed here, the Hamiltonian flows towards either an Ising-like Hamiltonian or an XY-like Hamiltonian.
Which case occurs depends on the parameter $\tau_2$ in \cref{Eq:tau_2_Parameter}.
Let $k$ be the number of iterations of the RG procedure, then from \cref{Theorem:tau_2_Divergence} we see that there are two cases: $\tau_2 = -2^k$ always, or $\tau_2 = -2^k$ initially, and once a sufficiently large value of $k$ is reached it begins to diverge as $\tau_2>-2^k+\Omega(4^k)$.
Determining which case occurs is undecidable.
Moreover, the value of $k$ at which we go from the first case to the second is uncomputable.
Thus, determining the trajectory of the system for an arbitrary value of $\varphi$ is uncomputable.
Even if $\varphi$ were known exactly, we see that the Hamiltonian's path in parameter space would be unpredictable.

Contrast this with chaotic behaviour:
for chaotic systems, a tiny difference in the initial system parameters can lead to large diverges in trajectories later.
Here the difficulty in predicting behaviour arises as it is usually difficult to determine the initial system parameters exactly.
However, if the system parameters are known exactly, it should theoretically be possible to ascertain the long-time system.
RG flows which undergo chaotic behaviour have been demonstrated before \cite{McKay_Berker_Kirkpatrick_1982, Svrakic_1982, Derrida_Eckmann_Erzan_1999, Damgaard_Thorleifsson_1991, Morozov_Niemi_2003}.

The behaviour of the RG trajectory shown here is stronger than this in that even if the initial parameters characterising the microscopic interactions are known \emph{exactly}, determining which fixed point the system may flow to is not possible to determine.
We compare this to a similar uncomputability result in \cite{Moore90} which showed that computing the trajectory of a particle in a potential is uncomputable.

The Hamiltonian discussed in this work is highly artificial and the RG scheme reflects this.
Indeed, this Hamiltonian has an enormous local Hilbert space dimension and its matrix elements are functions of both $\varphi$ and the binary length of $\varphi $, $|\varphi|$.
Both of these factors are unlikely to be present in naturally occurring Hamiltonians.
Thus an obvious route for further work is to consider RG schemes for more natural Hamiltonians which display undecidable behaviour.

Furthermore, although the RG scheme is essentially a simple BRG scheme, the details of its construction and analysis rely on knowledge of the structure of the ground states.
Due to the behaviour of this undecidable model, any BRG scheme will have to exhibit similar behaviour to the one we have analysed rigorously here.
But it would be nice to find a simpler RG scheme for this Hamiltonian (or other Hamiltonians with undecidable properties) which is able to truncate the local Hilbert space to a greater degree, without using explicit a priori knowledge of the ground state, for which it is still possible to prove this rigorously.

The Hamiltonian and RG scheme constructed here could also be used to prove rigorous results for chaotic (but still computable) RG flows.
Indeed, if we modify the Hamiltonian $H(\varphi)$ so that instead of running a universal Turing Machine on input $\varphi$, it carries out a computation of a (classical) chaotic process (e.g.\ repeated application of the logistical map), then two inputs which are initially very close may diverge to completely different outputs after some time.
By penalising this output qubit appropriately, the Hamiltonian will still flow to either the gapped or gapless fixed point depending on the outcome of the chaotic process under our RG map, but the RG flow will exhibit chaotic rather than uncomputable dynamics.

Given the RG scheme here, it is also relevant to ask is whether we can apply a similar scheme to the Hamiltonians designed in \cite{Bausch_1D_Undecidable, Bausch_Cubitt_Watson}.
Although we do not prove it here, we expect to be able to apply the modified BRG developed in this work to these Hamiltonians in an analgous way.
The only additional consideration is the so-called "Marker Hamiltonian" component of both of these constructions which would need additional care in a rigorous proof.
Since the Marker Hamiltonian has a similar ground state structure to the circuit-to-Hamiltonian mapping --- consisting of superpositions of a particle propagating along a line --- we expect a similar RG process to suffice.
%However, this can be treated in a similar way to the circuit-to-Hamiltonian mapping.
As a result, we do not expect an fundamentally different behaviour in the RG flow from the Hamiltonian analysed here.

%==============================================================
\phantomsection
\addcontentsline{toc}{section}{Acknowledgements}
\section*{Acknowledgements}
%==============================================================

E.O.\ and T.S.C.\ are supported by the Royal Society.
J.D.W.\ is supported by the EPSRC Centre for Doctoral Training in Delivering Quantum Technologies (grant EP/L015242/1).
This work has been supported in part by the EPSRC Prosperity Partnership in Quantum Software for Simulation and Modelling (grant EP/S005021/1), and by the UK Hub in Quantum Computing and Simulation, part of the UK National Quantum Technologies Programme with funding from UKRI EPSRC (grant EP/T001062/1).

% -------------Bibliography--------------
% ---------------------------------------
\phantomsection
\addcontentsline{toc}{section}{References}
\printbibliography
% ---------------------------------------
% ---------------------------------------

% =====================================================================
% ===================================APPENDIX==========================
% =====================================================================
\newpage

\appendix

\begin{center}
  {\LARGE \textbf{Appendix}}
\end{center}

% ===========================================================
\section{Recontructing Robinson pattern of 2D plane} \label{appendix:pattern}
% ===========================================================

A first interesting fact is that $2\times 2$ supertiles having a parity cross on the bottom left pointing up-right must have the following structure.
A parity left tile on top-left corner, a parity down tile in the bottom-right corner, and consequently there must be a free cross on the top-right of the supertile.
The orientation of the free cross will uniquely determine the type of left tile and down tile in the same supertile.
Thus, there are only 4 supertiles with a parity cross pointing up-right:

\begin{figure}[h!]
\begin{center}
  \includegraphics[width=0.9\textwidth]{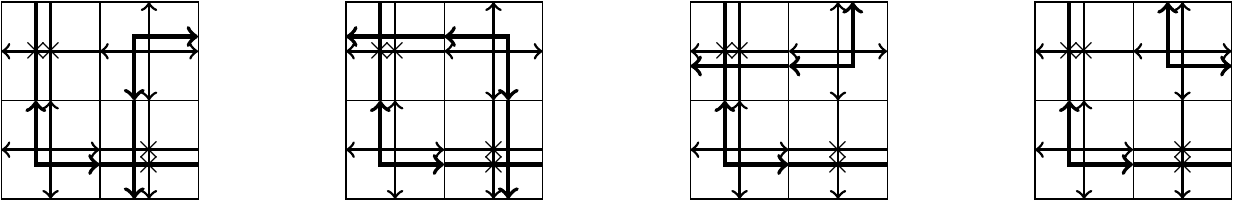}
\end{center}
\end{figure}
We make a first educated guess: each of these four supertiles corresponds to the $1\times 1$ parity cross having the same orientation as the free cross contained in the top-right of the supertile.

\medskip

If we want to tile the plane according to the Robinson pattern, these supertiles must then appear in alternate positions in alternate rows.
To this aim, we will assign parity rules to the supertiles according to the orientation of the $1\times 1$ parity cross on the bottom-left of each supertile (please note: this is a parity associated to the supertile as a whole and it is different to its inner parity structure). That is,

\begin{figure}[h!]
\begin{center}
	\includegraphics[width=0.85\textwidth]{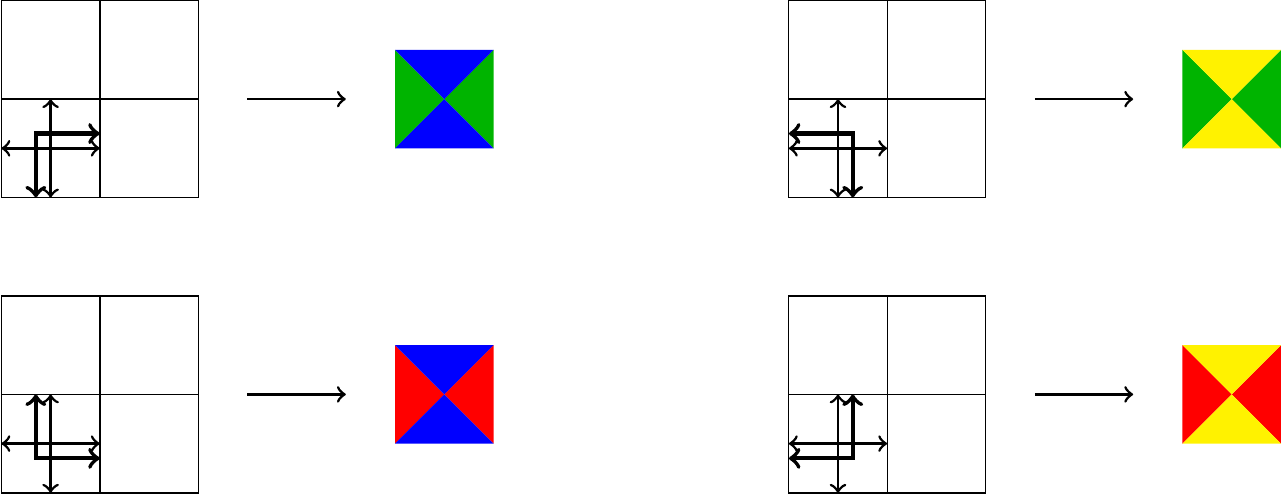}
\end{center}
\end{figure}

Thus, $2\times 2$ supertiles with a up-right $1\times 1$ parity cross in the bottom left must be interleaved in the vertical direction with a supertile  with a $1\times 1$ parity cross pointing bottom-right and in the horizontal direction with a supertile having a $1\times 1$ up-left parity cross.
Finally, supertiles with a down-left cross will alternate on the diagonal with the supertiles having a up-right cross.

\medskip

Using these parity rules and the usual arrow heads/tails constraints, we shall obtain the adjacency relations for the supertiles which have to be obeyed.
We make a point here: the only constraints that we will use in the tiling of the 2D plane are the ones set by these adjacency rules.

\medskip

We shall now reconstruct the basic 3-square in Robinson's argument, this time using $2\times 2$ supertiles.
At the corners of these 3-squares there must be the four supertiles that we have identified as parity crosses.
By strictly following the adjacency rules, we will end up with exactly four possible 3-squares, that we will relate to the 3-squares with $1\times 1$ tiles.
No other configuration of a 3-square is allowed!

\vspace{1cm}

%\todo{Uncomment here}
\begin{center}
  \includegraphics[width=0.9\textwidth]{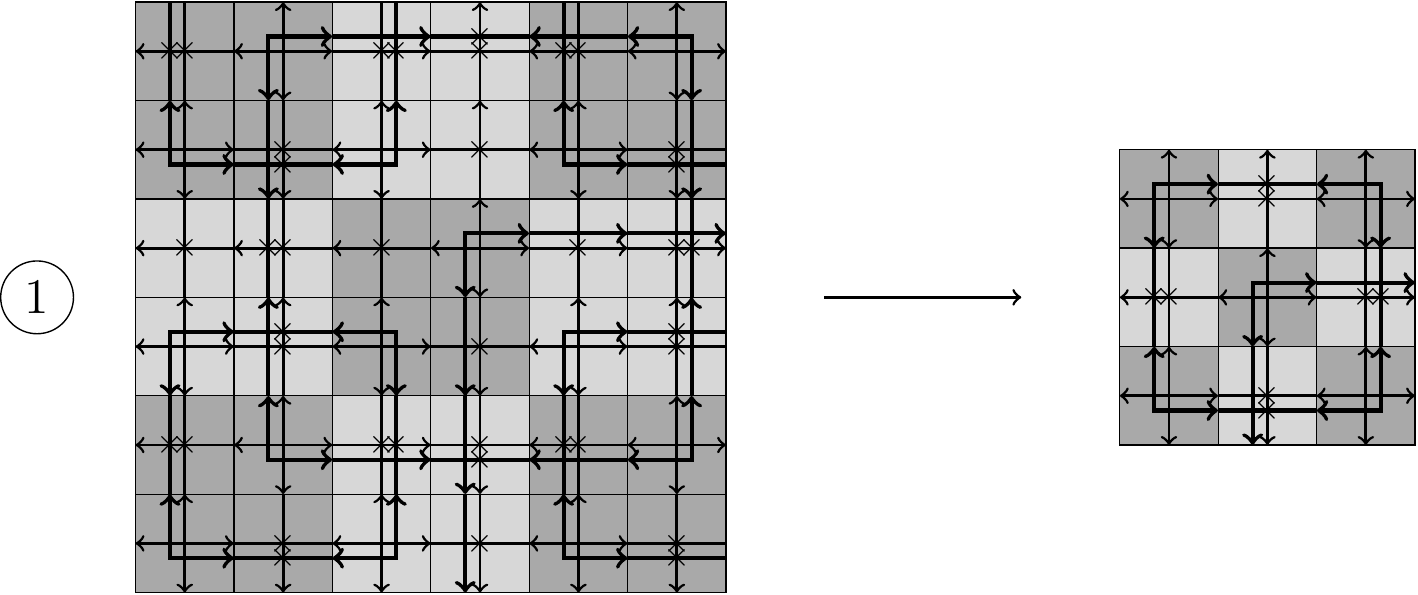}

  % ========================
  % ========================
  \vspace{2cm}
  % ========================
  % ========================

  \includegraphics[width=0.9\textwidth]{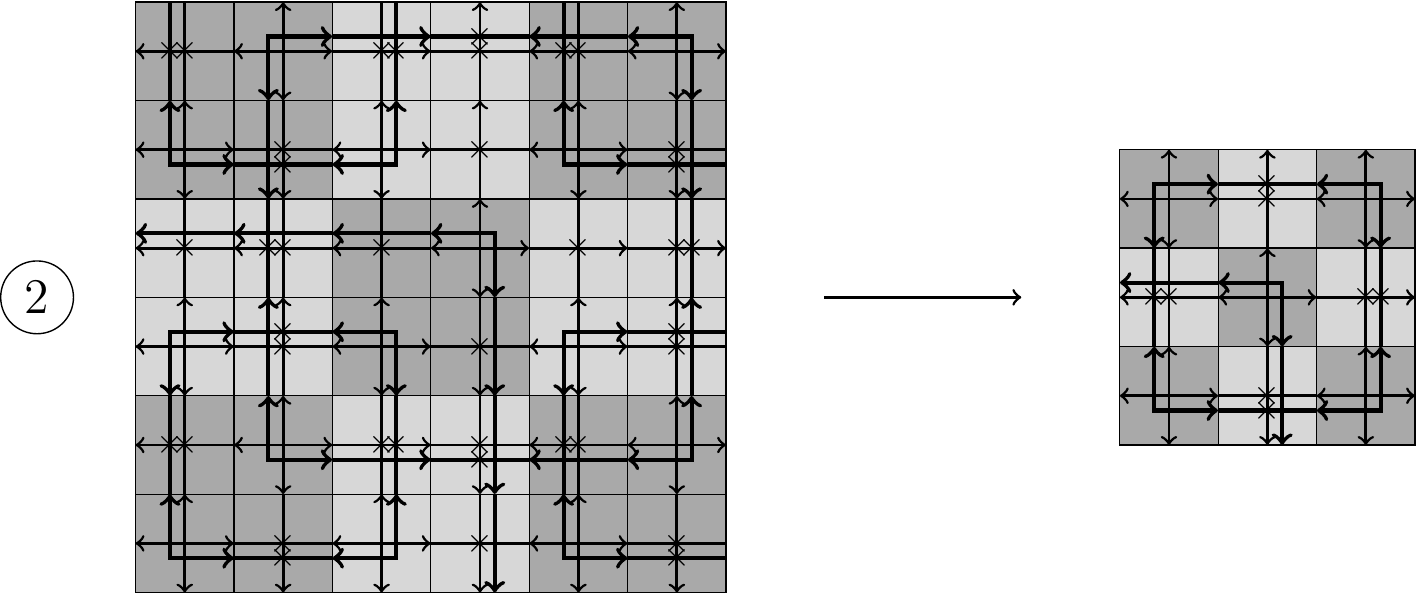}

  % ========================
  % ========================
  \vspace{2cm}
  % ========================
  % ========================

  \includegraphics[width=0.9\textwidth]{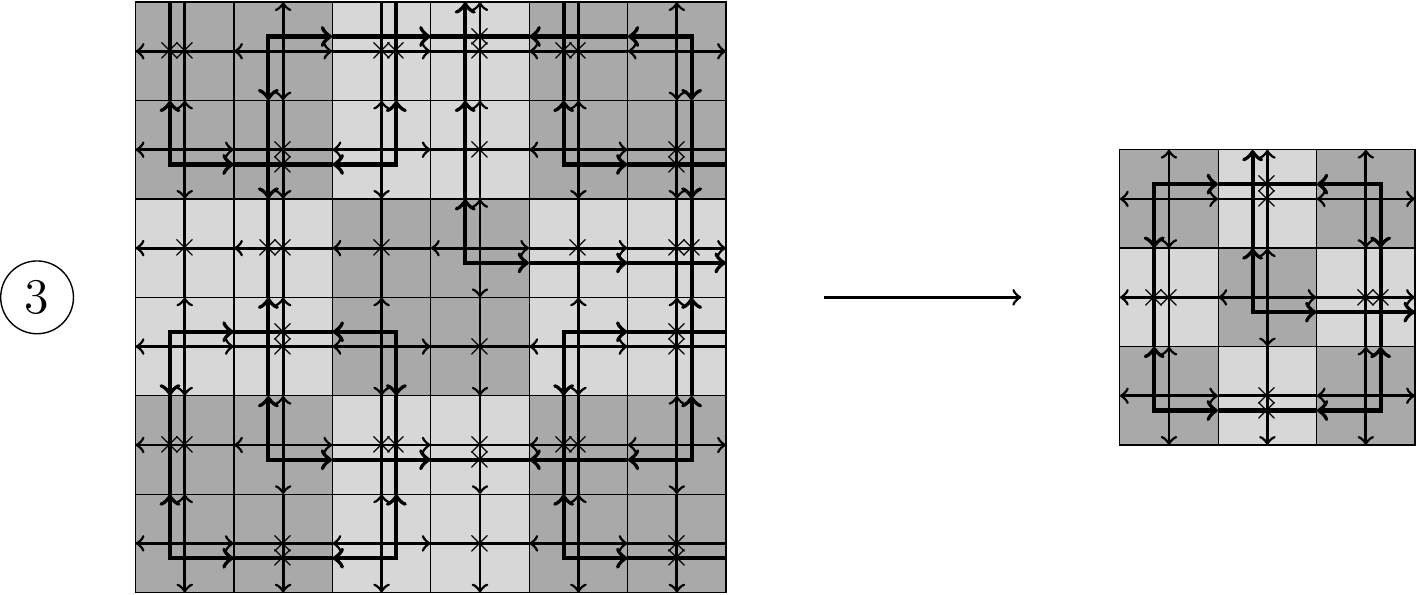}

  % ========================
  % ========================
  \vspace{2cm}
  % ========================
  % ========================

  \includegraphics[width=0.9\textwidth]{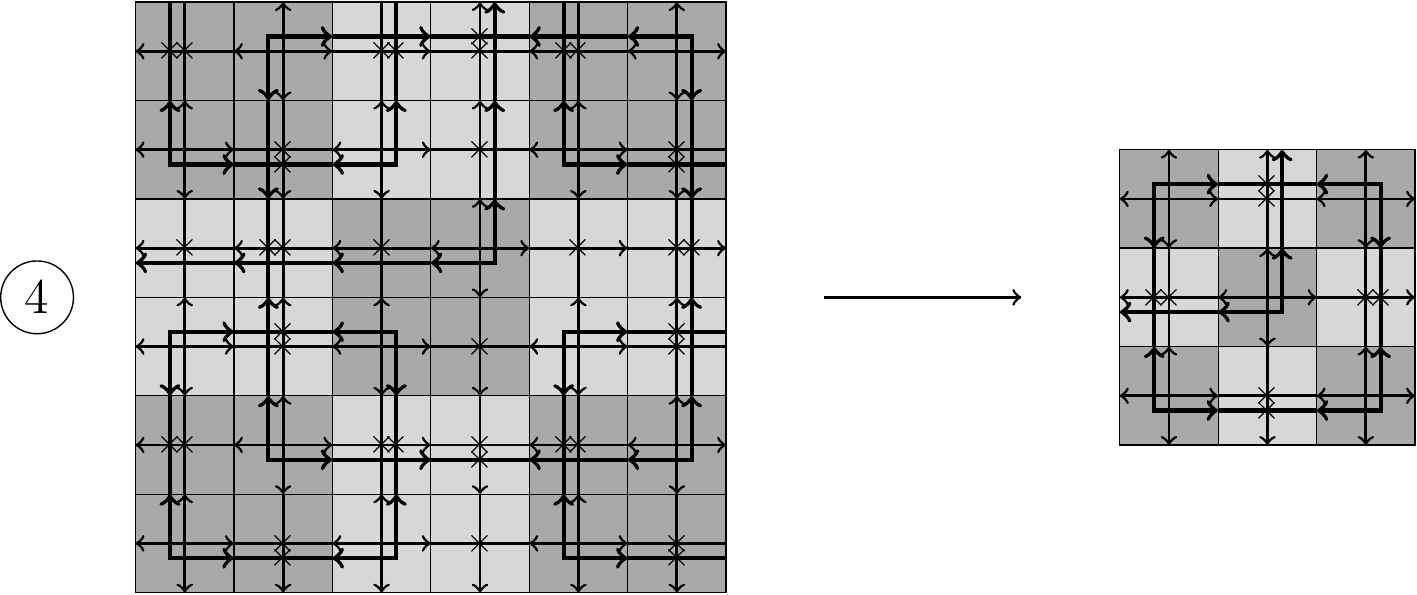}

\end{center}

\vspace{1cm}

We remark that the position of the $2\times 2$ parity crosses at the corner of the 3-square are fixed, and that the central $2\times 2$ supertile of these 3-squares -- corresponding to one of the the $1\times 1$ free crosses -- uniquely determines the remaining ones.

\medskip

At this time we have recognised the first 20 tiles.
Among these are the parity cross and free cross tiles.
The remaining ones will be determined by those supertiles placed between the 3-squares, again in analogy to the Robinson pattern.

\medskip

We consider the 3-squares in the illustration above, labelled from 1 to 4. When we pick the first 3-square, we note that only the second 3-square can be placed at its right, and they must be interleaved with a string made of three supertiles put in vertical order: only three configurations for these strings are allowed. Below we illustrate their arrow markings as well as their renormalisation onto Robinson tiles.
% with parity vertical-free-vertical, from top to bottom.
Note that the central Robinson tile in each renormalised string has free parity and is hence different from the parity vertical tile having the same arrow markings appearing in the 3-squares.
%, which has instead vertical parity.
Note also that the tile at the top of each string is the identical to the one at its bottom (both in markings and parity).

% ========================
% ========================
\vspace{1.2cm}
% ========================
% ========================

\begin{center}
  \includegraphics[width=0.75\linewidth]{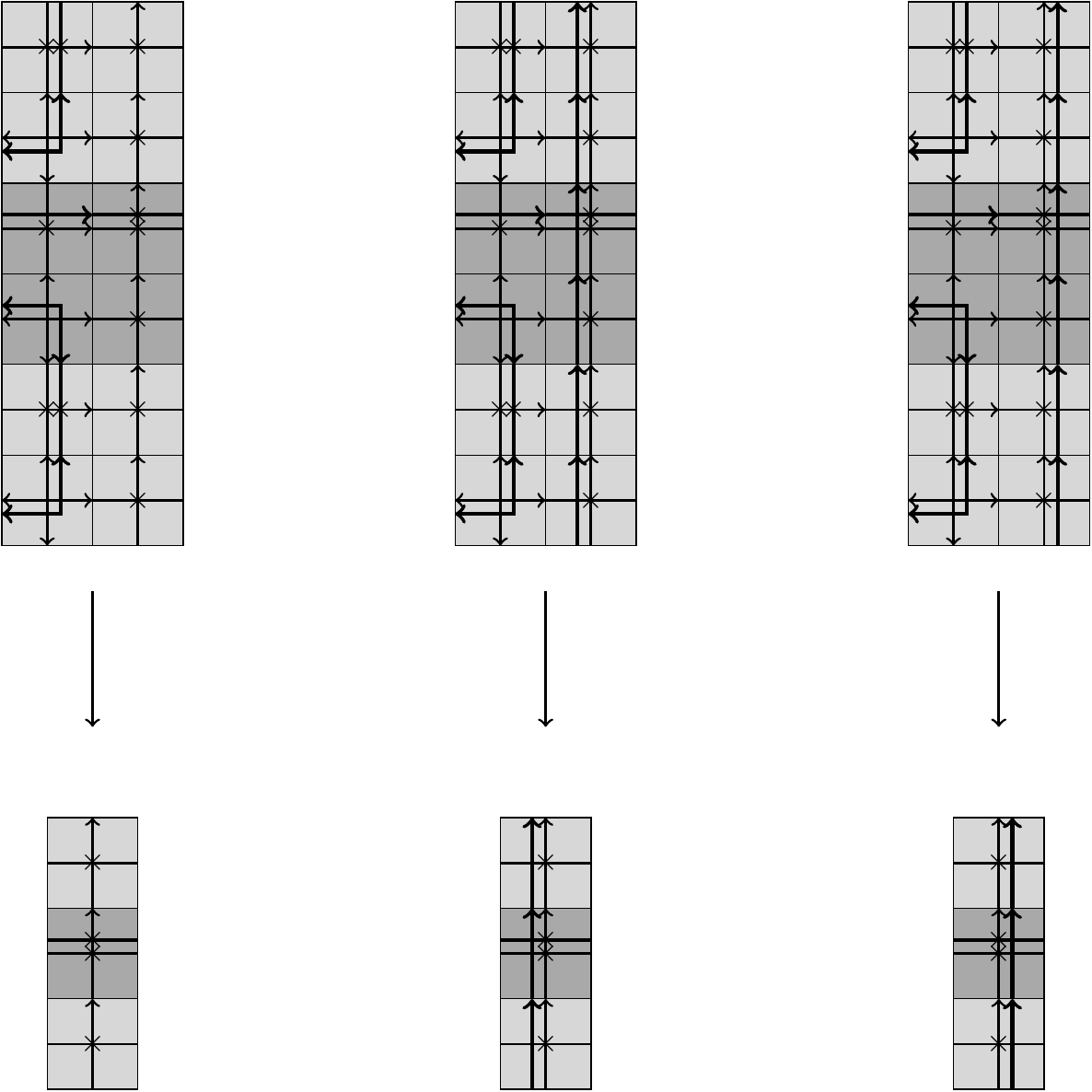}
\end{center}

% ========================
% ========================
\vspace{.3cm}
% ========================
% ========================

Analogously, starting from the third 3-square only the forth one is allowed to be placed at its right, and between them there must be one of the following strings of supertiles, that we renormalise as illustrated below.

% ========================
% ========================
\vspace{.3cm}
% ========================
% ========================
\begin{center}
  \includegraphics[width=0.75\linewidth]{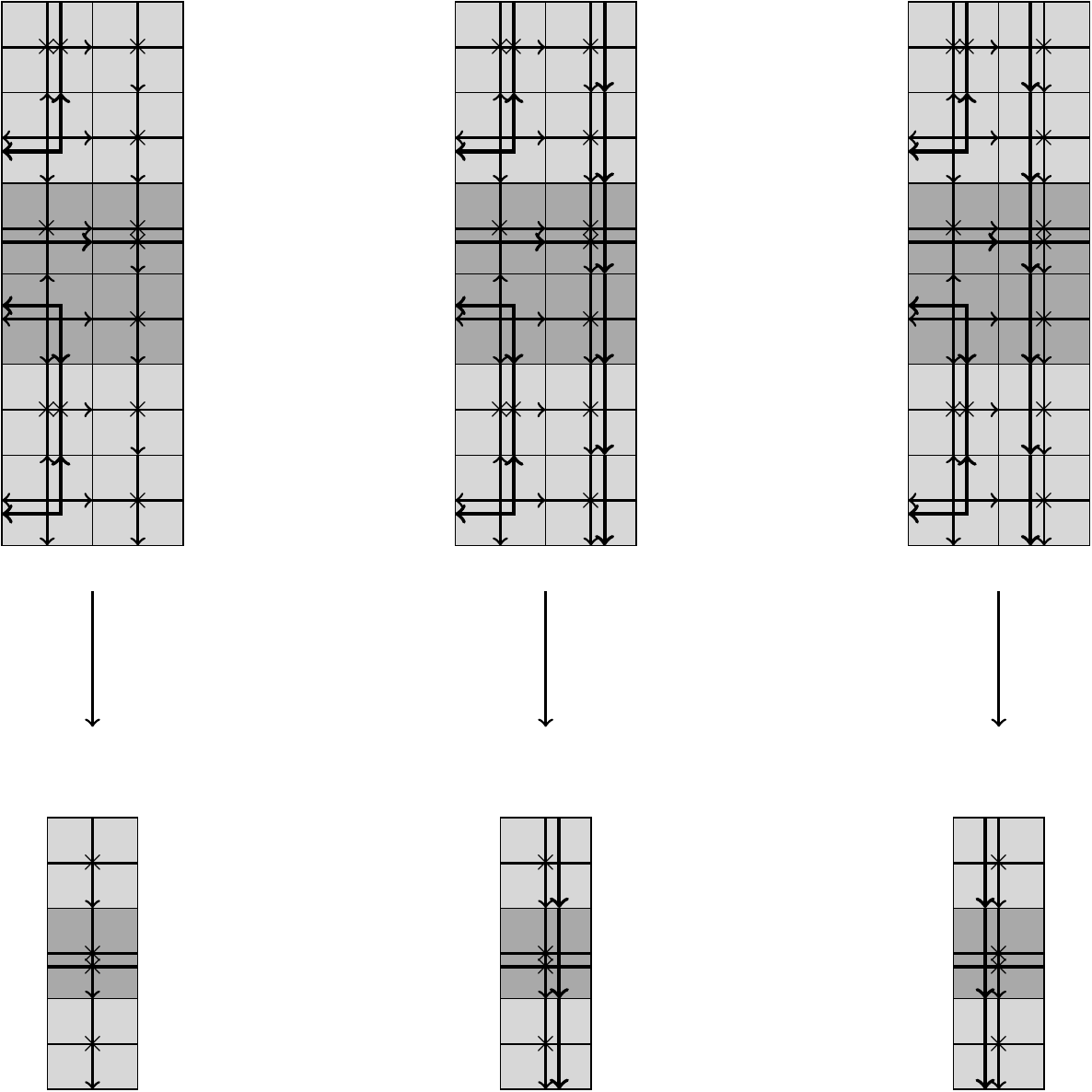}
\end{center}

% ========================
% ========================
\vspace{.6cm}
% ========================
% ========================

Vertical Robinson tiles with free parity and a single horizontal line correspond to supertiles between 3-squares whose free crosses do not face each others.
More precisely, the mapping is given by

% ========================
% ========================
\vspace{.3cm}
% ========================
% ========================

\begin{center}
  \includegraphics[width=0.9\linewidth]{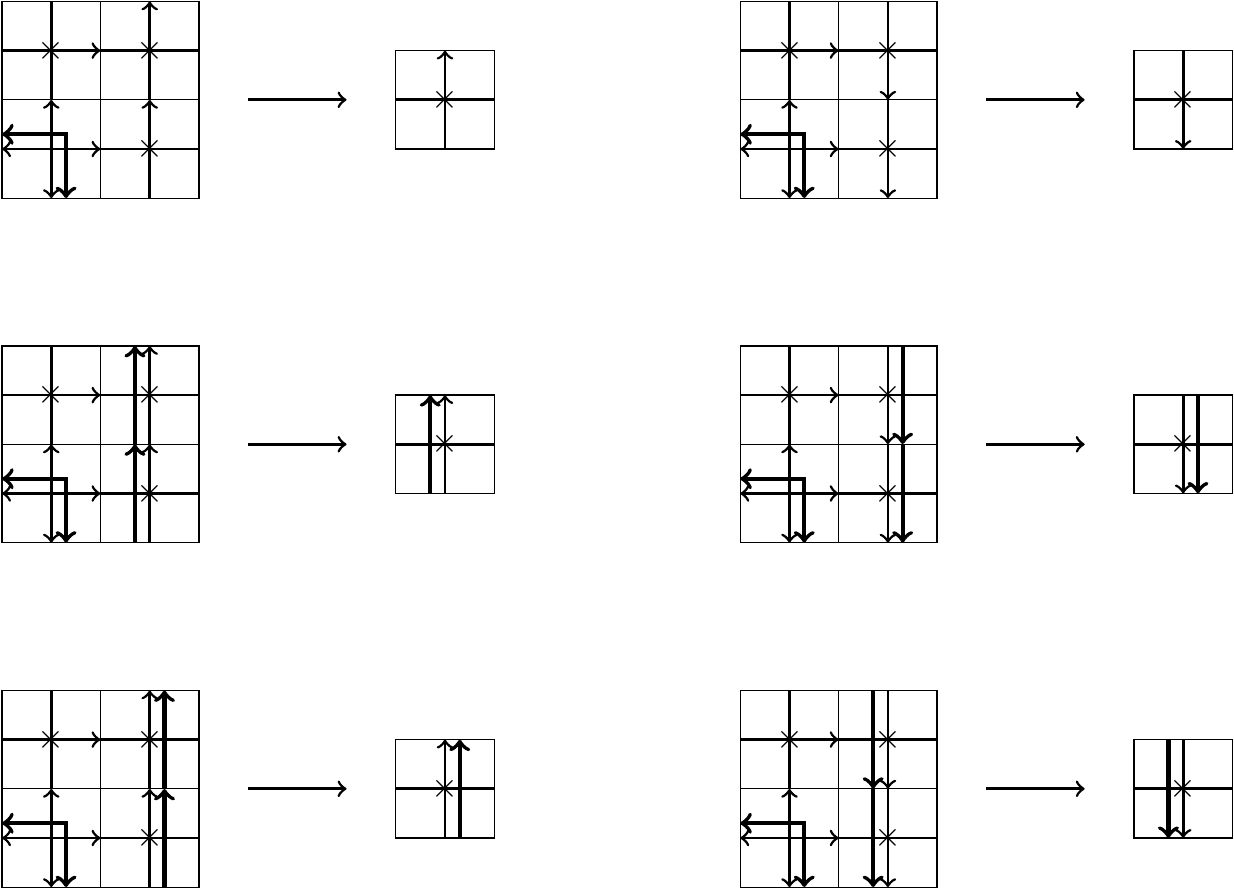}
\end{center}

\vspace{0.3cm}

We have at this point a correspondence between 38 tiles and 38 supertiles; in addition to parity crosses and free crosses, now all vertical arms associated with both parities have been identified. The remaining 18 tiles are horizontal arms.
To find them, we proceed in analogous way.

\medskip
%\newpage

Below the first 3-square we can place only the third 3-square, interleaved with one of the following strings of three supertiles put in horizontal order. Again, we note that the left and right supertiles of each string coincide and are thus mapped to the same Robinson tile and that the central tile has free parity.
% ========================
% ========================
\vspace{.3cm}
% ========================
% ========================
\begin{center}
  \includegraphics[width=0.8\linewidth]{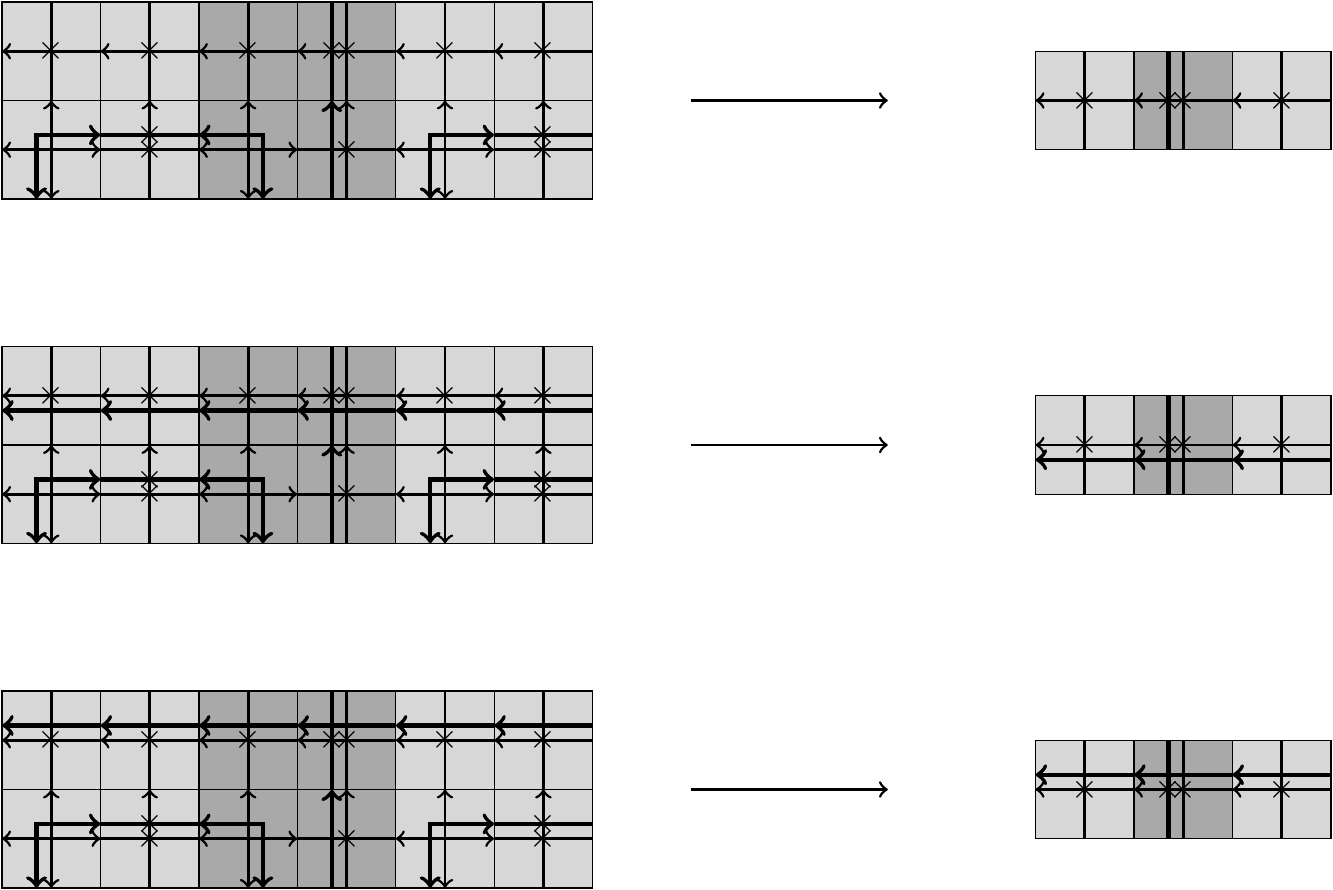}
\end{center}
% ==============================
\vspace{.6cm}
% ==============================

We have other three strings of supertiles that are allowed to stay between the second three square placed above the fourth 3-squares.
% ==============================
\vspace{.3cm}
% ==============================
\begin{center}
	  \includegraphics[width=0.8\linewidth]{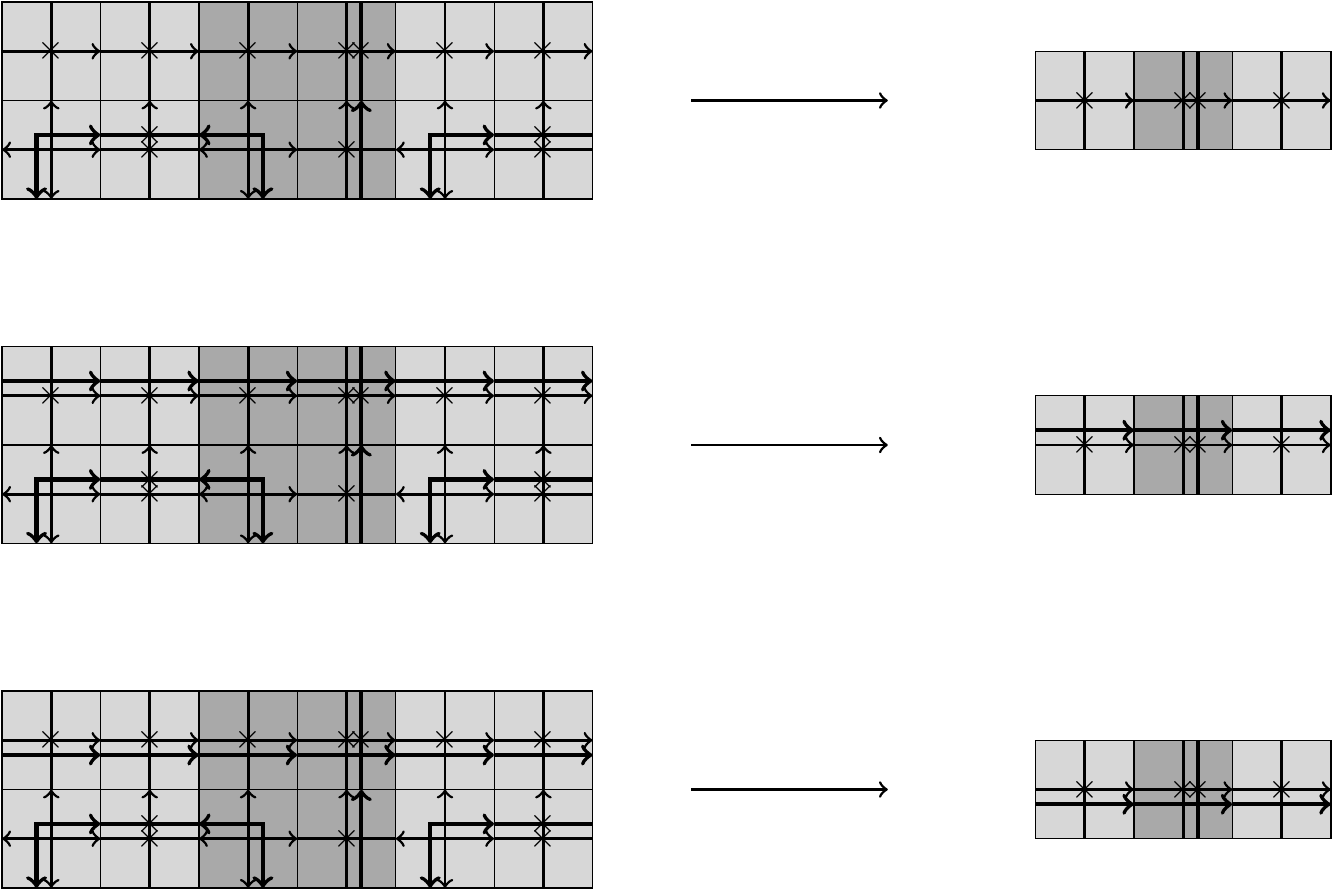}
\end{center}
% ==============================
\vspace{.6cm}
%\newpage
% ==============================

It remains to associate the last 6 free horizontal Robinson tiles to the supertiles that are still unmatched. These must be placed between 3-squares whose free crosses do not face each other.
% ==============================
\vspace{.6cm}
% ==============================
\begin{center}
	  \includegraphics[width=0.9\linewidth]{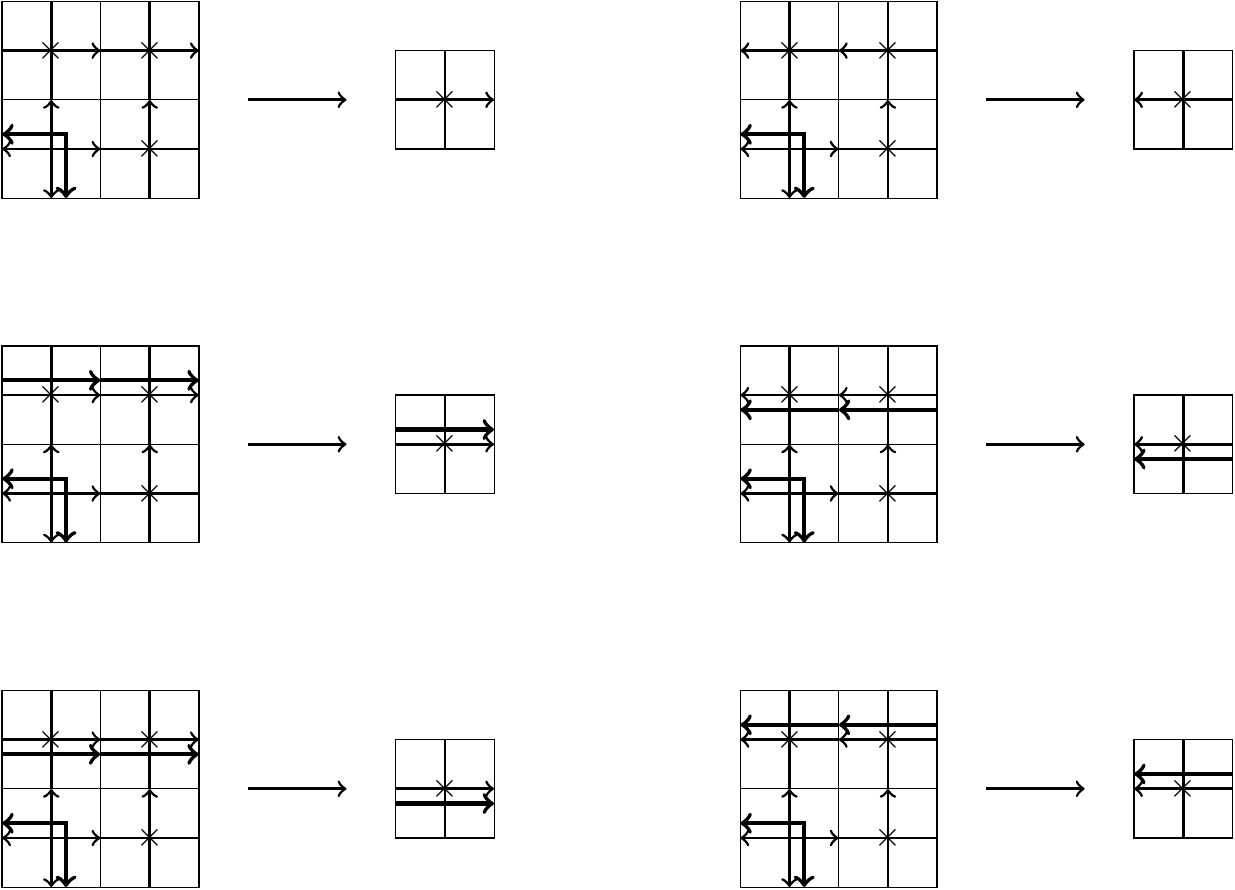}
\end{center}
% ==============================
\vspace{.6cm}
% ==============================

All 56 Robinson tiles have finally been identified with a subset of the allowed $2\times 2$ supertiles. With this renormalisation, one can verify that it is possible to reproduce the Robinson pattern of the plane with $2\times 2$ supertiles.
More importantly, one can ascertain that the adjacency rules for the $2\times 2$ supertiles, under this projection, correspond exactly to the rules for the Robinson tiles (cf.~\cref{appendix:mathematica}).
Stated in another way, the directed graphs representing respectively the adjacency rules of the 56 supertiles and the rules for the Robinson tiles are isomorphic.
Thus, we have achieved a complete renormalisation under which \cref{thm:tiles_graph_isomorphism} and \cref{cor:Robinson_pattern} hold.

\vspace{.5cm}
%\newpage

%=============================================================
\section{Mathematica notebook}\label{appendix:mathematica}
%=============================================================

Available in the arXiv submission folder is a \textit{Mathematica notebook} for the explicit construction of the renormalisation map in \cref{def:renormalization_map}, in the case when the parity cross occupies the bottom-left position of the grid.

\medskip

We begin inserting by hand the adjacency rules for the Robinson tiles, where the tiles are numbered according to the order given in~\cref{fig:Robinson_tiles}, from left to right, top to bottom. For each tile we list what are the ones that can stay above (variable \textit{adjup1x1} in the notebook) or on its right (variable \textit{adjright1x1}), respectively. Using these rules, we then construct all allowed supertiles with a parity cross in the bottom-left position; the total number of those new elements is 68. In the notebook, each supertile is represented by a $2\times 2$ matrix whose entries are numbers from 1 to 56 corresponding to the Robinson tiles which are composing it. We then construct adjacency rules for these supertiles by obeying arrow markings and parity constraints on the shared edge.
% and the inner parity structure of the supertile on the top-left of~\cref{fig:inner_parity}.

The renormalisation map is represented by the variable \emph{labelrenormalization}: the number at the position $j$ corresponds to the Robinson tile associated to the supertile $j$. The supertiles that are not appearing in the Robinson pattern discussed in~\cref{sec:not_appearing} are mapped to numbers from 57 to 68; these are not connected to any Robinson tile, and subsequently removed from the bijection. Finally, we re-write the adjacency rules for supertiles under this bijection and confirm that the graph is isomorphic to the one of the Robinson tiles.

%\includepdf[pages={-}]{tiling_notebook.pdf}

%=================================================================================
%\vfill
\section{Proof of \cref{Lemma:Full_Renormalization}} \label{Appendix:Total_RG_Proof}
%=================================================================================

For convenience we state Lemma 51 of \cite{Cubitt_PG_Wolf_Undecidability}.
\begin{lemma}[Tiling + quantum layers, Lemma 51 of \cite{Cubitt_PG_Wolf_Undecidability}]\label{Lemma:tiling+quantum}
  Let $h_c^{\mathrm{row}},h_c^{\mathrm{col}}\in\B(\C^C\ox\C^C)$ be the local interactions of a 2D tiling Hamiltonian $H_c$, with two distinguished states (tiles) $\ket{L},\ket{R}\in\C^C$. Let $h_q\in\B(\C^Q\ox\C^Q)$ be the local interaction of a Gottesman-Irani Hamiltonian $H_q(r)$, as in  \cref{Sec:Quantum_RG}.
  Then there is a Hamiltonian on a 2D square lattice with nearest-neighbour interactions $h^\mathrm{row},h^\mathrm{col}\in\B(\C^{C+Q+1}\ox\C^{C+Q+1})$ with the following properties: For any region of the lattice, the restriction of the Hamiltonian to that region has an eigenbasis of the form $\ket{T}_c\ox\ket{\psi}_q$, where $\ket{T}_c$ is a \emph{product} state representing a classical configuration of tiles. Furthermore, for any given $\ket{T}_c$, the lowest energy choice for $\ket{\psi}_q$ consists of ground states of $H_q(r)$ on segments between sites in which $\ket{T}_q$ contains an $\ket{L}$ and an $\ket{R}$, a 0-energy eigenstate on segments between an $\ket{L}$ or $\ket{R}$ and the boundary of the region, and $\ket{e}$'s everywhere else.
\end{lemma}

For the rest of this section we denote $\ket{\Rk(L)}$ and $\ket{\Rk(R)}$ to be the states in the set of $k$-time renormalised Robinson tiles with a down-left and down-right red cross marking on them, respectively.
For simplicity we break down \cref{Lemma:Full_Renormalization} into two separate parts: the first deals with the first two claims and the last deals with the third claim.

%\newpage

\begin{lemma}[Restatement of Claim 1 and 2 in~\cref{Lemma:Full_Renormalization}]
  \label{Lemma:Full_Renormalization_2}
  Let $H_u^{\Lambda(L)}=\sum h_u^{row(j,j+1)} + \sum h_u^{col(i,i+1)}$, where
  \begin{subequations}\label{TQ:overall_H_2}
    \begin{align}
      h^{\mathrm{col}}_{j,j+1} =
      \label{TQ:cols_2} &h_c^{\mathrm{col}}\ox \1_{eq}^{(j)} \ox \1_{eq}^{(j+1)}\\
      h^{\mathrm{row}}_{i,i+1} =
      \label{TQ:Hc_2}   &h_c^{\mathrm{row}}\ox\1_{eq}^{(i)}\ox\1_{eq}^{(i+1)}\\
      \label{TQ:Hq_2}   &+\1_c^{(i)}\ox\1_c^{(i+1)}\ox h_q\\
      % force < above L
      \label{TQ:<L_2}   &+\ketbra{L}^{(i)}_c \ox (\1_{eq}-\ketbra{\leftend})^{(i)}
                          \ox \1_{ceq}^{(i+1)}\\
                          % force L below <
      \label{TQ:L<_2}   &+(\1_c- \ketbra{L}_c)^{(i)} \ox \ketbra{\leftend}^{(i)}
                          \ox \1_{ceq}^{(i+1)}\\
                          % force > above R
      \label{TQ:>R_2}   &+\1_{ceq}^{(i)} \ox \ketbra{R}^{(i+1)}_c
                          \ox (\1_{eq} - \ketbra{\rightend})^{(i+1)}\\
                          % force R below >
      \label{TQ:R>_2}   &+\1_{ceq}^{(i)}
                          \ox (\1_c-\ketbra{R})^{(i+1)}_c
                          \ox\ketbra{\rightend}^{(i+1)}\\
                          % force non-blank in q-layer to left of R in c-layer
      \label{TQ:1R_2}   &+\1_c^{(i)} \ox \ketbra{e}^{(i)}_e
                          \ox \ketbra{R}^{(i+1)}_c \ox \1_{eq}^{(i+1)}\\
                          % force non-blank in q-layer to right of L in c-layer
      \label[term]{TQ:L1_2}   &+\ketbra{L}^{(i)}_c \ox \1_{eq}^{(i)}
                                \ox\1_c^{(i+1)} \ox \ketbra{e}^{(i+1)}_e\\
                                % forbid non-blank to right of blank in q-layer except over L
      \label{TQ:10-a_2} &+\1_c^{(i)} \ox \ketbra{e}^{(i)}_e
                          \ox (\1_c-\ketbra{L})^{(i+1)}_c
                          \ox (\1_{eq}-\ketbra{e})^{(i+1)}_e\\
                          % forbid non-blank to left of blank in q-layer except over R
      \label{TQ:10-b_2} &+(\1_c-\ketbra{R})^{(i)}_c
                          \ox (\1_{eq}-\ketbra{e})^{(i)}_e
                          \ox\1_c^{(i+1)} \ox \ketbra{e}^{(i+1)}_e \\
                        &+ \1_{ceq}^{(i)} \ox \1_{ceq}^{(i+1)} \label{TQ:2-Local-Constant_2}\\
      h^{(1)}_i = &-(1+\alpha_2(\varphi) )\1_{ceq}^{(i)}, \label{TQ:1-Local-Constant_2}
    \end{align}
  \end{subequations}
  for a constant $\alpha_2(\varphi)$.

  Then the $k$ times renormalised Hamiltonian under the RG mapping of \cref{Def:V^u_Isometry}, $\Rk(H_u)^{\Lambda(L\times H)}$, has the following properties:
  \begin{enumerate}
  \item For any finite region of the lattice, the restriction of the
    Hamiltonian to that region has an eigenbasis of the form  $\ket{T}_c\otimes\ket{\psi_i}$ where $\ket{T}_c\in \Rk(\HS_c)^{\Lambda(L\times H)}$ is a classical tiling state, $\ket{\psi_i}\in \Rk(\HS_{eq})^{\Lambda(L\times H)}$. \label{H_e_Renorm:Separable_Eigenstates_2}

  \item Furthermore, for any given $\ket{T}_c$, the lowest energy choice for $\ket{\psi}_q$ consists of ground states of $\Rk(H_q)(r)$ on segments between sites in which $\ket{T}_c$ contains an $\ket{\Rk(L)}$ and an $\ket{\Rk(R)}$, a 0-energy eigenstate on segments between an $\ket{\Rk(L)}$ or $\ket{\Rk(R)}$ and the boundary of the region, and $\ket{e^{\times2^k}}$’s everywhere else.
    Any eigenstate which is not an eigenstate of $\Rk(H_q)(r)$ on segments between sites in which $\ket{T}_c$ contains an $\ket{\Rk(L)}$ and an $\ket{\Rk(R)}$ has an energy $>1$. \label{H_e_Renorm:GS_Structure_2}
    % \item $\lambda_0(\r(H_u)(L/2))=\lambda_0(H_u(L))$. \label{H_e_Renorm:GS_Energy_2}

  \end{enumerate}
\end{lemma}

\newpage

\begin{proof}
  \phantom{a}
  \paragraph{Claim \ref{H_e_Renorm:Separable_Eigenstates_2}} ~\newline
  The fact the eigenstates of the unrenormalised Hamiltonian are a product state across $\HS_c$ and $\HS_{eq}$, $\ket{T_c}\ket{\psi}_{eq'}$ is from \cref{Lemma:tiling+quantum} (Lemma 51 of \cite{Cubitt_PG_Wolf_Undecidability}).
  The structure of the eigenstates of the renormalised Hamiltonian is then preserved as per \cref{Lemma:Separable_Eigenstates}.

  \paragraph{Claim \ref{H_e_Renorm:GS_Structure_2}} ~\newline
  Start by considering what each of the local terms looks like after  applying the renormalisation isometries.
  We treat each term in the above lemma in succession.
  Start with local interactions encoding the classical tiling, terms \ref{TQ:cols_2} and \ref{TQ:Hc_2}.
  The isometry decomposes as $V^u_{(i,i+1)(j,j+1)}=(\\1 \ox \Pi_{gs})V^c_{(i,i+1)(j,j+1)}\ox V^{eq}_{(i,i+1)(j,j+1)}$, hence the classical Hamiltonian terms transform as per~\cref{Lemma:Classical_RG_Hamiltonian}.

  We next consider the renormalisation of the Gottesman-Irani Hamiltonian $h_q$. %: in particular we note that we are restricting to the subspace $k_{i,j}$ we are restricting to a subspace of $\HS^{(k-1)}_{eq}\otimes \HS^{(k-1)}_{eq}$.
  All of these states are mapped by $V^{eq}_{(i,i+1)(j,j+1)}$ to a $2\times 1$ chain, which $V^q_{(i,i+1)}$ acts on as per \cref{Lemma:RG_GI_Properties}.
  Thus $h_q$ transforms as per \cref{Lemma:RG_GI_Properties}.

  \paragraph{Coupling Terms}	~\newline
  We first note that given a $2\times 2 $ block, we wil get two sets of coupling terms: one between $c$ and $eq_1$ and another set between $c$ and $eq_2$.
  Thus the terms will have the structure $h_a^{i,i+1}\otimes h_{eq_1}^{i,i+1}\otimes h_{eq_2}^{i,i+1}$, where $h_{eq_1}^{i,i+1}$ and  $h_{eq_2}^{i,i+1}$ are identical except they act on different parts of the local Hilbert space.

  We will then ``integrate out'' $eq_2$ in the next stage of the renormalisation procedure leaving us with only a single set.
  Thus for the purposes of the RG procedure, we need only consider how the coupling terms transform for a particular $(i,i+1; j)$ set (as we will integrate out the other set anyways).
  ~\newline

  We now consider the terms coupling the classical and quantum parts of the Hilbert space.
  Consider term \ref{TQ:<L_2}.
  In any $2\times 2$ block in the restricted subspace, at most one \textbf{free} $\ket{L}$ or $\ket{R}$ may appear (i.e. not parity cross), and under the classical renormalisation mapping, we see that a $2\times 2$ block with a free cross is mapped to a cross supertile of the same colour and with relevant orientation.
  Any parity cross is removed in the renormalisation step, as per \cref{Sec:Tiling_Renormalization}.
  Then we realise that the $2\times 2$ block only receives the penalty iff $\ket{L}$ is not combined with $\ket{\leftend}$.
  Since under the RG operations $\ket{L}\rightarrow \ket{\r(L)}$, and $\ket{\leftend}\ket{x}\rightarrow \ket{\r(\leftend, x)}$
  we see that the new term must penalise states which do not satisfy these states being paired.
  % the renormalised equivalents of $\ket{L}$ which we denote $\ket{\r(L)}$, $\r(L)\in T_2'$.
  The parity $\ket{L}$ tiles will be integrated out, however, these are associated with history states that will be integrated out in the same step, and hence can be ignored.
  Thus term \ref{TQ:<L_2} becomes
  \begin{align}
    \ketbra{\r(L)}^{(i)}\ox \left(\1_{eq'} - \ketbra{\leftend x} - \ketbra{e\leftend} \right) \ox \1_{ceq'}^{(i+1)},
  \end{align}
  where $\ket{x}\in \frk{B}$ are single site states of the original Hamiltonian.
  By similar reasoning, after $k$ applications of the RG mapping, we get
  \begin{adjustwidth}{-2cm}{-2cm}
  \begin{align}
    \ketbra{\Rk(L)}^{(i)}\ox \left(\1_{eq'} - \sum_m\sum_{x_t\in \frk{B}} \ketbra{e^{\times m}\leftend \{x_t\}^{\times 2^k-m-1}}  \right) \ox \1_{ceq'}^{(i+1)}.
  \end{align}
  \end{adjustwidth}
  The term \ref{TQ:>R_2} transforms analogously.

  Now consider term \ref{TQ:L<_2}.
  Again, $2\times 2$ blocks in the restricted subspace with the free tile being $\ket{L}$ get renormalised to $\ket{\r(L)}$.
  We see that this term penalises anything but $\ket{\leftend}$ being combined with it, and hence we see it is mapped to
  \begin{align}
    (\1_c- \ketbra{\r(L)}_c)^{(i)} \ox \left(\ketbra{\leftend x}^{(i)} + \ketbra{e \leftend}^{(i)}\right) \ox \1_{ceq}^{(i+1)}.
  \end{align}
  By similar reasoning, after $k$ iterations we get
  \begin{adjustwidth}{-2cm}{-2cm}
  \begin{align}
    (&\1_c- \ketbra{\Rk(L)}_c)^{(i)} \ox \left( \sum_m\sum_{x_t\in \frk{B}} \ketbra{e^{\times m}\leftend \{x_t\}^{\times 2^k-m}} \right)^{(i)} \ox \1_{ceq}^{(i+1)}.
  \end{align}
  \end{adjustwidth}
  The \ref{TQ:R>_2} transforms analogously.

  We now consider term \ref{TQ:L1_2}.
  If we consider the term acting between $2\times 2$ blocks, then this is only violated if there is a $\ket{L}_c$ at site $(i,j)$ and at the neighbouring site $(i+1,j)$ is in state $\ket{e}_e$.
  The renormalised basis states which get penalised by this are then:
  \begin{align}
    \ketbra{\r(L)}^{(i)}_c \ox \1_{eq}^{(i)}
    \ox\1_c^{(i+1)} \ox \left(\ketbra{ee}_e+\sum_{\ket{x}\in \frk{B}}\ketbra{ex}_{q'} \right)^{(i+1)}.
  \end{align}
  After $k$ iterations this becomes
  \begin{adjustwidth}{-2cm}{-2cm}
  \begin{align}
    \ketbra{\Rk(L)}^{(i)}_c \ox \1_{eq}^{(i)}
    \ox\1_c^{(i+1)} \ox \left( \sum_m\sum_{x_t\in \frk{B}} \ketbra{e^{\times m}\{x\}^{\times 2^k-m}} \right)^{(i+1)}.
  \end{align}
  \end{adjustwidth}
  Term \ref{TQ:1R_2} transforms analogously.

  \medskip

  We now consider term \ref{TQ:10-a_2}.
  This term forces a non-$\ket{e}_e$ to the left of any other non-blank in the $q$-layer, except when a non-blank coincides
  with an $\ket{L}$ in the c-layer.
  Again, we see that this penalty term is zero within any $2\times 2$ blocks in the restricted subspace $\kappa_{i,j}$, so we need only consider the interactions between such states.
  If there is a $\ket{e}_e$ state next to a $\ket{x}$ state in the blocks, then we see that the quantum part of this tile must get mapped to $\ket{e}_e$ or $\ket{\r(x)}$.
  %Thus we penalise such configurations
  The new term in the Hamiltonian becomes
  \begin{adjustwidth}{-2cm}{-2cm}
  \begin{align}
    \1_c^{(i)} \ox \left( \ketbra{ee}_e + \sum_{\ket{y}\in \frk{B}}\ketbra{ye}_{q} \right)^{(i)} \ox (\1_c-\ketbra{\r(L)})^{(i+1)}_c \ox (\1_{eq}-\ketbra{ee})^{(i+1)}_e.
  \end{align}
  \end{adjustwidth}
  After $k$ iterations of the RG map the term becomes
  \begin{adjustwidth}{-2cm}{-2cm}
  \begin{align}
    \1_c^{(i)} &\ox \left( \sum_{m=1}\sum_{x_t\in \frk{B}}\ketbra{\{x_t\}^{\times 2^k-m}, e^{\times m}} \right)^{(i)} \ox (\1_c-\ketbra{\r(L)})^{(i+1)}_c \\ &\ox (\1_{eq}-\sum_{m=1}\sum_{x_t\in \frk{B}} \ketbra{e^{\times m}, \leftend ,\{x_t\}^{\times 2^k-m-1}} )^{(i+1)}_e.
  \end{align}
  \end{adjustwidth}
  Term \ref{TQ:10-b_2} transforms analogously.

  \medskip

  \paragraph{Identity Terms}~\newline
  Finally we need to consider how terms of the form $\1_{ceq}^{(i)}$ and $\1_{ceq}^{(i)}\ox \1_{ceq}^{(i+1)}$ transform; as per \cref{Remark:Local_Projectors}
  these terms appear as the Hamiltonian is iterated.
  consider the two local terms $\1^{(i,j)}\ox \1^{(i+1,j)}$:
  \begin{align}
    \left( \1^{(i,j)}\ox \1^{(i+1,j)} + \1^{(i,j+1)}\ox \1^{(i+1,j+1)}\right)
    &\rightarrow 2\1^{(i/2,j/2)}.
  \end{align}
  Similarly, consider
  \begin{align}
    \left(\1^{(i+1,j)}\ox \1^{(i+2,j)} + \1^{(i+1,j+1)}\ox \1^{(i+2,j+1)} \right) &\rightarrow 2\1^{(i/2,j/2)}\ox \1^{(i/2+1,j/2)}.
  \end{align}
  Consider the $\1^{(i,j)}$ terms, then
  \begin{align} \label{Eq:Identity_Term_Transform}
    \1^{(i,j)} + \1^{(i+1,j)}+ \1^{(i,j+1)} +\1^{(i+1,j+1)}  &\rightarrow  	4\1^{(i/2,j/2)}.
  \end{align}
  Combining these terms, we see that these create new 1-local terms which, after $k$ iterations have coefficients:
  \begin{align}
    (-4^k+\sum_{m=0}^k (4^m\times 2^{m-k}))&\1^{(i/2+1,j/2)} =  -2^{-k}\1^{(i/2+1,j/2)},
  \end{align}
  and 2-local terms of the form:
  \begin{align}
    2^k&\1^{(i/2,j/2)}\ox \1^{(i/2+1,j/2)}.
  \end{align}
  Note that these 2-local terms only occur in the row interactions, and remain zero for the column interactions.

    ~\newline
  So far we have shown that all terms in the Hamiltonian transform to an analogous term to one in the original Hamiltonian.
  Now note the fact the Hamiltonian can be block-decomposed into subspaces with respect to states containing $\leftend$ and $\rightend$, and into a classical and quantum part.
  Then realise that the local quantum Hilbert space can be decomposed as $\Rk(\HS_e)\oplus \Rk(\HS_{q})$.
  These properties allow the proof from Lemma 51 of \cite{Cubitt_PG_Wolf_Undecidability} to be applied (we refer the reader to this proof for brevity) which also shows that states which are not $\Rk(H_q)$ eigenstates between $\ket{\Rk(L)}$ and $\ket{\Rk(R)}$ markers have energy at least 1.

\end{proof}

With this, we now wish to prove claim 3 of \cref{Lemma:Full_Renormalization} and hence need to find the ground state energy for the renormalised Hamiltonian.
To do so we need the concept of \emph{tiling defects}:
\begin{definition}[Tiling Defect]\label{Def:Tiling_Defect}
A pair $\ket{t_a}_{i,j},\ket{t_b}_{i+1,j}\in \HS_c$ form a tiling defect if they violate the local term between them: $\bra{t_a}\bra{t_b}h_c^{i,i+1}\ket{t_a}\ket{t_b}=1$.
Similarly, $\ket{t_a}_{i,j},\ket{t_b}_{i+1,j}\in \Rk(\HS_c)$ form a tiling defect if they violate the renormalised local term between them: $\bra{t_a}\bra{t_b}\Rk(h_c)^{(i,i+1)}\ket{t_a}\ket{t_b}=1$.
\end{definition}

In the following lemma we show the ground state is a state with no tiling defects, and as a result the only energy contribution comes from ground states of the Gottesman-Irani Hamiltonians.

\begin{lemma}[Restatement of Claim 3 in~\cref{Lemma:Full_Renormalization}]
  \label{Lemma:put-promise-together} \hfill\newline
  Let $h_c^{\mathrm{row}},h_c^{\mathrm{col}}\in\cB(\C^C\ox\C^C)$ be the local interactions of the tiling Hamiltonian associated with the modified Robinson tiles, let $\Rk(h_c^{row})^{i,i+1},\Rk(h_c^{col})^{j,j+1}$ be the local interactions after $k$ RG iterations, and let $h^\mathrm{row},h^\mathrm{col}\in\cB(\C^{C+Q+1}\ox\C^{C+Q+1})$ be the local interactions defined in \cref{Lemma:Full_Renormalization_2}.
  For a given ground state configuration (tiling) of $\Rk(H_c)$, let $\mathcal{L}$ denote the set of all horizontal line segments of the lattice that lie between down/right-facing and down/left-facing red crosses (inclusive) in the Robinson tiling after $k$ RG mappings.

  Then the renormalised Hamiltonian on a 2D square lattice of width $L$ and height $H$ with nearest-neighbour interactions $\Rk(h^\mathrm{row}),\Rk(h^\mathrm{col})$ has a ground state energy $\lambda_0(\Rk(H)^{\Lambda(L\times H)})$ contained in the interval
%  \begin{align}\label{eq:interval-0-defects_renormalised}
%    \bigg[ &(g(k) - 4^k\alpha_2(\varphi)) LH - 2^{k}H  + \sum_{n=1}^{\lfloor\log_4(L/2)\rfloor}
%             \left(
%             \left\lfloor\frac{H}{2^{2n+1-(k\ mod 2)}}\right\rfloor
%             \left(\left\lfloor\frac{L}{2^{2n+1-(k\ mod 2)}}\right\rfloor -1\right)
%             \right)
%             \lambda_0(\Rk(H_{q})(4^{n-\lfloor(k\ mod 2)/2\rfloor})) , \\
%           &(g(k) -  4^k\alpha_2(\varphi)) LH  - 2^{k}H  +  \sum_{n=1}^{\lfloor\log_4(L/2)\rfloor}
%             \left(\left(
%             \left\lfloor\frac{H}{2^{2n+1-(k \ mod 2)}}\right\rfloor +1\right)
%             \left\lfloor\frac{L}{2^{2n+1-(k\ mod 2)}}\right\rfloor
%             \right)
%             \lambda_0(\Rk(H_{q})(4^{n-\lfloor(k\ mod 2)/2\rfloor})) \bigg]
%  \end{align}
      \begin{align}
  \bigg[ &(g(k) - 4^k\alpha_2(\varphi)) LH - 2^{-k}H
  + \sum_{n=1}^{\lfloor\log_4(L/2)\rfloor}
  \bigg(
  \left\lfloor\frac{H}{2^{2n+1 (k \ mod 2)}}\right\rfloor \\
  &\times\left(\left\lfloor\frac{L}{2^{2n+1 - (k \ mod 2)}}\right\rfloor -1\right)
  \bigg)
  \lambda_0(\Rk(H_{q})(4^{n- \lfloor (k \ mod 2)/2\rfloor})) , \\
  &(g(k) -  4^k\alpha_2(\varphi)) LH - 2^{-k}H
  +  \sum_{n=1}^{\lfloor\log_4(L/2)\rfloor}
  \bigg(\left(
  \left\lfloor\frac{H}{2^{2n+1- (k \ mod 2)}}\right\rfloor +1\right) \\
  &\times \left\lfloor\frac{L}{2^{2n+1 - (k \ mod 2)}}\right\rfloor
  \bigg)
  \lambda_0(\Rk(H_{q})(4^{n - \lfloor (k \ mod 2)/2\rfloor})) \bigg]
  \end{align}
  where
  \begin{align} \label{Eq:g(k)_Definition}
    g(k) = 4^k\sum_{4^n+1<2^k}4^{-2n-1}\lambda_0(H_q(4^n)).
  \end{align}
\end{lemma}
\begin{proof}
  We identify the red down-left and down-right cross tiles from the $k$-times renormalised tile set with the $\ket{\Rk(L)}$ and $\ket{\Rk(R)}$ state respectively.
  For convenience, assume $k\in 2\N$ (we will deal with the other case separately $k\in 2\N+1$).
  From \cref{Lemma:Full_Renormalization_2} the ground state of the Hamiltonian is a product state $\ket{T}_c\ox \ket{\psi_0}_{eq}$ has a $\ket{e^{\times 2^k}}$ state combined with every tile except those between $\ket{\Rk(L)}$ and $\ket{\Rk(R)}$, where instead there is a ground state of a $\Rk(H_{q})$ Hamiltonian between the two markers.
  For such states, the terms \ref{TQ:<L_2}-\ref{TQ:10-b_2} give zero energy contribution and we need only consider the terms \ref{TQ:cols_2}, \ref{TQ:Hc_2}, and \ref{TQ:Hq_2}.
  The terms \ref{TQ:2-Local-Constant_2} and \ref{TQ:1-Local-Constant_2} are constant offsets, and so we will ignore them initially and consider them at the end.

  We now consider the energy of the tiling + quantum; from lemma 48 of \cite{Cubitt_PG_Wolf_Undecidability} the number of segments is lower bounded by $\geq \lfloor H2^{-2n-1} \rfloor(\lfloor L2^{-2n-1} -1 \rfloor)$ and upper bounded by $\leq \lfloor H2^{-2n-1} +1\rfloor(\lfloor L2^{-2n-1}  \rfloor)$.

  In the case we have $d$ defects in the tiling, the energy is at least
  \begin{align}
    E(d \text{ defects})
    &= d + LH(g(k)  - 4^k\alpha_2(\varphi) )   + \sum_{\ell\in\mathcal{L}} \lambda_0(\Rk(H_{q})(\abs{\ell})) \\
    &\geq d + LH(g(k)  - 4^k\alpha_2(\varphi))   \\
    &+ \sum_{n=1}^{\lfloor\log_4(L/2)\rfloor}
      \Biggl(
      \left\lfloor\frac{H}{2^{2n+1}}\right\rfloor
      \left(\left\lfloor\frac{L}{2^{2n+1}}\right\rfloor -1\right)
      - 2d
      \Biggr) \lambda_0(\Rk(H_{q})(4^n)),
  \end{align}
  where in the second line we have used the result from lemma 49 of \cite{Cubitt_PG_Wolf_Undecidability} to bound the number of segments of size $2^{2n}$ is at least $\left\lfloor\frac{H}{2^{2n+1}}\right\rfloor\left(\left\lfloor\frac{L}{2^{2n+1}}\right\rfloor -1\right) - 2d$.
  Note, that  lemma 49 of \cite{Cubitt_PG_Wolf_Undecidability} still applies to the renormalised Hamiltonian terms as the tiling rules for the renormalised tile set are identical to the original tile set, as per \cref{Lemma:Classical_RG_Hamiltonian}.

  It can be shown from definition 50 of \cite{Cubitt_PG_Wolf_Undecidability} that $\sum_{n=1}^\infty\lambda_0(H_q(4^n+1)) < 1/2$, and since each defect carries an energy penalty of at least $1$ we see the ground state is always achieved in the case where there are no defects and hence the Robinson tiling is correct.
  Thus we see that the ground state is given by
  \begin{align}
    E = LH  (g(k) -   4^k\alpha_2(\varphi))  + \sum_{\ell\in\mathcal{L}} \lambda_0(\Rk(H_{q})(\abs{\ell})).
  \end{align}
  Again we use the bound on the number of segments allowed from lemma 48 of \cite{Cubitt_PG_Wolf_Undecidability} to show that the ground state energy lies in the bounds
  \begin{adjustwidth}{-2cm}{-2cm}
	\begin{align}
		\sum_{\ell\in\mathcal{L}} \lambda_0(\Rk(H_{q})(\abs{\ell})) \in  \bigg[ &\sum_{n=1}^{\lfloor\log_4(L/2)\rfloor}
		\left(
		\left\lfloor\frac{H}{2^{2n+1}}\right\rfloor
		\left(\left\lfloor\frac{L}{2^{2n+1}}\right\rfloor -1\right)
		\right)
		\lambda_0(\r^{(k)}(H_{q})(4^n)) , \\
		&\sum_{n=1}^{\lfloor\log_4(L/2)\rfloor}
		\left(\left(
		\left\lfloor\frac{H}{2^{2n+1}}\right\rfloor +1\right)
		\left\lfloor\frac{L}{2^{2n+1}}\right\rfloor
		\right)
		\lambda_0(\r^{(k)}(H_{q})(4^n)) \bigg]
	\end{align}
	\end{adjustwidth}
  Finally consider the constant energy offset from the terms \ref{TQ:2-Local-Constant_2} and \ref{TQ:1-Local-Constant_2}.
  After $k$ iterations of the RG mapping, from the definition of $g(k)$ in \cref{Eq:g(k)_Definition}, the coefficient of the $\1^{(i)}$ term is
  \begin{align}
    b_1 &:= 4^k\sum_{4^n+1<2^k}4^{-2n-1}\lambda_0(H_q(4^n)) + 4^k(1-\alpha_2(\varphi)) - 4^{k}\sum_{m=1}^{k}2^{-m}  \\
        &= 4^k\sum_{4^n+1<2^k}4^{-2n-1}\lambda_0(H_q(4^n)) +  4^k(1-\alpha_2(\varphi)) - 4^{k}(1-2^{-k}),
  \end{align}
  where the $- 4^{k}\sum_{m=1}^{k}2^{-m}$ term arises due to part of the 2-local terms being integrated into the 1-local terms.
  The coefficient in front of the 2-local term $ \1^{(i)}\ox\1^{(i+1)}$ is then $b_2:= -2^k$.
  The energy contribution from these term is
  \begin{adjustwidth}{-2.cm}{-2.5cm}
  \begin{align}
    &b_1 LH + b_2 (L-1)H = (b_1  + b_2) LH - b_2H \\
                        &=  \left(4^k\sum_{4^n+1<2^k}4^{-2n-1}\lambda_0(H_q(4^n))+ 4^k(1-\alpha_2(\varphi)) - 4^{k}(1-2^{-k}) -2^{-k} \right)LH
                        + b_2H \\
                        &=  \left(4^k\sum_{4^n+1<2^k}4^{-2n-1}\lambda_0(H_q(4^n)) -4^{k}\alpha_2(\varphi) \right)LH - 2^{k}H \\
                        &= (g(k)-4^{k}\alpha_2(\varphi))LH - 2^{k}H,
  \end{align}
  \end{adjustwidth}
  where $g(k)$ is defined in the lemma statement.
  Adding this to the energy contribution from the renormalised Gottesman-Irani segments gives the value in the lemma statement.

  For $k\in 2\N+1$ all of the above goes through with
  \begin{align}
  L/2^{2n+1}&\rightarrow L/2^{2n+1 - (k\ mod 2)},\\
  H/2^{2n+1}&\rightarrow H/2^{2n+1 - (k\ mod 2)},\\
  \lambda_0(H_q(4^n)) &\rightarrow \lambda_0(H_q(4^{n - \lfloor(k\ mod 2)/2\rfloor})).
  \end{align}
  This accounts for distances being reduced by a factor of two in alternate RG steps.

\end{proof}

\end{document}